\documentclass[a4paper,12pt,onehalfspacing,headrules,tools]{report}
\usepackage[utf8]{inputenc}
\usepackage[T1]{fontenc}
\usepackage[french]{babel}
\usepackage{color}
\usepackage{amssymb}
\usepackage{amsmath}
\usepackage{epsfig}
\usepackage{graphicx}
\usepackage[ruled,vlined,french,titlenumbered,linesnumbered]{algorithm2e}
\usepackage{fancyhdr}
\usepackage{longtable}
\usepackage[francais]{minitoc} 
\usepackage{rotating} 
\usepackage{siunitx}
\usepackage{subcaption}
\usepackage{float}
\usepackage{lscape}
\usepackage{minitoc}

\usepackage{array}
\usepackage{multirow}
 \date{}
\setlongtables

\newtheorem{proof}{Preuve}

\newtheorem{definition}{Définition}

\newtheorem{example}{Exemple}

\newtheorem{proposition}{Proposition}

\renewcommand{\chaptermark}[1]{\markboth{#1}{}}

\fancypagestyle{plain}{\fancyhead{}
\cfoot{\thepage}
 }

\begin{document}
\begin{titlepage}
\center\footnotesize RÉPUBLIQUE TUNISIENNE\\
\vspace{0.1em}\footnotesize MINISTÈRE DE L'ENSEIGNEMENT SUPÉRIEUR ET DE LA RECHERCHE SCIENTIFIQUE\\
\vspace{0.1em}\footnotesize UNIVERSITÉ DE TUNIS EL MANAR\\
\vspace{0.1em}\footnotesize FACULTÉ DES SCIENCES DE TUNIS\vspace{0.1em}

\begin{figure}[!ht]
\centering
\includegraphics[width=0.2\textwidth]{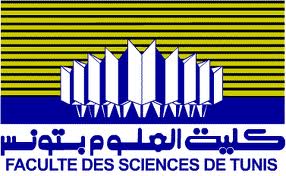}
\end{figure}

\vspace{2em}
{\bfseries\huge Thèse de Doctorat} \\
\vspace{0.25em}
{\slshape\normalsize Présenté en vue de l’obtention du diplôme de} \\
\vspace{0.25em}
{\textrm{\normalsize\textbf{Doctorat en Sciences Informatique}}} \\ 
\vspace{1em}
{\slshape\normalsize Par } \\
\vspace{1em}
{\bfseries\large Seif Allah \textsc{Ben Chaabene}} \\
\vspace{3em}\hrule
\vspace{1em}
{\textrm{\Huge\textbf{Proposition d’approches de déploiement des unités de bord de route  dans les réseaux véhiculaires}}}
\vspace{1em}
\hrule

\vspace{4em}
\textbf{Soutenu le 23/08/2021 devant le jury composé de:} \\ \vspace{1em}

\begin{table}[!ht]
\centering
\resizebox{0.9\linewidth}{!}{
\begin{tabular}{p{6cm}p{6.5cm}l}
\textbf{M. Belhassen \textsc{ZOUARI}} &  & (Président)\vspace{2px}\\
\textbf{M. Walid \textsc{BARHOUMI}} &  & (Rapporteur)\vspace{2px}\\
\textbf{M. Djamal \textsc{BENSLIMANE}} &  & (Rapporteur)\vspace{2px}\\
\textbf{Mme. Hella \textsc{KAFFEL}} &  & (Examinateur)\vspace{2px}\\
\textbf{M. Sadok \textsc{BEN YAHIA}} &  & (Directeur de thèse)\vspace{2px}
\end{tabular}
}
\end{table}

\vspace{5em}
\textbf{Au sein du laboratoire: LIPAH}


\center\large{2020-2021}
\end{titlepage}

\dominitoc
\pagenumbering{roman}
\chapter*{Dédicace}
\bigskip \bigskip \bigskip
\bigskip \bigskip \bigskip
\bigskip \bigskip \bigskip
\bigskip \bigskip \bigskip
\begin{center}
Je dédie ce travail à mes parents pour l'amour qu'ils m'ont toujours donné, leurs sacrifices et leurs tendresses. \\
Je le dédie également à mes frères, ma sœur, ma femme soumaya  pour son soutien et son encouragement et la petite Mayar "bayo".\\

Qu'ils trouvent ici le témoignage de ma gratitude et de mon respect.
\end{center}
\bigskip
\begin{flushright}\begin{small}\textit{Seif Allah Ben Chaabene}.\end{small}\end{flushright}
\newpage
\chapter*{Remerciements}
En premier lieu, je remercie Dieu le tout puissant pour ses faveurs et sa gratitude. Ensuite, je tiens à exprimer mes remerciements à Monsieur Belhssen ZOUARI, pour avoir accepté de présider le jury de la soutenance.\\

Je tiens à exprimer ma gratitude à Monsieur Walid BARHOUMI et Monsieur Djamal BENSLIMANE, pour avoir accepté de rapporter ce travail, ainsi qu'à Madame Hella KAFFEL de me faire l'honneur de faire partie du jury de cette thèse.\\

Je tiens à exprimer mon respect, ma gratitude et mes chaleureux remerciements à Monsieur Sadok BEN YAHIA, mon directeur de recherche, pour m'avoir conseillé, encouragé et soutenu tout au long de la thèse.\\

J'adresse mes remerciements les plus chaleureux à mon co-directeur de thèse Monsieur Taoufik YEFERNY, pour tous les conseils techniques et scientifiques qu'il m'a apportés et qui m'ont beaucoup aidé à accomplir mes travaux.\\

Enfin, j'adresse un merci spécial à tous les membres du département des Sciences de l'informatique de la Faculté des Sciences de Tunis. J'adresse aussi un remerciement particulier à tous mes collègues et tous ceux qui m'ont aidé de près ou de loin pour finir ce travail.

\tableofcontents
\listoffigures
\listoftables
\newpage
\listofalgorithms
\newpage
\thispagestyle{empty}
\begin{center} \LARGE{\textbf{Notations}} \end{center}
$\mathcal{MF} :$ Ensemble des motifs fréquents.\\
$\mathcal{MR} :$ Ensemble des motifs rares.\\
$\mathcal{MF}e :$ Ensemble des motifs fermés.\\
$\mathcal{MFF} :$ Ensemble des motifs fermés fréquents.\\
$\mathcal{MFR} :$ Ensemble des motifs fermés rares.\\
$\mathcal{GM} :$ Ensemble des motifs générateurs minimaux.\\
$\mathcal{GMR} :$ Ensemble des motifs générateur minimaux rares.\\
$\mathcal{GMF} :$ Ensemble des motifs générateur minimaux fréquents.\\
$\mathcal{MFM} :$ Ensemble des motifs fréquents maximaux.\\
$\mathcal{MRM} :$ Ensemble des motifs rares minimaux.\\
$\mathcal{MZ} :$ Ensemble des motifs zéros.\\

\newpage
\pagenumbering{arabic}
\chapter{Introduction générale}
\minitoc
\markboth{Introduction générale}{Introduction générale}
\section{Contexte général}

Les réseaux sans fil, en particulier les réseaux adhoc, sont de plus en plus populaires aujourd'hui en raison de leurs avantages tels que la mobilité des terminaux qui permet la communication n'importe où, n'importe quand, avec n'importe qui, etc. Les réseaux adhoc mobiles sont caractérisés par la nature dynamique de la topologie du réseau, du support de communication radio, de la coopération distribuée, de l'absence d'une infrastructure préexistante ou d'un point de gestion, des ressources limitées, de la connectivité sporadique, de l'absence d'autorité de certification et de la vulnérabilité physique du nœud. L'évolution rapide de la communication des données sans fil a conduit les chercheurs à explorer leurs applications et leurs travaux dans le réseau véhiculaire nommé VANET ('Vehicular Adhoc NETwork'). VANET est une forme d'un réseau mobile adhoc permettant l'établissement d’une communication au sein d'un groupe de véhicules les uns à portée des autres et entre les véhicules et les équipements fixes à portée, usuellement appelés \textit{RSU} ('Road Side Unit'). En effet, VANET est utilisé pour soutenir le développement des Systèmes de Transport Intelligent STI (Intelligent Transportation System ITS). Les motivations pour lesquelles beaucoup de travaux de recherches s'y intéressent sont : la sécurité routière, la mobilité, la productivité et la protection de l'environnement. L'idée qu'un véhicule partage des informations utiles avec d'autres objets pour assurer la sécurité et la vie humaine sur la route est vraiment convaincante. Les statistiques de l'organisation de coopération et de développement économiques (OCDE) en 2016, et pour la France uniquement, le nombre des accidents impliquant des victimes est $57 522$, le nombre de blessures est $72 645$, et le nombre de décès causés par des accidents de la route est de $3477$. Mis à part le problème de la sécurité, toutes les villes du monde connaissent des problèmes d'embouteillages et de congestion ce qui engendre un gaspillage de carburant, perte de temps et pollution de l'environnement. Basé sur le rapport annuel 2019 d'INRIX, les heures dans la congestion à Paris sont au 7ème rang mondial avec une perte de 165 heures par conducteur par contre le conducteur à New York passe $140$ heures en congestion par an. Ce problème conduit également à un problème d'environnement où beaucoup de combustibles fossiles sont gaspillés.\\

Au cours des dernières années, les chercheurs, le gouvernement et l'industrie automobile se sont intéressés aux VANET, où plusieurs applications de STI ont vu le jour non seulement pour des applications de sécurité, mais aussi pour des applications de divertissement. Par conséquent, de nombreuses applications sont proposées pour les VANET tels que la prévention des accidents, la réduction de la congestion, la prévention des embouteillages et l'accès à Internet. La conception et la mise en œuvre de protocoles, d'applications et de systèmes pour les VANET nécessitent de considérer ses caractéristiques distinctives, en particulier la grande mobilité des véhicules, le changement rapide de la topologie et le chemin prévu. En outre, il doit également tenir compte de plusieurs facteurs, telle l'exigence de qualité de service (QoS).
Dans ce contexte, plusieurs projets de recherche et développement sur VANET ont été lancés :

\begin{itemize}
\item Cooperative Vehicle-Infrastructure System, CVIS (2006-2010) \cite{ding2014}
\item SAFESPOT (2006-2010) \footnote{http://www.safespot-eu.org/}
\item Safety-related Cooperating Cyber-Physical Sytems, SafeCOP (2016-2019) \cite{KOS06} \footnote{https://cordis.europa.eu/project/id/692529/fr}
\item Vehicle Safety Communication, VSC-2 (2005-2009) \footnote{https://www.nhtsa.gov/sites/nhtsa.gov/files/811492d.pdf}
\item Vehicle Infrastructure, VII (2005-2009) \footnote{http://ral.ucar.edu/projects/vii.old/vii/docs/VIIArchandFuncRequirements.pdf}
\item CarCoDe (2013-2015) \footnote{http://www.cister.isep.ipp.pt/itea2-carcode/theproject.php}
\end{itemize}

\section{Motivations et contributions}
D'après les chiffres énoncés par les organisations mondiales de la santé et de la lutte contre la pollution concernant respectivement ($i$) le nombre des décès à cause des accidents routiers et ($ii$) le taux de pollution mondiale à cause de dégagement de CO2 à partir des échappements des véhicules nous remarquons que l'augmentation suit une suite géométrique durant cette dernière décennie.\\
Tout d'abord, concernant le nombre de mortalité et le nombre des blessées, les chiffres annoncés par l'OMS (Organisation Mondiale de la Santé), dans le sommet mondial à Genève, en 2019 parlent de 1,35 million de morts chaque année et $50$ millions personnes gravement blessés, $90$\% de ces chiffres concerne les pays en développement. En outre, la gravité de ces chiffres dépasse le plan social pour atteindre le plan économique comme annoncé par l'étude de la Banque mondiale en Janvier 2018 \footnote{www.banquemondiale.org/fr/news/press-release/2018/01/09/road-deaths-and-injuries-hold-back-economic-growth-in-developing-countries}: "Pour les pays qui n’investissent pas dans la prévention des accidents de la route, le manque à gagner s’élèverait entre 7 et 22 \% du PIB par habitant sur une période de 24 ans".
\begin{figure}[!ht]
\centering
\includegraphics[width=1\textwidth]{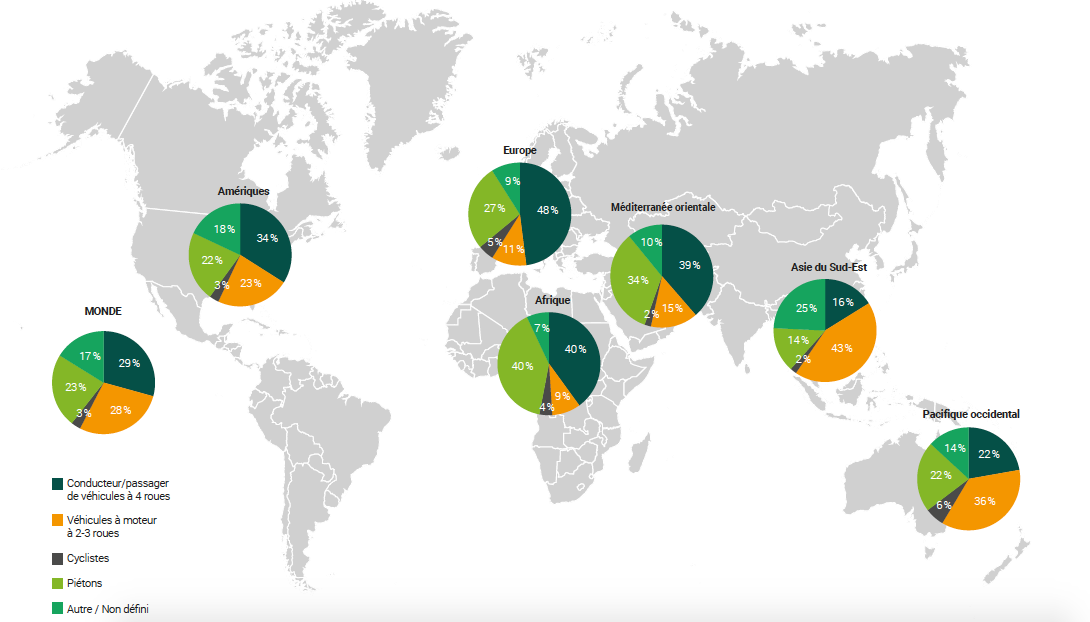}
\caption{Nombre des véhicules et des piétons dans les différents continents}
\end{figure}\\
Ensuite, en ce qui concerne la pollution causée par le dégagement des gaz d'échappement des véhicules, cette augmentation peut être expliquée par trois principaux facteurs : 
($i$) la multiplication du nombre des constructeurs automobiles ce qui implique une augmentation du nombre de véhicules vendu chaque année, ($ii$) l'évolution démographique de la population mondiale (augmentation de 1 milliard entre 2010 et 2020) et l'espérance de vie à la naissance est un facteur favorisant l'apparition des nouvelles zones urbaines et ($iii$) la hausse de l'indice de pouvoir d'achat, couramment nommé PIB par habitant en sciences économique, dans la plupart des pays du monde.\\
Les statisticiens affirment que cette courbe de pollution terrestre suivra la même pente dans les nouvelles décennies ce qui peut engendrer un phénomène destructif pour la terre, une part très importante de la pollution vient des impacts sociaux et économiques, ces causes ont poussé des gouvernements, des organisations, des constructeurs et même des banques à investir dans des domaines d'intelligence routières afin de comprendre en premier lieu les causes déclenchantes de ces augmentations et dans un deuxième lieu de trouver les solutions nécessaires pour réduire ces dégâts.\\

Dans cette thèse, nous nous sommes intéressés au problème de couverture des réseaux véhiculaires VANET dans des applications de gestion de trafic et des applications de divertissement. Pour faire face à ce problème et dans le cadre de 3 contributions de cette thèse, nous répondons aux questions suivantes : \textit{"Combien faut-il de Road Side Unit pour couvrir une carte ? Comment minimiser le coût de déploiement des \textit{RSU}s et maximiser la couverture ? Où placer les Road Side Unit pour mieux augmenter la couverture de la carte ? Comment prévoir les mouvements futurs des véhicules afin de mieux placer les Road Side Unit ?"}.\\

Notre première contribution,  \textrm{\emph{SPaCov/SPaCov+}}, consiste, en premier lieu, à adapter une méthode d'analyse de données séquentielles pour l'utiliser dans un jeu de données décrivant les mouvements véhiculaires, la première phase consiste à extraire les séquences \underline{les plus visitées} (séquences fréquentes) et  \underline{les moins visitées} (séquence rares) par les véhicules afin de les couvrir. En deuxième lieu, le problème est de trouver \textbf{le nombre} et \textbf{l'emplacement} des \textit{RSU} à déployer  pour couvrir les séquences extraites dans la première phase. La deuxième phase de cette méthode, nommée CoC, consiste à chercher les emplacements des \textit{RSU}s afin de maximiser la couverture des noeuds trouvés dans la première phase. Pour résoudre ce problème, nous avons fait recourt à une notion mathématique qui suppose l'ensemble des noeuds sous forme d'un hypergraphe et pour le couvrir nous avons besoin d'extraire les traverses minimales à partir de ce dernier. Cette notion mathématique qui décrit formellement une traverse minimale d'un hypergraphe nous garantit que le nombre des \textit{RSU}s sera minimal (traduit par un coût de déploiement minimal) avec une couverture maximale de l'hypergraphe.\\

Notre deuxième contribution, nommée dans le reste de ce manuscrit  \textrm{\emph{HeSPiC}}, est une méthode supervisée basée sur une heuristique pour le déploiement des \textit{RSU}s.  \textrm{\emph{HeSPiC}} permet de placer $k$ \textit{RSU}s, $k$ est un nombre choisi par l'utilisateur selon ses contraintes budgétaire, dans les croisements optimaux afin de maximiser la couverture dans le réseau VANET. En effet, notre contribution utilise un module heuristique pour classer les croisements afin de choisir les emplacements dans lesquels on placera des \textit{RSU}s. Le choix des emplacements des \textit{RSU}s est assuré par un module nommé HBCC (Heuristic Budget-Constrained Coverage), qui adapte trois fonctions mathématique pour classer les croisements candidats : ($i$) poids de chaque croisement, ($ii$) matrice de distance entre les croisements et ($iii$) une probabilité qui suit la loi de poisson pour prévoir les mouvements des véhicules au futur.  Le module HBCC tiens compte de trois principaux facteurs avec l'utilisation des trois fonctions mathématiques citées auparavant à savoir ($i$) l'aspect spatial avec l'utilisation d'une matrice de distance entre les croisements, ($ii$) l'aspect temporel (historique du mouvement des véhicules : du passé jusqu'au présent) avec l'utilisation de la fonction mathématique adaptée à notre problématique Poids, nommée $\mathcal{W}$ dans le reste de ce manuscrit,  et ($iii$) l'aspect temporel mais avec une projection  sur les mouvements futurs  des véhicules ce qui permet de bien classer les croisements.\\

Notre dernière contribution, nommée  \textrm{\emph{MIP}} (Most Important Patterns), consiste à réduire la taille gigantesque des données enregistrées lors de mouvements des véhicules en appliquant des fonctions mathématiques pour garder seulement des trajectoires nommées \textit{utiles} et \textit{bénéfiques}. La méthode est composée de deux grandes parties:
\begin{enumerate}
    \item Consiste à trouver les patterns utiles et bénéfiques afin de garder uniquement les trajectoires dite utiles.
\item C'est le même module CoC utilisé dans la première contribution  \textrm{\emph{SPaCov+}} qui permet de couvrir l'ensemble des trajectoires utiles avec l'application d'une fonction mathématique qui permet de trouver les traverses minimales à partir d'un hypergraphe.
\end{enumerate}
\section{Organisation du manuscrit}

Le premier chapitre dénonce le contexte général, les motivations et un survol général sur nos contributions.\\

Le deuxième chapitre détaille les travaux de recherche qui se focalisent sur les systèmes de transport intelligents (STI). Dans ce chapitre nous présentons ($i$) la différence entre l'internet of things et l'internet of everything, ($ii$) une comparaison entre les réseaux mobiles MANET et VANET et une description de plusieurs pistes de recherche à savoir : les travaux qui s'intéressent à la dissémination et au routage dans les réseaux VANET, la sécurité d'échange d'informations dans les réseaux VANET et à la fin nous détaillons davantage les travaux pointant sur le problème de déploiement des \textit{RSU}s dans les réseaux VANET puisque nos approches font partie de ce type de travaux.\\

Le troisième, quatrième et cinquième chapitre présentent respectivement les trois nouvelles approches  \textrm{\emph{SPaCov/SPaCov+}},  \textrm{\emph{HeSPiC}} et  \textrm{\emph{MIP}}. Une description détaillée de chaque approche avec les fondements mathématiques ainsi qu’un exemple simulant le mouvement des véhicules sont présentés dans la première partie de chaque chapitre.  La présentation des algorithmes ainsi qu'une étude théorique de la complexité avec les preuves de complétude de ces derniers fait l'objet de la deuxième partie de ces trois chapitres.\\

Le sixième chapitre présente une étude expérimentale qui évalue les performances respectives des approches  \textrm{\emph{SPaCov/SPaCov+}},  \textrm{\emph{HeSPiC}} et  \textrm{\emph{MIP}} via le simulateur OMNET++. Nous évaluons les performances des deux premières approches ( \textrm{\emph{SPaCov/SPaCov+}} et  \textrm{\emph{MIP}}  avec les résultats  des méthodes pionnières de la littérature. Cependant,  l'évaluation de la troisième approche, nommé  \textrm{\emph{HeSPiC}}, sera effectuée en changeant différents paramètres tels que le nombre de véhicules, la zone de couverture, la plage de communication et le budget en terme du nombre des \textit{RSU}s.

Ce manuscrit se termine par une conclusion, qui résume l'ensemble de nos travaux ainsi qu'une ouverture sur quelques perspectives futures de recherche.
\newpage
\thispagestyle{empty}
\mbox{}
\newpage
\chapter{État de l'art}{}
\chaptermark{État de l'art}
\minitoc
\section{Introduction}
Ces dernières années, les systèmes de transport intelligents (STI) sont devenus un nouveau moyen pour résoudre divers problèmes de circulation routière, tels que les accidents de la route, la congestion, etc. en reliant les personnes, les routes et les véhicules dans un réseau d'information et de communication via des technologies de pointe. À cet égard, les réseaux adhoc des véhicules (VANET) sont apparus comme un moyen intéressant pour déployer des services STI y compris des applications de sécurité et autres. Les VANET sont principalement composés de véhicules échangeant des informations entre eux et avec une infrastructure, par exemple, des unités routières (\textit{RSU}), installées le long des routes. En effet, les VANET prennent en charge trois types de communication à savoir Vehicle-to-Vehicle (V2V), Vehicle to Infrastructure ou \textit{RSU} (V2I) et Infrastructure-to-Infrastructure (I2I). La collecte d'information des données véhiculaires circulant dans une grande ville avec une topologie complexe et leurs diffusions nécessitent le déploiement d’un grand nombre de \textit{RSU} coûteux à différents endroits. À cet égard, plusieurs études de recherche se sont concentrées sur le calcul des meilleurs emplacements des \textit{RSU}s et le nombre adéquat de ces derniers. 
\section{De l'Internet of Things vers l'Internet of Everything }
L'internet des objets couramment noté IoT et l'internet de Tout couramment noté IoE sont deux concepts qui se ressemblent beaucoup, mais ils ne sont pas identiques. Il existe à la fois des caractéristiques communes et distinctes de l'Internet de tout (IoE) par rapport à l'Internet des objets (IoT). Les solutions Internet des objets étant de plus en plus courantes, il est important de comprendre comment ces deux termes sont interdépendants.\\
Tout d'abord, nous commençons par expliquer les deux notions.\\
\subsection{Internet des objets - IoT }
Historiquement le terme IoT a été inventé par Kevin Asthon, un technologue au sein de la société Procter \& Gamble, il décrit la technologie comme étant la connexion de plusieurs appareils à l'aide d'étiquette RFID pour la gestion de la chaîne d'approvisionnement. Cependant, son idée sur la connectivité d'objets par le réseau RFID n'est plus valable puisque la connexion des appareils de nos jours dépasse les limites du réseau RFID pour atteindre le réseau mondial internet en se basant sur le protocole IP. 
Au début du 21 ème siècle, le terme IoT devient très utilisé par les médias et par les journaux ce qui conduit à une première conférence mondiale sur ce nouveau concept d'internet des objets qui s'est tenu en suisse en 2008. Dans cette dernière conférence 23 participants de pays différents ont discuté les communications sans fil à courte portée et les réseaux de capteurs. À partir de l'année 2010 l'IoT devient un domaine incontournable dans la production technologique, plusieurs définitions ont vu le jour dans des revues scientifiques comme étant une technologie qui permet l'interconnexion entre l'Internet et les objets, les lieux et les environnements physiques. À titre d'exemple, en 2020 nous dépassons les 30 milliards objets connectés dans le monde (voir Figure \ref{nb-iot}), un chiffre représentant 4 fois la population mondiale, ce qui implique que dans les prochaines décennies avec l'augmentation du flux de données généré par ces objets, le trafic d'information sur Internet ne sera plus contrôlé par des êtres humains mais sera dans une grande partie supervisé et traité par des machines voir des objets.
\begin{figure}[!ht]
\centering
\includegraphics[width=0.9\textwidth]{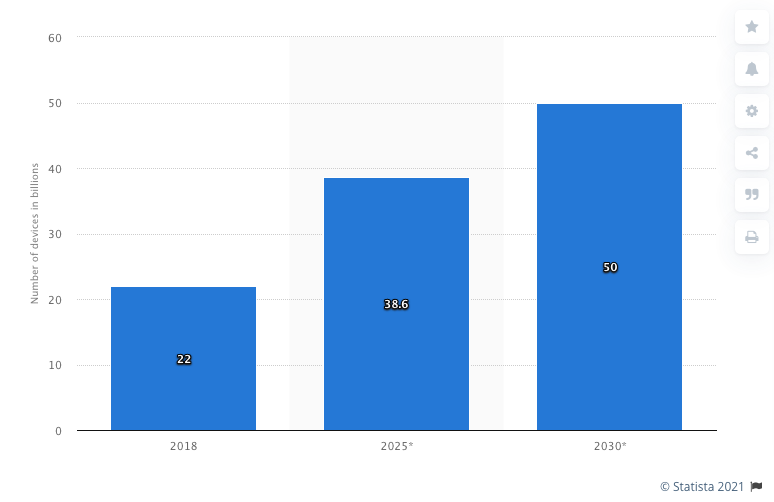}
\caption{Évolution du nombre d'objet connecté entre 2018 et 2030 }
\end{figure}
\label{nb-iot}     

\subsection{Internet de tout - IoE }
Par définition, c'est l’ensemble du monde physique qui sera interconnecté. Aussi bien les personnes, via les réseaux sociaux, que les objets avec le développement de l’IoT, les processus et les données avec le basculement vers le cloud, etc. Ce futur "Internet du tout-connecté" se traduira par la fusion concrète d’Internet, de l’Internet des objets, et du big data, une définition popularisée par la société Cisco. L’IoE sera un formidable levier de croissance en permettant aux entreprises, aux collectivités et aux particuliers de faire d’importantes économies, et en créant une panoplie de nouveaux emplois. L’Internet of Everything devrait permettre d’accroître la productivité et de réduire les coûts des entreprises, ainsi que d’optimiser les performances des États et des administrations publiques. Ce qui, bien sûr, profitera aux citoyens par ricochet.\\
L’Internet de tout pourrait soutenir de nombreux développements et améliorations comme par exemple l’habitat intelligent, le contrôle optimisé des consommations énergétiques et naturelles, le stationnement intelligent, les réseaux véhiculaires intelligent, une tarification routière plus adaptée, etc.\\
Il permettrait aussi d’accroître les performances administratives, en améliorant la productivité des agents et en réduisant les coûts d’exploitation.\\
En outre, et afin de faire connecter ce nombre d’objet, qui grimpe d’une manière exponentielle, avec les bases de données qui contiennent des quantités d'informations importante et qui augmentent avec le temps, les chercheurs ont inventé des infrastructures réseaux afin de pouvoir assurer un ratio raisonnable entre la connectivité et le débit selon le secteur de fonctionnement de chaque domaine. \\
Dans ce qui suit, nous présentons les technologies actuelles d'infrastructures réseaux IoT. Chaque technologie, afin de s'imposer, doit se développer au niveau mondial et s'associer à diverses sociétés (constructeur informatique, exploitant de réseaux de téléphonie mobile et fixe Internet, les moteurs de recherche, les sociétés informatiques). Le réseau doit permettre d'envoyer (mais aussi recevoir dans certains cas) de très petits messages sur de longues portées sans passer par des systèmes coûteux de réseaux mobiles et en consommant peu d'énergie. Les objets qui communiquent à l'intérieur d'un immeuble à l'instar de WiFi, bluetooth, Z-wave, ZigBee, etc. ou sur des longues distances à savoir GSM, 3G, GPRS, CDMA, 4G et 5G ne font pas partie de l'Internet des objets. Cependant, les réseaux SigFox, Lora, MQTT, NB IoT, Qowisio, etc s'imposent comme étant des infrastructures réseaux IoT.
\begin{itemize}
\item SigFox : Le réseau Sigfox est une offre de service basée sur les protocoles réseaux sans fil bidirectionnel. L'offre commerciale actuelle en France permet, par jour et par périphérique (capteur), le rapatriement de 140 messages de 12 octets et l'envoi de 4 messages aux détecteurs. Le réseau utilise la bande de $868$ MHz avec un débit de $100$ bit/s. La portée est de $30$ à $40$ km grâce à des antennes de $2$ m de hauteur.
\item Lora : Le réseau dit Lora est un réseau dont le protocole LoRaWAN s'appuie sur la technique de modulation (couche physique) LoRa. LoRa utilise les bandes de fréquences ISM (868 MHz en Europe et 915 MHz aux USA), une portée comprise entre 15 et 20 km dans les zones rurales et entre 3 et 8 km dans les zones urbaines et fournit un débit faible entre 300bps et 50Kbps selon le facteur d’étalement (spreading factor)
\item MQTT : est un protocole de messagerie publish-subscribe basé sur le protocole TCP/IP, c'est l'acronyme de Message Queuing Telemetry Transport. Dans sa première version, le projet est écrit par la société IBM en 1999 l'objectif était d'avoir un protocole efficace en bande passante, léger et utilisant peu d'énergie de batterie. A ce jour 5 versions ont été apparues et plusieurs agents MQTT sont disponibles à savoir ActiveMQ, ejabberd, JoramMQ et Mosquitto.
\item NB IoT : C'est la technologie supporté par la Chine pour le développement du secteur, après avoir été identifié et soutenu en amont par les autorités, le NB-IoT a désormais été adopté par les grands groupes technologiques chinois qui se positionnent en leaders mondiaux du déploiement de réseaux et de la production d’équipements. En Chine, le développement d’applications IoT est favorisé par de nombreuses plateformes et initiatives locales, notamment pour la STI, smart city et l’internet industriel. À l’international, les acteurs chinois coordonnent leurs efforts et on notamment réussi à faire inclure le NB-IoT dans les standards 5G dans le cadre du 3rd Generation Partnership Project (3GPP). Néanmoins, bien que de nombreux analystes attendent une forte pénétration du NB-IoT au niveau mondial, le succès de la technologie dépendra aussi fortement des caractéristiques propres à chaque marché : infrastructures existantes, usages, intérêts industriels ou encore enjeux de cyber sécurité.
\end{itemize}
\begin{landscape}
\subsection{Comparaison entre les infrastructures réseaux}
\vspace{3cm}
\begin{table}[h!]
 \centering
 \begin{center}

\label{rx-iot}

\begin{tabular}{|c||c|c|c|c|}
\hline
&MQTT&LoRa&SigFox&NB-IoT\\
\hline
\hline
Débit&Selon l'implémentation &50Kbit/s& 50Kbit/s& 250Kbit/s\\
\hline
Mémoire& OUI &NON &NON &NON\\
\hline
Protocole de communication  &TCP/IP& LoraWan& LPWAN& NB-LTE\\
\hline
Max Données& Non limité& Non limité& 140Msg/Jour& Non limité\\
\hline
Bande de fréquence& N.D & 868 Mhz/902 Mhz&868 Mhz/ 902Mhz&200 Khz\\
\hline
Serveur dédié &Broker Système& Serveur LoRaWan& Pas de serveur& Pas de serveurto\\
\hline
\end{tabular}

\caption{Comparaison entre les infrastructures réseaux IoT.}
\end{center}
\end{table}
\end{landscape} 

\section{Les réseaux mobiles MANET et VANET}
Nous présentons en premier lieu les réseaux mobiles adhoc MANET (Mobile Adhoc Network) ainsi que les nouvelles contraintes de déploiement des systèmes P2P sur ces réseaux. En second lieu, nous présentons le réseau des véhicules qui forme une spécification du réseau MANET couramment nommé VANET (Vehicular Adhoc Network).\\
\subsection{Réseau mobile MANET}
Dans un premier temps, le groupe de travail s'est attaché aux questions de performances dans les réseaux Adhoc et au développement d'une série de protocoles de routage alors expérimentaux, tant dans la famille des réactifs (AODV, DSR) que des proactifs (OLSR, TBRPF), ou bien encore des hybrides (ZRP).
De nos jours, un réseau adhoc mobile (MANET) est principalement 
Constitué par des dispositifs mobiles (smartphone, téléphone cellulaire, PDA, etc.) qui communiquent entre eux en utilisant des technologies de communication sans fil à l'instar de Bluetooth, wifi, etc. La communication dans le réseau MANET se fait d'une manière décentralisée sans avoir une infrastructure existante et fixe à l'avance. En effet, la communication des nœuds de réseau MANET se fait par deux manières :
\begin{itemize}
    \item Une communication directe (dans la zone couverte) entre deux nœuds voisins.
    \item Une communication indirecte (hors zone couverte) entre deux nœuds non voisins qui nécessite l'intervention d'un ou plusieurs nœuds adjacents (nommé aussi intermédiaires) pour acheminer le message entre l'émetteur et le destinataire. 
\end{itemize}

\subsection{Réseau mobile VANET}
Vehicular adhoc Network ou réseau adhoc de véhicules noté VANET, est une forme de Mobile adhoc NETworks MANET, pour fournir des communications au sein d'un groupe de véhicules à portée les uns des autres et entre les véhicules et les équipements fixes à portée, usuellement appelés équipements de la route (Road Side Unit pour l'acronyme de RSI).\\
VANET peut être utilisé pour soutenir le développement des Systèmes de Transport Intelligent STI (Intelligent Transportation System ITS). Les motivations pour lesquelles beaucoup travaux de recherches s'intéressent sont : la sécurité, la mobilité, la productivité et  la protection de l'environnement. L'idée qu'un véhicule partageant des informations utiles avec d'autres objets pour assurer la sécurité de la vie humaine sur la route est vraiment convaincante, Les statistiques de l'Organisation de coopération et de développement économiques (OCDE) en 2016, en France seulement, le nombre des accidents impliquant des victimes est de 57522, le nombre de blessures est de 72645, et le nombre de décès causés par des accidents de la route est de $3477$. Mis à part le problème de la sécurité. Le but optimal est que les réseaux véhiculaires contribueront à des routes plus sûres et plus efficaces à l'avenir en fournissant des informations opportunes aux conducteurs et aux autorités intéressées. Dans son développement, VANET fournira une communication de données continue à grande vitesse non seulement entre les véhicules et les infrastructures routières, mais aussi avec les piétons, les appareils et les réseaux cellulaires. Ce nouveau système de communication véhiculaire appelé Vehicle-to-everything (V2X)  est composé de Vehicle-to-Vehicle (V2V), Vehicle-to-Infrastructure (V2I), Vehicle-to-Device (V2D) et Vehicle-to-Network (V2N) communication.\\
Ces dernières années, plusieurs organisations ont créé la standardisation de l'architecture VANET. L'un d'entre elles est l'IEEE qui propose l'architecture appelée WAVE (Wireless Access in Vehicular Environments). Il existe également d'autres normes proposées par l'ETSI, appelées architecture ETSI TC-ITS, et architecture CALM proposée par l'ISO.
\begin{figure}[!ht]
\centering
\includegraphics[width=0.9\textwidth]{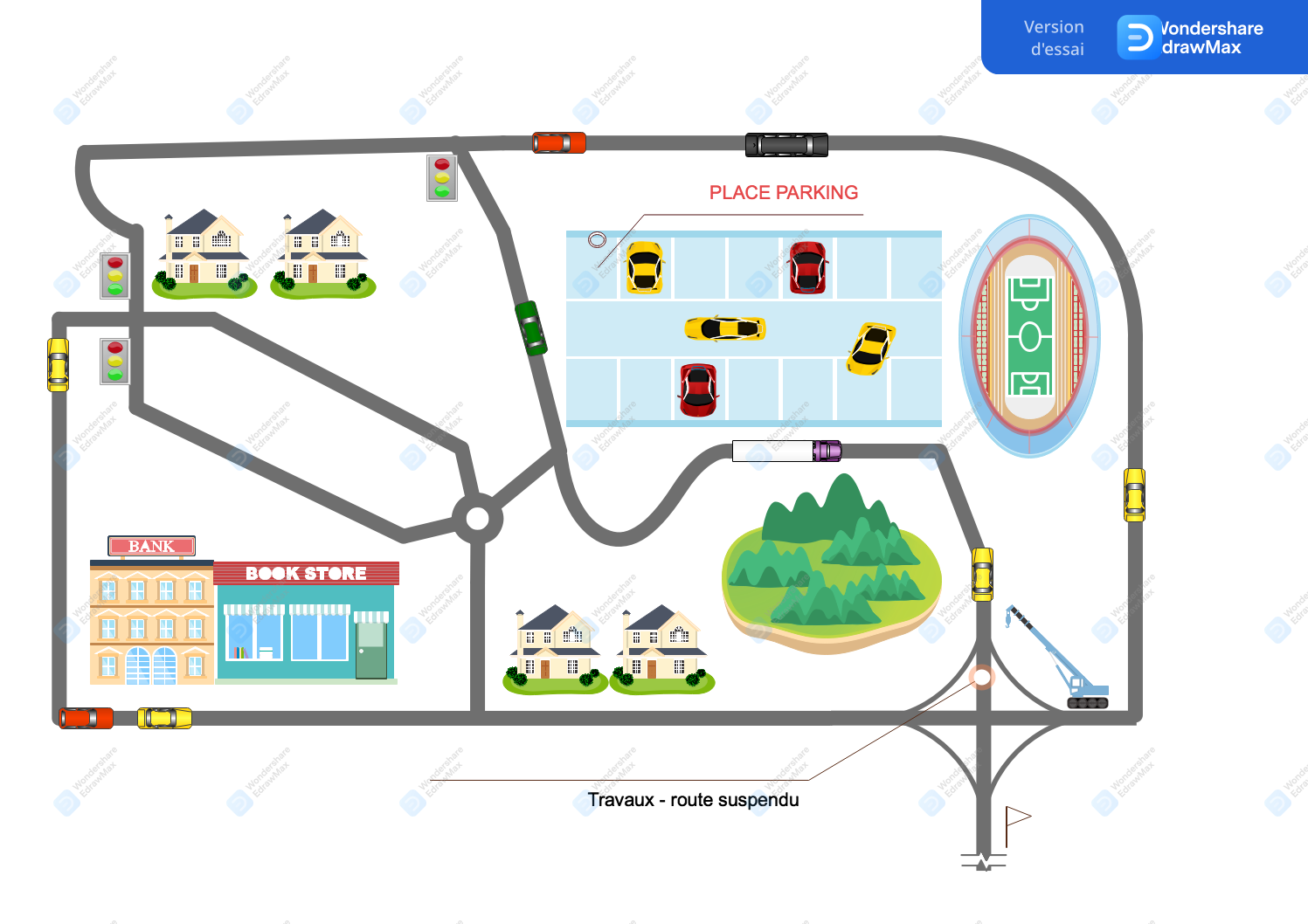}
\caption{Réseau routier d'une ville avec communication véhiculaire.}
\end{figure}\label{vanet}

\section{Classification et description des méthodes de déploiement des \textit{RSU}s}
Le défi florissant serait de ($i$) minimiser le nombre de \textit{RSU}s autant que possible afin de réduire le coût de déploiement et ($ii$) maximiser le taux de couverture, c'est-à-dire le nombre de véhicules qui font contact avec au moins un \textit{RSU} lors de leurs déplacements. Dans la littérature, les stratégies de déploiement de \textit{RSU} existantes, appelées stratégies de couverture, peuvent être classifiées en trois grandes catégories :
\begin{itemize}
\item Couverture spatiale \cite{K08} \cite{K10} \cite{K14}
\item Couverture temporelle \cite{H09} \cite{H10} \cite{H11}
\item Couverture Spatio-temporelle \cite{B02} \cite{B12} \cite{B13} \cite{B1}
\end{itemize}

\begin{figure}[!ht]
\centering
\includegraphics[width=0.9\textwidth]{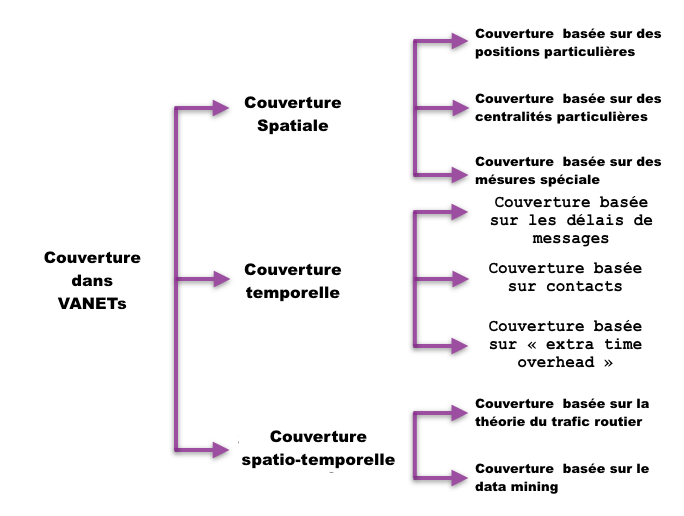}
\caption{Classification des stratégies de déploiement des \textit{RSU}s.}
\end{figure}\label{cov}

Tout d'abord, nous décrivons la première famille des méthodes qui s'appuie sur les attributs spatiaux de la topologie de la carte routière pour calculer les emplacements en maximisant le nombre de véhicules pouvant être contactés, par exemple, les intersections les mieux centralisées.\\
Néanmoins, la couverture temporelle se concentre sur la couverture de la communication entre les véhicules à déplacement rapide et les \textit{RSU} fixes.\\
Enfin, la couverture spatio-temporelle repose sur les patterns de mobilité des véhicules.  En effet, l'idée derrière cette famille d'approche est de calculer l'emplacement des \textit{RSU}s couvrant tous les modèles de mobilité des véhicules circulant dans une ville donnée. En outre, les méthodes de couverture spatio-temporelle existantes supposent que les véhicules se déplacent dans une ville donnée selon des modèles de mobilité qui pourraient être déduits à partir de leurs trajectoires. En outre, couvrir ces trajectoires qui forment les modèles de mobilité des véhicules, revient à maximiser le taux de couverture totale.\\
La couverture basée sur l'aspect spatial et celle basée sur l'aspect temporel ne représente pas la réalité pour décrire la mobilité du véhicule. Cependant, la mobilité des véhicules sera plus proche de la réalité si nous utilisons les deux aspects ensemble. Plusieurs travaux fondent leurs approches sur la segmentation des cartes en régions et d'étudier la mobilité des véhicules sur ces régions, alors que d'autres travaux plus réalistes étudient le mouvement des véhicules pour avoir plus de détail sur le modèle de mouvement. À notre connaissance, aucune approche précédente n'utilise l'ordre chronologique de passage des véhicules par les croisements afin de représenter avec précision le modèle de mouvement.\\

Les travaux de la littérature  sur le déploiement des \textit{RSU}s peuvent être divisés en trois catégories, comme illustré dans la Figure \ref{cov}

\begin{enumerate}
    \item \textbf{Approches basées sur la couverture spatiale} \cite{K08} \cite{K10} \cite{K14} : Cette catégorie est basée sur l'analyse de l'attribut spatial d'une topologie routière. La couverture spatiale peut choisir de déployer des \textit{RSU}s dans une position spéciale comme : l'intersection des jonctions, le point médian de la zone (différentes formes géographiques de la zone ont été proposées), ou les intersections avec une centralité élevée.
    \item \textbf{Approches basées sur la couverture temporelle} \cite{H09} \cite{H10} \cite{H11}: Cette catégorie est basée sur l'aspect temporel de la communication entre les véhicules en mouvement et les \textit{RSU}s fixes. Les contraintes utilisées dans la couverture temporelle sont plus proches des facteurs temporels (vitesse de transmission, temps de conduite, etc.).
    \item \textbf{Approches basées sur la couverture Spatio-temporelle} \cite{B02} \cite{B12} \cite{B13} \cite{B1}: Cette catégorie combine à la fois les aspects spatiaux et temporels pour décrire le mouvement des véhicules dans le temps et dans l'espace ou exploiter leurs propres modèles de mobilité à partir de l'historique des traces de véhicules.
\end{enumerate}
\subsection{Méthodes basées sur la couverture Spatiale}
Kafsi et al. Introduisent dans son article \cite{K08} une étude de la connectivité des réseaux VANET en modifiant la densité des véhicules, porté de la communication, etc. Le modèle du trafic est représenté par les auteurs sous forme d'une grille de hachage de $N$ routes verticales et de $N$ routes horizontales (d'où $N^{2}$ intersections). La distance entre deux lignes horizontales consécutives (de même pour les lignes verticales) est de $L$ mètres. Trois types d'intersections sont proposés dans l'article:
\begin{enumerate}
    \item Absence de feux de signalisation: les véhicules cèdent à tout véhicule arrivant de leur droite.
    \item  Présence de feux de signalisation synchronisé: Toutes les intersections sont synchronisées, de sorte que la lumière soit verte dans les lignes horizontales à toutes les intersections en même temps.
    \item Présence des feux de signalisation alternés:  ajustement de la durée de la phase des feux de signalisation pour simuler le vert effet de vague, afin que les conducteurs puissent passer par intersections dans une direction principale sans s'arrêter
\end{enumerate}
Les véhicules entrent dans le domaine de la simulation à des moments choisis uniformément pendant toute la durée de la simulation. Si le temps de cette dernière est suffisamment long, les arrivées de véhicules peuvent localement être vues comme un processus de Poisson d'intensité. Chaque véhicule entrant choisit sa destination (point de sortie) de deux manières différentes:
\begin{enumerate}
    \item Destination Anti-Diamétrique: Le véhicule source et la destination (point de sortie) se trouvent sur la même ligne ou colonne du réseau routier.
    \item Destination aléatoire: les véhicules sélectionnent uniformément leur destination parmi les $N$ points d'extrémité de route du côté opposé du réseau routier.
\end{enumerate}
Après avoir atteint leur destination, les véhicules quittent le domaine de simulation.\\

Kchiche et al. \cite{K14} fournit deux types de déploiement, le premier basé sur la centralité visant une optimisation globale de l'accès aux services et un deuxième déploiement basé sur des points centraux visant l'équité. Toutefois, l'article propose un modèle et une estimation de la moyenne des données envoyées dans les voyages considérés. Le papier considère deux objectifs:  ($i$) maximisation de la moyenne des données envoyées avec succès et ($ii$) la maximisation du minimum de données envoyées avec succès entre les différents voyages par souci d'équité. En effet, deux schémas de déploiement ont été proposés l'un basé sur la centralité et l'autre sur des points centraux qui visent l'équité entre les véhicules.\\
Dans l'article de Lee et al. \cite{K10}, l'auteur décrit son processus suivant les étapes suivantes: la méthode commence par une sélection initiale selon laquelle chaque intersection est un candidat. Pour chaque cercle entourant la position candidate avec le rayon égal à la plage de transmission, le nombre de rapports de véhicules à l'intérieur du cercle est compté. Après avoir ordonné les candidats par décompte, le schéma de placement fait survivre le candidat lorsqu'il est séparé de tous les autres candidats déjà sélectionnés par au moins les critères de distance. Les données d'historique de localisation, obtenues à partir d'un système de suivi de véhicule en temps réel, sont utilisées pour concevoir et évaluer un schéma de placement des \textit{RSU}s.\\
\subsection{Méthodes basées sur la couverture temporelle}
Kaur et al. \cite{H09} propose une méthode de déploiement des \textit{RSU}s basé sur la minimisation du temps parallèle nécessaire pour placer les unités en bordure de route dans une zone donnée et pour atteindre un taux de couverture élevé. Dans son papier, l'auteur décrit les principales étapes adoptées lors de la mise en oeuvre du processus:
\begin{itemize}
    \item étape 1: Définit le nombre des \textit{RSU}s, les dimensions de la zone à couvrir. Placer la première unité les autres vont être placées arbitrairement.
    \item étape 2: Divise l'ensemble d'entrée en deux parties.
    \item étape 3: Appel pour chaque sous-ensemble, divisé dans l'étape précédente, l'algorithme optimiste de déploiement des \textit{RSU}s.
    \item étape 4: Si le résultat n'est pas optimal, refaire l'étape numéro 2, jusqu'à ce que le déploiement optimiste soit atteint.
\end{itemize}
Dans son article, Sun et al. \cite{H10} posent le problème de déploiement des \textit{RSU}s dans un cadre spécial. Ainsi, l'auteur cherche à trouver le nombre et les emplacements des \textit{RSU}s pour mieux couvrir le réseau tout en essayant d'assurer la confidentialité des données des utilisateurs. En effet, l'article propose des contraintes temporelles à savoir la durée de conduire, durée de validation du certificat, etc.\\
Dans son papier Wu et al. \cite{H11} proposent une analyse sur les temps de transmissions des messages entre les véhicules E2E (end to end) et traitent ce problème sur des modèles de mouvement à sens unique et à double sens pour évaluer les temps d'attentes et les temps de latences dans chaque modèle de mouvement. En outre, les auteurs proposent un modèle de déploiement qui utilise les véhicules comme étant des points d'accès mobiles (\textit{RSU}s mobiles) dans certains cas et certaines conditions. Cependant, le papier décrit les modèles d'une manière formelle sans comparaison avec aucune autre méthode. Les expérimentations et les résultats sont déduits à partir des formules mathématiques énoncés dans ce même papier en changeant les quelques paramètres à savoir vitesse, temps d'envoi, distance de propagation, vitesse de propagation, etc.\\

Puisque notre cadre comprend trois méthodes de couverture spatio-temporelle, nous nous concentrons sur le reste de cette section uniquement sur les méthodes pionnières appartenant à cette catégorie. En effet, l'idée générale derrière les méthodes spatio-temporelles existantes est de trouver les modèles de mobilité des véhicules circulant dans une ville donnée, puis de calculer les emplacements des \textit{RSU}s appropriés couvrant tous les modèles de mobilité extraits. Par conséquent, trouver les modèles de mobilité pertinents est une étape primordiale. À cet égard, certaines méthodes s'appuyaient sur la théorie des flux de trafic \cite{B02} pour définir les modèles de mobilité des véhicules. Il convient de mentionner que la théorie des flux de trafic définit trois descriptions de mobilité \cite{B02}:
\begin{enumerate}
      \item \textbf{Description microscopique }: tous les véhicules sont identifiés individuellement. La position et la vitesse de chaque véhicule définissent l'état du système comme des variables dépendantes de la variable de temps ;
      \item \textbf{Description cinétique}: l'état du système est toujours identifié par la position et la vitesse des véhicules. Cependant, leur identification ne se réfère pas à chaque véhicule mais à une distribution de probabilité appropriée ;
      \item \textbf{Description macroscopique}: l'état est décrit par des quantités moyennées localement, c'est-à-dire la densité, la vitesse de masse et l'énergie, considérées comme des variables dépendantes du temps et de l'espace.
\end{enumerate}

\subsection{Méthodes basées sur la couverture Spatio-temporelle}
\subsubsection{Description de l'algorithme RoadGate} Dans son article \cite{B13}, les auteurs  étudient le problème de déploiement des \textit{RSU}s qui sert à fournir les performances de connectivité souhaitées tout en minimisant le nombre de \textit{RSU} déployées. L'idée est d'exploiter le modèle de mobilité stable dans le temps pour trouver les emplacements de déploiement optimaux d'une trace de véhicule réaliste, observons le modèle de mobilité et proposons un modèle de graphe pour le caractériser. Sur la base du modèle de graphe, le but est de transformer le problème de déploiement de la passerelle en un problème de sélection de vertex dans un graphe. L’article réduit le problème de couverture minimale des sommets et propose un algorithme heuristique nommé RoadGate (voir algorithme \ref{roadgate}) qui permet de faire une recherche de type glouton pour trouver les positions optimales.\\

\begin{algorithm}
\BlankLine
  \KwIn{$\mathcal{G}$: Mobility Graph,$\gamma$:longueur maximal du trajet,$\mathcal{P}_u$: probabilité de rencontre (fixé par l'utilisateur)}
  \KwOut{$\mathcal{\Bar{V}}$: Ensemble de résultat.}
  
  $\mathcal{M}$: Matrice transition de G.\\
  **calcul de 1-étape, 2-étape, ..., $\gamma$-étape de matrice transition $\mathcal{M}$.\\
  calcul de VPP $\pi^{\gamma}_{ij}$ entre tous les nœuds selon l'étape précédente (**).\\
  $\mathcal{\Bar{V}} = \emptyset$. \\
  $\mathcal{\theta} = \emptyset$. \\
 \BlankLine 
  \While{$|\mathcal{\theta}|$ < $|\mathcal{V}(G)|$}
        {
        \BlankLine
         $j $= $argMax_{x \in \mathcal{V} - \mathcal{\Bar{V}}}(|\{ k | k \in \mathcal{V} - \mathcal{\theta}, \lambda \geq P_u\}|) $\\
         InsertInV($j$, $\mathcal{\Bar{V}}$)\\
         $\mathcal{\theta} = \mathcal{\theta} \cup \{k | \lambda_{k}\mathcal{\Bar{V}} \geq P_u\}$
       
        }
    
  return ($\mathcal{V}$)

\caption{Algorithme RoadGate}
\label{roadgate}
\end{algorithm}
\subsubsection{Description de l'algorithme MPC } 
Dans cet article\cite{B1}, Yeferny et al., décrivent leurs méthode comme étant une approche spatio-temporelle pour les applications non sécuritaires. La méthode cherche à maximiser le taux de couverture tout en minimisant le nombre des \textit{RSU}s déployés afin de baisser le coût de déploiement. Les contributions énoncées par les auteurs sont:
\begin{itemize}
    \item L'utilisation d'une méthode d'extraction de clusters, à savoir l'utilisation de la théorie de l'analyse du concept formel FCA \cite{FCA} (Formal Concept Analysis), à partir des données de mobilités pertinentes.
    \item Calcul des emplacements des \textit{RSU}s pour couvrir les modèles de mobilité extrait dans la partie 1. L'auteur à utiliser la notion des traverses minimales des hypergraphes afin de couvrir les modèles de mobilité tout en minimisant le nombre des \textit{RSU}s déployés.
\end{itemize}
La méthode sera décrite en quatre étapes. Dans la suite, nous décrivons les étapes de l'algorithme MPC.
\begin{enumerate}
    \item Étape 1: consiste à lire les emplacements des véhicules à partir d'un fichier plat, la description des mouvements des véhicules sera représentée par les positions de chaque véhicule dans le map.
    \item Étape 2: consiste à construire le contexte d'extraction à partir des données traitées dans l'étape 1.
    \item Étape 3: consiste à calculer l'ensemble minimal $\mathcal{F}$ du concept formel en utilisant un algorithme nommé QualityCover.
    \item Étape 4: consiste à couvrir d'une manière minimale l'ensemble $\mathcal{F}$ tout en calculant les traverses minimales de ce dernier ensemble.
\end{enumerate}

\subsubsection{Description de l'algorithme ECSA} Dans son article \cite{B12}, les auteurs décrivent l'algorithme comme étant de type Glouton pour la recherche combinatoire exhaustive qui fournit une borne inférieur du rapport d'approximation. L'algorithme essaie de nombreux cycles dont chacun peut produire un schéma de déploiement et sélectionne le meilleur schéma de déploiement parmi tous les schémas possibles en termes de couverture spatio-temporelle. L'algorithme commence par une première itération avec un sous-ensemble $G$ de localisation candidates, et $n$ est le nombre d'éléments de $G$, $|G| = n$, à chaque itération, l'algorithme glouton insère des localisations candidates parmi les emplacements déployés selon une équation de déploiement. L'algorithme fournit une couverture de détection supérieur à $1-e^{-1}=63\%$ du résultat optimal lorsque la taille de l'ensemble initial est nulle, c'est-à-dire $n = 0$. Dans le cas de coûts arbitraires pour les emplacements candidats, un rapport d'approximation de $1-e^{-1}$ peut encore être obtenu lorsque la taille de l'ensemble initial est supérieure ou égale à trois $n \geq 3$.\\
\begin{algorithm}

\BlankLine
  \KwIn{$\mathcal{R}$,$w_0$,$\mathcal{C}$,$n$}
  \KwOut{$\mathcal{S}$: L'ensemble des emplacements de déploiement.}
  \BlankLine
 $\mathcal{R'}\leftarrow \mathcal{R}$, cost = 0
  
  \ForEach{$G$ of $\mathcal{R}$}
        {
        \BlankLine
         $\mathcal{R^{'}}\leftarrow \frac{\mathcal{R}}{G} $
         cost = w(G) \\
 // sélectionner l'emplacement qui maximise le ratio du gain marginal \\
 Max = 0, $\lambda$=null \\
 \ForEach{$\lambda \in \mathcal{R^{'}}$}
 {
 $MGR(\omega^{'}_1, \omega_1) = \frac{cp(\omega^{'}_1) - cp(\omega_1)}{C_\lambda}$ \\
 \If{Max < MGR($\lambda$)}
 {Max = MGR($\omega^{'}_1, \omega_1$)}
 
 {$\omega_1 \leftarrow \omega_1 \cup w$}

 }
      $Max = 0$, $\lambda = null$ 
        }
\ForEach{sous ensemble de $G$, $|G|< n$ $w(G)< w_0$}
        {
        \BlankLine
         $\omega_2$ $\leftarrow G$

        }

\caption{Algorithme ECSA}
\label{ecsa}
\end{algorithm}
\\
En outre, l'application de la théorie des flux de trafic est très difficile dans un cadre réaliste et elle ne parvient pas dans la plupart des cas à trouver des modèles de mobilité efficaces. Par conséquent, d'autres méthodes de couverture étudiées ont exploité les modèles de mobilité à partir des trajectoires des véhicules. Xiong et al. \cite{B13} introduit un nouveau modèle de graphe pour caractériser le flux de trafic. Le modèle de mobilité est détecté en divisant le système routier en zones uniformes (forme et surface) puis en calculant la probabilité de transition temporelle entre toutes les zones. Sur la base du modèle de graphe, ils ont transformé le problème de déploiement des \textit{RSU}s en un problème de sélection de sommets dans un graphe. En le réduisant au problème de couverture minimale des sommets, ils ont montré que le problème de déploiement des \textit{RSU}s est un problème NP-complet. Ensuite, ils ont proposé un algorithme heuristique appelé RoadGate qui cherche les positions optimales afin de mettre les \textit{RSU}s. Zhu et al. \cite{B12} a introduit une nouvelle méthode de détection des mouvements futurs de véhicules basés sur la chaîne de Markov. À cet égard, ils ont extrait les schémas de mobilité des véhicules à partir des fichiers journaux et ont, ensuite, prédit leurs déplacements futurs. Par ailleurs, les auteurs ont introduit un algorithme glouton  pour calculer les emplacements appropriés des \textit{RSU}s. Trullols et al. \cite{FB12} propose une approche de téléchargement collaboratif de données entre véhicules. Ils ont suggéré de placer des \textit{RSU}s dans les zones avec volume de passage important afin que les véhicules jouent le rôle de relais. À cette fin, ils ont introduit une nouvelle topologie de flux de trafic en modélisant la carte routière sous forme d'un graphe, où les routes et les intersections représentent respectivement des sommets et des arêtes. Yeferny et al \cite{B1} ont introduit une nouvelle méthode pour représenter les modèles de mobilité des véhicules à partir de leurs fichiers de trace basée sur la théorie de l'analyse de concept formel (FCA). Ensuite, ils ont calculé les emplacements de \textit{RSU}s adéquats pour couvrir les modèles de mobilité extraits par le nombre minimal de \textit{RSU}. Ils ont formulé ce dernier problème comme étant un problème d'extraction des traverses minimales d'un hypergraphe \cite{J16}. Dans leur article, les auteurs montrent que la stratégie proposée a réduit le coût de déploiement et obtenu de bons résultats en termes de taux de couverture, latence et overhead. Cependant, la technique FCA adaptée utilisée pour calculer les modèles de mobilité était basée sur une base de données transactionnelle représentant les relations entre les trajectoires des véhicules et les croisements traversés. Cette représentation considère la trajectoire d'un véhicule comme un ensemble non ordonné de jonctions croisées. Par conséquent, ils ignorent un facteur clé, qui est l'ordre qu'un véhicule traverse un croisement. En effet, un nombre supplémentaire de modèles de mobilité pourrait être extrait, conduisant à une augmentation du nombre des \textit{RSU}s utilisées. 

\subsection{Tableau comparatif des différentes approches}
Afin de mettre en lumière les avantages et les inconvénients des approches les plus utilisées de la littérature, nous surlignons dans le Tableau \ref{tab-comparatif} les différents critères qu'un système de déploiement des \textit{RSU}s doit satisfaire pour les différentes approches de la littérature. La comparaison est établie selon les propriétés suivantes: 
\begin{itemize}
    \item Type: indique la famille de la méthode et le fondement mathématique sur lequel elle est basée.
    \item Taux de Couverture: indique le taux de couverture fournit par la méthode dans des différents types de maps (réel ou simulation dense ou éparse).
    \item contrainte budgétaire : indique si la méthode utilise une contrainte budgétaire afin de fixer d'avance le nombre de \textit{RSU} maximum à utiliser.
    \item Adaptivité: indique si la méthode s'adapte aux évolutions des patterns de mobilités en temps réel.
    \item Scalabilité: indique si l'approche arrive à réussir le passage à l'échelle.
\end{itemize}
\subsection{Critiques et contributions}
Dans ce qui suit, nous présentons une étude comparative des méthodes de déploiement des \textit{RSU}s basée sur une couverture spatio-temporelle. Notre étude sera limitée aux trois approches MPC, RoadGate et ECSA et ce pour deux raisons: 
\begin{itemize}
 \item Ces trois approches peuvent être comparées avec nos approches puisqu'elles ont les mêmes jeux de données en entrée (les modèles de mouvement) et fournissent le même ensemble de sortie en termes de type (nombre de \textit{RSU} et leurs emplacements).
 \item Ces trois approches utilisent des concepts mathématiques différents ce qui les rendent un échantillon de test élargie. Dans la littérature, plusieurs méthodes peuvent appartenir à la famille des approches spatio-temporelle à savoir \cite{benr16} \cite{aslam12} \cite{Lochert08} cependant toutes ces approches sont basé sur les mêmes concepts mathématiques utilisé par les méthodes MPC, RoadGate et ECSA ce qui argumente notre choix pour les approches avec lesquelles nous faisons les comparaisons expérimentales. 
\end{itemize}

\begin{landscape} 
\begin{table}[ht]
 \centering
 \begin{center}
\begin{tabular}{|p{5cm}||p{5cm}|p{5cm}|p{5cm}|}
\hline
&RoadGate&MPC&ECSA\\
\hline
\hline
Type de données &Graphe &Transactionnelle & Matricielle\\
\hline
Méthode utilisée & couverture par sommets& traverse minimale & Maximal delay constraint coverage\\
\hline 
Scalabilité (véhicule/$Km^2$) & $260 000 v/250Km^2$& $1000 véhicules /2.5 Km^2$&50-2000 véhicule / 2000 bus et taxi\\
\hline
jeu de donnée (réel ou simulation)& Simulation (MMTS) & Simulation (SUMO) & Shangai, chine \\
\hline
Domaine d'utilisation &Non safety application
& Non safety application
 & Non safety application\\

\hline
\end{tabular}
\caption{Comparaison des méthodes de couverture spatio-temporelle de la littérature.}\label{tab-comparatif}
\end{center}
\end{table}

\end{landscape} 
\section{Conclusion}
Dans ce chapitre, nous avons présenté le système de transport intelligents STI ainsi que les réseaux sans fil utilisés, les infrastructures réseaux déployés et les types de communications autorisés dans ce type de système. De surcroît, nous avons décrit les limites et les problèmes de déploiement d'une couverture réseau stable dans un système de transport non centralisé avec la classification des méthodes existantes dans la littérature et la présentation de chacune.\\
Dans le chapitre suivant, nous présenterons notre première méthode de déploiement des \textit{RSU}s basée sur des données séquentielles nommée  \textrm{\emph{SPaCov}} ainsi que sa version améliorée nommée  \textrm{\emph{SPaCov+}}.
\newpage
\thispagestyle{empty}
\mbox{}
\newpage

\chapter{SPaCov: fouiller pour mieux placer}
\minitoc
\chaptermark{SPaCov}

\section{Introduction}
Dans ce chapitre, nous présentons notre première approche spatio-temporelle qui permet d'exploiter les modèles de mobilité à partir des trajectoires du mouvements des véhicules. A cet égard, nous décrivons une nouvelle méthode de représentation des modèles de mobilité de véhicules basé sur la théorie des patterns séquentiels en data mining. Dans ce qui suit, nous rappelons et adaptons certaines notions mathématiques de patterns séquentiels.\\
La méthode  \textrm{\emph{SPaCov+}} est une amélioration de la méthode basique  \textrm{\emph{SPaCov}}. Dans le reste de ce chapitre, nous détaillons davantage ces deux méthodes.
\section{Fondements mathématiques}
\begin{definition}\textsc{\textbf{\textsc{(}Base de données séquentielle \textsc{)}}}\mbox{}\newline
Une base de données séquentielle est un ensemble ordonné de séquence. Dans notre cas, la séquence est une liste ordonnée de croisement (ou junction) traversé par un véhicule, noté $< j_{1} j_{2} \ldots j_{k} >$.
\end{definition}

\begin{table}[ht]
 \centering
 \begin{center}
\begin{tabular}{|p{2cm} |>{\centering\arraybackslash}p{4cm}c|}
\hline
$\mathcal{V}$&$\mathcal{S}$&\\
\hline
\hline
$v_{1}$&$j_{2}j_{6}j_{7}j_{1}$ &\\

$v_{2}$& $j_{3}j_{6}j_{4}$&\\

$v_{3}$&$j_{5}j_{3}j_{7}$&\\

$v_{4}$& $j_{6}j_{5}$&\\

$v_{5}$& $j_{2}j_{6}j_{7}$&\\

$v_{6}$&$j_{3}j_{6}$&\\

$v_{7}$&$j_{6}j_{5}j_{3}j_{7}$&\\

$v_{8}$&$j_{6}j_{3}$&\\

\hline
\end{tabular}
\caption{Base de données séquentielle $\mathcal{D}$}\label{cxt2}
\end{center}
\end{table}

\begin{example}
Le Tableau \ref{cxt2} illustre une base de données séquentielle $\mathcal{D}$, où $\mathcal{V} = \{v_{1}, v_{2}, v_{3}, v_{4}, v_{5}, v_{6},v_{7}, v_{8}\}$ est un ensemble de trajectoires des véhicules et $\mathcal{S}$  représente une séquence de croisements dont ces derniers véhicules les traversent.
\end{example}

\begin{definition}\textsc{\textbf{\textsc{(}support d'une séquence \textsc{)}}}\mbox{}\newline\label{defsupp}
Soit $s$  est une séquence de croisements. Le support de $s$, noté $supp(s)$, est une indication  de la fréquence à laquelle la séquence $s$ apparaît dans $\mathcal{D}$. Il est calculé comme suit:
\[ supp(s)=\frac{|\{s_i \in \mathcal{D} | s \lesssim s_i\}|}{|\mathcal{D}|}\]
Le symbole $\lesssim$ représente une inclusion ordonnée (ordre des éléments dans les séquences), par exemple $(j_{3}j_{4})$ $\lesssim$ $(j_{2}\mathbf{j_{3}}j_{5}j_{6}\mathbf{j_{4}})$ par contre $(j_{3}j_{4})$ $\not \lesssim$. $(j_{2}\mathbf{j_{4}}j_{5}j_{6}\mathbf{j_{3}})$.
\end{definition}

\begin{example}
\[ 
supp((j_{6}j_{7}))=\frac{|\{v_1,v_5,v_7\}|}{|\mathcal{D}|} = \frac{3}{8} = 42\%
\]
\end{example}

\begin{definition}\textsc{\textbf{\textsc{(} Séquence fréquente \textsc{)}}}\mbox{}\newline\label{SFr}
La séquence $s$ est dite fréquente si le support de $s$ dans $\mathcal{D}$ est supérieur ou égal à un \emph{seuil spécifié par l'utilisateur} noté  $minsup$. Formellement, nous définissons l'ensemble $\mathcal{FS}$ de séquences fréquentes dans $\mathcal{D}$ comme:
\[ \mathcal{FS} = \{s \in \mathcal{D} | supp(s) \geqslant minsup\}\]\label{FSDEF}.
\end{definition}

\begin{example}
Considérant un seuil $minsup=2/8$ et la base de données séquentielle $\mathcal{D}$ de l'exemple \ref{cxt2}, alors la séquence $s=<j_{2}j_{6}j_{7}>$ est fréquente  puisque son support $supp(s)=2/8  \geqslant minsup$.
\end{example}

\begin{definition}\label{MFSDEF}
\textsc{\textbf{\textsc{(}Séquence fréquente maximale \textsc{)}}}\mbox{}\newline
Soit l'ensemble des séquences fréquentes $\mathcal{FS}$. Formellement, nous définissons l'ensemble $\mathcal{MFS}$ des séquences fréquentes maximales dans $\mathcal{D}$ comme suit: 
\[ \mathcal{MFS} = \{s \in \mathcal{FS} | \nexists s^{'} \in \mathcal{FS}, s \lesssim s^{'} \} \]
\end{definition}

\begin{example}
Pour un seuil support $minsup=2/8$, la séquence $s=<j_{5}j_{3}j_{7}>$ est considéré comme une séquence fréquente maximale dans $\mathcal{D}$ puisque $supp(s)=2/8 \geqslant minsup$  \textit{et} tous ces sur-ensembles sont infréquents. Cependant, la séquence $s1=<j_{5}j_{3}>$ de support égal à $2/8$, elle est fréquente mais elle ne peut pas être considérée comme une séquence maximale fréquente puisqu'elle existe une séquence incluant $s1$ est fréquente qui est $s$.
\end{example}

\begin{definition}\textsc{\textbf{\textsc{(}Séquence infréquente ou Séquence rare\textsc{)}}}\mbox{}\newline\label{SRa}
La séquence $s$ est dite  infréquente (ou rare) si le support de $s$ dans $\mathcal{D} $ est strictement inférieur à $minsup$. Formellement, nous définissons l'ensemble $\mathcal{RS}$ des séquences  infréquentes dans $\mathcal{D}$ comme:
\[ \mathcal{RS} = \{s \in \mathcal{D} | supp(s) < minsup\} \]\label{RSDEF}.
\end{definition}
\begin{example}
Considérant un seuil $minsup = 2/8$ et la base de données séquentielle $\mathcal{D}$ de l'exemple \ref{cxt2}, la séquence $s=<j_{6}j_{5}j_{7}>$ est infréquente puisque son support $supp(s) = 1/8 < minsup$. 
\end{example}

\begin{definition}\textsc{\textbf{\textsc{(}Séquence rare minimale \textsc{)}}}\mbox{}\newline\label{MRSDEF}
Soit $\mathcal{RS}$ l'ensemble des séquences infréquentes dans $\mathcal{D}$. Formellement, nous définissons l'ensemble $\mathcal{MRS}$ des séquences rares minimales dans $\mathcal{D}$ comme:

\[ \mathcal{MRS} = \{s \in RS | \nexists s^{'} \in RS, s \gtrsim s^{'} \} \]
\end{definition}

\begin{example}
Pour un $minsup=2/8$, la séquence $s=<j_{6}j_{3}j_{7}>$ est considérée comme une séquence rare minimale (ou infréquente) dans $\mathcal{D}$ puisque: 
\begin{itemize}
    \item $supp(s)=1/8 < minsup$ \\
    \textit{et}
    \item toutes ses sous-séquences sont fréquentes, c'est à dire $<j_{6}j_{3}>$, $<j_{6}j_{7}>$ and $<j_{3}j_{7}>$.
\end{itemize}
     
\end{example}
Dans ce qui suit, nous commençons par décrire la première méthode que nous introduisons.

\section{Description de SPaCov}
L'objectif principal de l'approche \textrm{\emph{SPaCov}} est d'assurer une couverture élevée en utilisant le plus petit nombre possible de \textit{RSU}. Il prend en entrée une base de données séquentielle   $\mathcal {D} $ représentant un grand nombre de trajectoires de véhicules et fournit en sortie les emplacements des \textit{RSU}s appropriés. Cette méthode comporte deux modules, à savoir l'extraction de modèles de mobilité maximale (MMP) et le calcul de la couverture (CoC). Le premier module extrait un ensemble de modèles de mobilité maximale à partir d'une base de données séquentielle. Le deuxième module, calcule les emplacements des \textit{RSU}s appropriés qui couvrent tous les modèles de mobilité extraits. Il est à noter que, dans la littérature, un modèle de mobilité qui représente une séquence ordonnée de croisement traversée par des véhicules n'existe pas. Dans ce qui suit, nous détaillons les composants MMP et CoC.

\subsection{Les patterns de mobilité maximaux (MMP)}
Nous rappelons que le but de la méthode  \textrm{\emph{SPaCov}} est d'assurer une couverture élevée en utilisant le plus petit nombre possible de \textit{RSU}. Pour atteindre cet objectif, nous devons extraire le plus petit ensemble représentatif des modèles de mobilité fréquents tout en choisissant les meilleurs emplacements des \textit{RSU}s qui maximise le taux de couverture. En effet, nous recherchons les séquences fréquentes les plus longues. L'ensemble défini dans la définition \ref{MFSDEF}, comme l'ensemble des séquences fréquentes maximales. En plus, ils doivent être des modèles de mobilité maximale (c'est-à-dire chaque modèle de mobilité fréquente ne doit pas être inclus dans un autre modèle fréquent). À cette fin, nous nous appuyons tout d'abord sur l'algorithme prefixSpan \cite{prefixspan} pour extraire les modèles de mobilité fréquents (c'est-à-dire les séquences fréquentes des croisements) des trajectoires des véhicules. De plus, nous définissons l'algorithme getMaximalFrequent() (c.f. algorithme \ref {getMaximalFrequent}), qui ne sélectionne que les modèles maximaux parmi les plus fréquents.
\begin{algorithm}
\BlankLine
  \KwIn{$\mathcal{FS}$: sorted set of frequent sequences. $FS$ is sorted with respect to the size of sequences}
  \KwOut{$\mathcal{MFS}$: set of maximal frequent sequences.}
  \BlankLine
  \ForEach{$E$ of $\mathcal{FS}$}
        {
        \BlankLine
         deleteSubSequences($E$, $\mathcal{MFS}$)\\
         insert($E$, $\mathcal{MFS}$)
        }
  return ($\mathcal{MFS}$)

\caption{getMaximalFrequent() }
\label{getMaximalFrequent}
\end{algorithm}

\begin{example} \label{exp2}
Considérant un minsupp $\textit{minsup}= 2/8$ et la base de données séquentielle $\mathcal{D}$ donné dans l'exemple \ref{cxt2}. Pour extraire les modèles de mobilités fréquentes, le composant MMP procède comme suit:

\begin{enumerate}
  \item \textbf{Étape 1:} invoque l'algorithme prefixSpan, qui extrait l'ensemble ordonné des séquences fréquentes:
     \[
\begin{tabular}{l}
 $\mathcal{FS} =\{j_{2},j_{3},j_{5}, j_{6},j_{7}, <j_{2}j_{6}>,<j_{2}j_{7}>,<j_{3}j_{6}>,<j_{3}j_{7}>,<j_{5}j_{7}>,$\\
 $<j_{5}j_{3}>,<j_{6}j_{7}>,<j_{6}j_{5}>,<j_{6}j_{3}>,<j_{2}j_{6}j_{7}>,<j_{5}j_{3}j_{7}>\}$
\end{tabular}     
   \]
  \item \textbf{Étape 2:} invoque l'algorithme getMaximalFrequent, qui permet d'extraire l'ensemble des séquences fréquentes maximales à partir de l'ensemble des séquences fréquentes:
   \[
      \mathcal{MFS}=\{<j_{6}j_{5}>,<j_{3}j_{6}>,<j_{6}j_{3}>,<j_{2}j_{6}j_{7}>,<j_{5}j_{3}j_{7}>\}
   \]
\end{enumerate}
\end{example}

\subsection{Calcul de couverture (CoC)}
Le composant CoC vise à calculer un ensemble minimal de croisement, noté $ CR $, qui couvre l'ensemble $ \mathcal {MFS} $ de séquences fréquentes maximales qui représente les modèle de mobilité des véhicules. Par conséquent, le  placement des \textit{RSU}s, dans des croisements qui couvre ces modèles, garantit que les véhicules entreront en contact avec au moins un \textit{RSU} pendant leurs déplacements. En effet, calculer le $ CR $ revient à trouver une traverse minimale qui couvre où les sommets d'un hypergraphe peuvent être vu comme des croisements. Dans ce qui suit, nous rappelons brièvement quelques notions de base de la théorie des graphes et nous montrons comment nous adaptons cette théorie pour calculer l'ensemble $ CR $.

\begin{definition}\textsc{\textbf{\textsc{(}Hypergraph\textsc{)}}}\mbox{}\newline
Un hypergraphe $\mathcal{H}$ est un pair ordonné  $\mathcal{H} = (\mathcal{J} , \mathcal{E})$, où $\mathcal{J} =\{j_{1},\ldots,j_{n}\}$ est un ensemble fini d'éléments
et $\mathcal{E}=\{E_{1},\ldots,E_{m}\}$ est une famille des sous éléments de $\mathcal{J}$ tel que  \cite{T14}:
\begin{enumerate}
\item $E_{i} \neq \emptyset$ $(i=1,\ldots,m)$ et
\item $\bigcup_{i=m}^{i=1} E_{i} = \mathcal{J}$.
\end{enumerate}
Les éléments de $\mathcal{J}$ sont appelés sommets (vertices), tandis que les éléments de $\mathcal{E}$ sont appelés arrêtes (hyperedges) de l'hypergraphe $\mathcal{H}$\cite{T11} \label{hypergraph-def}.
\end{definition}
\begin{example}
En respectant la Définition \ref{hypergraph-def}, nous pouvons transformer l'ensemble $\mathcal{MFS}$ des séquences fréquentes maximales sous forme d'un hypergraphe $\mathcal{H} = (\mathcal{J} , \mathcal{E})$, tel que:
\begin{enumerate}
    \item $ \mathcal{J} = \bigcup_{S_{i} \in \mathcal{MFS}} S_{i}$
          
    \item $ \mathcal{E}=\mathcal{MFS}$
          
\end{enumerate}
La Figure \ref{hypergraph} illustre l'hypergraphe $\mathcal{H}$ associé à l'ensemble des séquences fréquentes maximales $\mathcal{MFS}$ de l'exemple \ref{exp2}. 
\begin{figure}[ht]
  \begin{center}
  \includegraphics[scale=0.5]{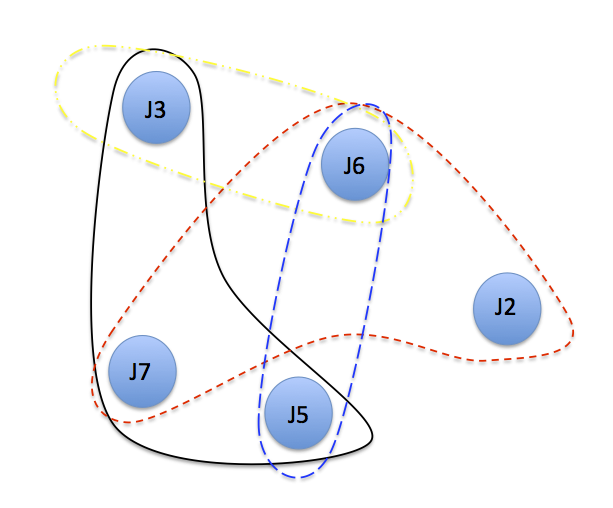}
  \caption{Hypergraphe représentant l'ensemble des séquences fréquentes maximales  de l'exemple \ref{exp2} } \label{hypergraph}
\end{center}
\end{figure}
\end{example}
\begin{definition}\textsc{\textbf{\textsc{(}Traverse minimale \textsc{)}}}\mbox{}\newline\label{trans}
Soit $\mathcal{H} = (\mathcal{J},\mathcal{E})$ est un hypergraphe. Un ensemble $T \subseteq \mathcal{J}$ est appelé une traverse (alias un ensemble de frappe) de $\mathcal{H}$ s'il croise tous ses arrêtes (hyperedges), c.à.d., $T \cap E_{i} \neq \emptyset $ $\forall E_{i} \in \mathcal{E}$. Une traverse T est dit minimale lorsque aucun sous-ensemble $T^{'}$ de $T$ n'est une traverse de $\mathcal{H}$. Dans le reste, $TrMin_{H}$ représente l'ensemble des traverses minimales de $\mathcal{H}$ \cite{T14}.
\end{definition}

\begin{example} \label{mt}
D'après l'hypergraphe de la Figure \ref{hypergraph} représentant l'ensemble $\mathcal{MFS}$ de l'exemple \ref{exp2}, nous extrayons l'ensemble des traverses minimales suivant: 
 \begin{center}
\begin{tabular}{c|c|c}
     (1) $TrMin_{1}=\{j_6, j_3 \}$ & (2) $TrMin_{2}=\{ j_6, j_5\}$  & (3)   $TrMin_{3}=\{ j_6, j_7\}$
     
\end{tabular}
\end{center}
\end{example}

Il est maintenant clair que le calcul de l'ensemble minimal de croisement $ CR $ qui couvre l'ensemble des séquences fréquents maximales $\mathcal{MFS} $ se résume au calcul d'une traverse minimale de l'hypergrahe $\mathcal{H}$ associée à l'ensemble $\mathcal{MFS}$. Dans la littérature, plusieurs algorithmes d'extraction des traverses minimales des hypergraphes ont été proposés. Dans notre cas, nous utilisons l'algorithme introduit par Murakami et al. \cite{T11} pour extraire l'ensemble des traverses minimales de l'hypergraphe associé à l'ensemble $\mathcal{MFS}$. A Chaque fois où l'algorithme fourni plus qu'une traverse minimale, nous choisissons le plus petit en terme de cardinalité comme $ CR $.
\begin{figure}[!ht]
\centering
\includegraphics[scale=0.6]{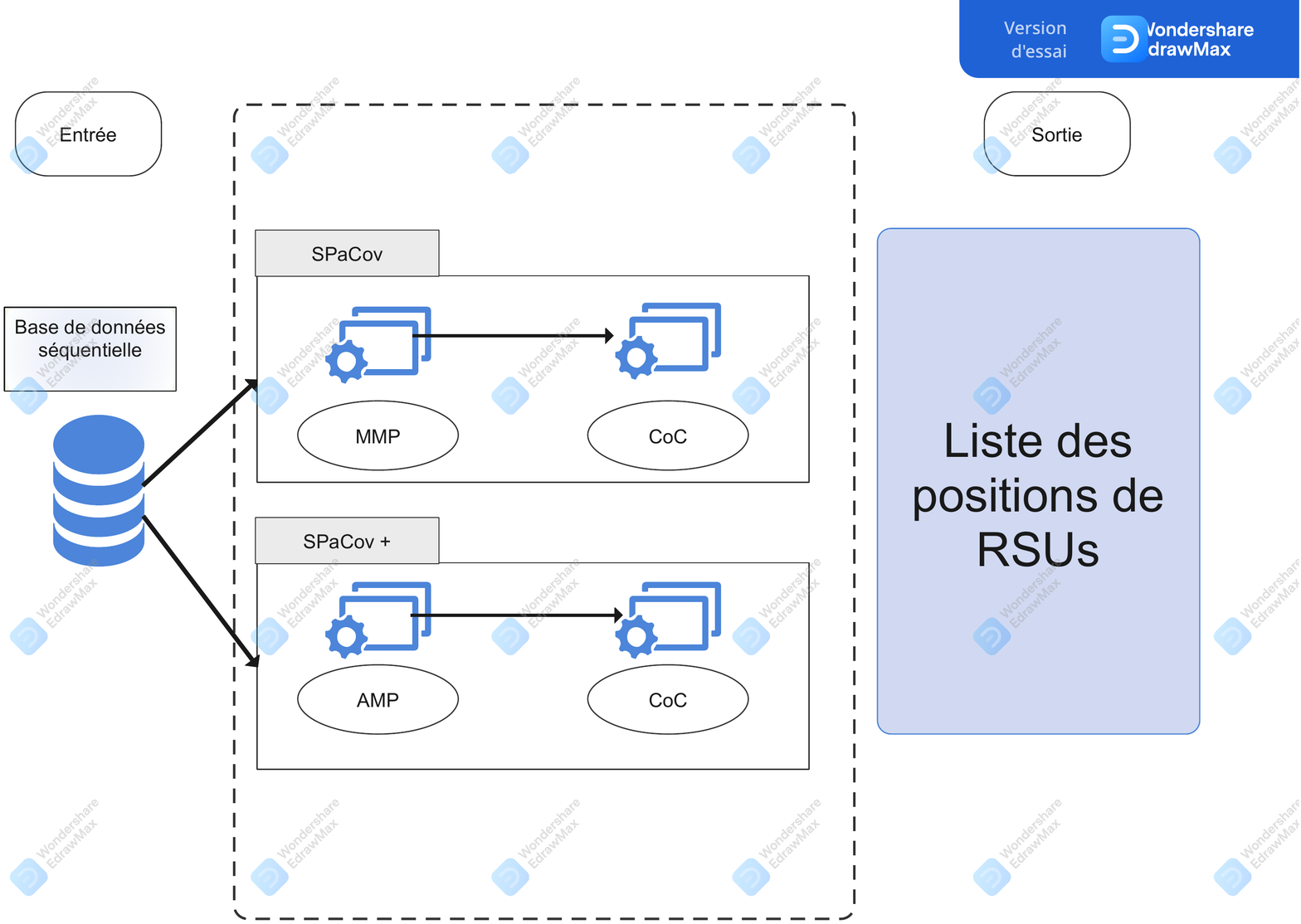}
\caption{L'architecture des méthodes  \textrm{\emph{SPaCov}} et  \textrm{\emph{SPaCov+}}.}
\end{figure}\label{spacov+}
\section{Description de  \textrm{\emph{SPaCov+}}}
Dans cette section, nous présentons la méthode améliorée de  \textrm{\emph{SPaCov}} qui sera nommée  \textrm{\emph{SPaCov+}}. Le but ultime de  \textrm{\emph{SPaCov+}} est d'atteindre un taux de couverture très élevé (c'est-à-dire plus de couverture que ce qui est fourni par \textrm{\emph{SPaCov}} tout en utilisant un nombre minime de \textit{RSU}. Pour atteindre cet objectif,  \textrm{\emph{SPaCov+}} extrait un ensemble représentatif de tous modèles de mobilité (c'est-à-dire les séquences fréquentes et celles rares). Cette extraction est assuré par le module 'All Mobility Pattern' (AMP) (voir la sous-section \ref{AMPSEC}). Ensuite, il calcule les emplacements des \textit{RSU}s appropriés qui couvrent tous les modèles extraits à l'aide du composant CoC de la méthode  \textrm{\emph{SPaCov}}. Il convient de mentionner que pour amélioré le taux de couverture, il faut généralement un nombre supplémentaire de \textit{RSU}. En effet, nous pouvons accepter ce surcoût pour une nette amélioration du taux de couverture.

\subsection{All Mobility Patterns (AMP)}\label{AMPSEC}
À partir d'une base de données séquentielle, nous pouvons extraire un nombre impressionnant de modèles de mobilité. Par conséquent, couvrir tous ces modèles nécessite également un grand nombre de \textit{RSU}. En outre, Une base de données séquentielle peut être vu comme l'union de deux ensembles $\mathcal{FS}$ (voir Définition \ref{SFr}) et $\mathcal{RS}$ (voir Définition \ref{SRa}) ; le premier contient les séquences fréquentes et l'autre contient les séquences rares. Par Conséquent, le modèle de mouvement des véhicules est inclus dans les deux ensembles $\mathcal{FS}$ et $\mathcal{RS}$ qui représentent un sous ensemble réduit et exacte de la base initiale. En effet, la déduction des modèles de mobilité peut se faire à partir de deux ensembles représentatifs $\mathcal{FS}$ et $\mathcal{RS}$.\\
Pour ces raisons, nous supposons que $\mathcal{FS}$ et $\mathcal{RS}$ forment un sous ensemble représentatif de la totalité des mouvements des véhicules. Comme le montre le module MMP, les modèles de mobilité fréquents pourraient être représentés par les séquences fréquentes maximales (Définition \ref{MFSDEF}). De plus, nous pouvons représenter les modèles de mobilités rares par les séquences rares minimales (Définition \ref{MRSDEF}), qui est un sous-ensemble de l'ensemble des séquences rares contenant les séquences rares les plus courts. Ainsi, l'ensemble représentatif de tous les modèles de mobilité, notez $\mathcal{AP}$, est l'union des deux ensembles $\mathcal{MFS}$ et $\mathcal{MRS}$ respectivement, les séquences fréquentes maximales et les séquences rares minimales. Pour calculer l'ensemble $\mathcal{AP}$, nous procédons comme suit:
\begin{enumerate}
\item \textbf{Étape 1:}  Extrait l'ensemble des séquences fréquentes maximales $\mathcal{MFS}$ à partir de la base de données séquentielle $\mathcal{D}$ par l'exécution du composent MMP; 
\item \textbf{Étape 2:} Extrait l'ensemble des séquences rares minimales $\mathcal{MRS}$ à partir de la base de données séquentielle $\mathcal{D}$. Dans cette partie, nous utilisons l'algorithme New\_SPAM \cite{prefixspan} pour l'extraction des séquences rares à partir de la base de données séquentielles $\mathcal{D}$. Ensuite, nous appelons l'algorithme getMinimalRare (algorithme \ref{getMinimalRare} page ) pour sélectionner les séquences rares minimales parmi celles qui sont rares. 
\item \textbf{Étape 3:} Élaguer les séquences nulles, dont leurs supports égale à zéro à partir de l'ensemble $\mathcal{MRS}$.
\item \textbf{Étape 4:} Élaguer toutes les 1-séquences (séquence de taille 1)  à partir de l'ensemble $\mathcal{MRS}$.
\item \textbf{Étape 5:}  $\mathcal{AP}=\mathcal{MRS} \cup \mathcal{MFS}$.
\end{enumerate}
Les deux étapes d'élagage sont bénéfique du point de vue de temps d'exécution, puisqu'ils vont \underline{minimiser le nombre de candidats} généré à chaque itération de l'algorithme  \underline{sans changer } \underline{le résultat final}.\\
Preuve: 
Les deux étapes 3 et 4 minimisent l'ensemble candidat et prennent de ne pas toucher au résultat final de l'ensemble $\mathcal{AP}$
\begin{itemize}
    \item Minimisation d'élément dans l'ensemble candidat:
     \begin{itemize}  
        \item L'étape 3 élague toutes les séquences nulles ainsi que toutes les sur-séquences d'une séquence nulle (elles le sont forcément) ce qui implique une réduction majeure dans le nombre de séquence générée à chaque itération.
        \item Dans l'étape 4, toutes les 1-séquences seront élaguées ainsi que toutes ces sur-séquences, ce qui traduit par la réduction du nombre de candidat à chaque itération (les sur ensembles d'une 1-séquence sont formé par des séquences nulles ou par des 1-séquences). 
        \item Exemple pire des cas: L'unique cas de figure où les deux étapes 3 et 4 ne minimisent pas le nombre des séquences candidates générées c'est lorsque nous admettons l'existence d'une  base de données séquentielle contenant un seul chemin  $s$ (séquence) composé de $n$ croisements (junction) et $m$ véhicules qui traversent tous le même chemin $s$, dans ce cas les deux étapes 3 et 4 n'affectent pas le nombre de séquences, une telle base de séquences ne contient aucune séquence nulle et aucune 1-séquence (dans le cas réel une telle base de données n'existe jamais).
      \end{itemize}
    \item  Le résultat final est inchangeable:
       \begin{itemize}  
        \item L'étape 3: L'ensemble final $\mathcal{AP}$ contient les traverses minimales de l'union entre l'ensemble $\mathcal{MRS}$ et $\mathcal{MFS}$, d'une manière formelle, les traverses minimales ne peuvent pas contenir des séquences nulles (voir Définition \ref{trans}) donc l'existence ou l'absence des séquences nulles dans $\mathcal{MRS}$ ne change pas l'ensemble $\mathcal{AP}$ puisque ce dernier ne peut pas contenir des 0-séquences. D'une manière informelle, l'existence d'une séquence nulle se traduit par l'existence d'un chemin non visité par aucun véhicule ce qui implique (i) un chemin non connecté (les croisements ne sont pas adjacents) ou bien (ii) un chemin non intéressant (personne ne le suit) dans les deux cas nous n'avons aucun intérêt à couvrir un tel chemin (d'où aucun intérêt à couvrir cette séquence).
        \item L'étape 4:  L'ensemble final $\mathcal{AP}$ contient les traverses minimales de l'union entre l'ensemble $\mathcal{MRS}$ et $\mathcal{MFS}$, d'une manière formelle, et contrairement à l'étape 3, les traverses minimales peuvent contenir des 1-séquences (voir Définition \ref{trans}) mais nous avons élagué ces éléments pour ne pas avoir l’ensemble des \textit{RSU}s égale au nombre des croisements de la carte. En outre, nous supposons un chemin, ensemble de croisement, un déplacement entre au moins deux croisements avec cette condition nous pouvons réduire le nombre des \textit{RSU}s et offrir une couverture spatiale meilleur que celle de  \textrm{\emph{SPaCov}}.
      \end{itemize}
\end{itemize}
\begin{algorithm}
\BlankLine
  \KwIn{$\mathcal{RS}$: Ensemble trié des séquences rares. $\mathcal{RS}$ est trié en fonction de la taille des séquences}
  \KwOut{$\mathcal{MRS}$: ensemble des séquences rares minimales.}
  \BlankLine
  $\mathcal{MRS} = \emptyset$\\
\ForEach{$E$ of $\mathcal{RS}$}
        {
        \BlankLine
         \If{(!existSubset($E$, $\mathcal{MRS}$))}
         {insert($E$, $\mathcal{MRS}$)}
         {}
        }
  return ($\mathcal{MRS}$)

\caption{getMinimalRare() }
\label{getMinimalRare}
\end{algorithm}
\begin{example}\label{expMRS}
En considérant un seuil $\textit{minsup} = 2/8$ et la base de données de séquentielle $\mathcal{D}$ donnée dans l'exemple \ref{cxt2}, pour extraire les modèles de mobilité, le composant AMP se procède comme suit:

\begin{itemize}
  \item \textbf{Étape 1:} il  extrait l'ensemble $\mathcal{MFS}$ suivant:
\[ \mathcal{MFS}=\{<j_{6}j_{5}>,<j_{3}j_{6}>,<j_{6}j_{3}>,<j_{2}j_{6}j_{7}>,<j_{5}j_{3}j_{7}>\} 
\]
  \item \textbf{Étape 2-3-4:} il appelle l'algorithme nommé getMinimalRare, qui extraira l'ensemble des séquences rares minimales. Pour réduire la taille de l'ensemble $\mathcal{MRS}$, le module supprime toutes les séquences nulles ainsi que les séquences de taille 1. L'ensemble final $\mathcal{MRS}$ est le suivant:
  
   \[
      \mathcal{MRS}=\{<j_{3}j_{2}>,<j_{6}j_{5}j_{7}>,<j_{6}j_{5}j_{3}>,<j_{6}j_{3}j_{7}>\} ~~ (après~~élagage)
   \]
   \item \textbf{Étape 5:} il calcule l'ensemble $\mathcal{AP}$ suivant:
   \[
   \mathcal{AP}=\mathcal{MFS} \cup \mathcal{MRS}
   \]
     \[
      \mathcal{AP}=\{<j_{6}j_{5}>,<j_{3}j_{6}>,<j_{6}j_{3}>,<j_{2}j_{6}j_{7}>,<j_{5}j_{3}j_{7}>,<j_{3}j_{2}>,<j_{6}j_{5}j_{7}>,
      \]
  \[
      <j_{6}j_{5}j_{3}>,<j_{6}j_{3}j_{7}>\}
     \]
\end{itemize}
\end{example}
\section{Étude de la complexité et de la validité}
Dans cette section, nous prouvons  la complétude et la terminaison de notre algorithme liée au deux méthodes  \textrm{\emph{SPaCov}} et  \textrm{\emph{SPaCov+}} puis nous calculons les complexités de deux méthodes.
\subsection{Terminaison et complétude}
Le nombre des \textit{RSU}s générés par l'algorithme est fini. En effet, le nombre $k$ des \textit{RSU}s est donné par l'utilisateur est au pire des cas $k$ doit être égal à la cardinalité de $|\mathcal{J}|$. Puisque l'ensemble fourni par l'algorithme est fini, les quatre boucles de l'algorithme parcourant cet ensemble sont alors finis.\\
De plus, le premier composant MMP qui génère $\mathcal{MFS}$ est prouvé fini. Donc, l'algorithme exécuté par la méthode \textrm{\emph{SPaCov}} se termine.\\
De même pour le deuxième composant à savoir AMP qui génère $\mathcal{MFS}$ et $\mathcal{MRS}$) il est prouvé comme fini. d'où, l'algorithme exécuté par la méthode \textrm{\emph{SPaCov+}} se termine.

\subsection{Complexité théorique}
Dans ce qui suit, nous prouvons la terminaison des algorithmes liés au composants MMP, AMP et CoC puis nous calculons la complexité au pire des cas de deux méthodes \textrm{\emph{SPaCov}} et \textrm{\emph{SPaCov+}}. Dans le cas de la méthode  \textrm{\emph{SPaCov}} nous pouvons déduire sa complexité par la somme des complexités des composant MMP et CoC. Par ailleurs, la complexité de \textrm{\emph{SPaCov+}} est la somme de deux complexités des composants AMP et CoC.\\
La complexité de \textrm{\emph{SPaCov}} et de \textrm{\emph{SPaCov+}} sont déduit à partir des composants MMP, CoC et AMP. 
\begin{itemize}
\item Complexité de MMP : \\
L'algorithme analyse d'abord la base de données pour calculer le support de chaque séquence. La complexité de cette étape est  $\mathcal{O}(n \times m)$ où n et m sont le nombre de transactions et la longueur moyenne des transactions respectivement. Ensuite, un élagage se fait pour éliminer les séquences non fréquentes et une recherche récursive, des séquences plus grande en terme de taille, à l'aide d'une recherche en profondeur d'abord. Pour chaque appel récursif, l'algorithme nécessite un temps de $\mathcal{O}(n\times m)$ pour créer une base de données projetée des séquences de préfixe, puis il nécessite un temps de $\mathcal{O}(n)$ pour calculer le support dans cette base projetée. 
\item Complexité de CoC : \\
{\`A} partir d'un hypergraphe de $a$ sommets et $b$ hyperarêtes nous avons une complexité exponentielle, au pire des cas, de l'ordre de $\mathcal{O}(b * a^{b+1})$ avec $a=|\mathcal{J}|$ et $b=|\mathcal{E}|$, les tests pour vérifier une traverse minimale s'effectuent en $\mathcal{O(1)}$.
\item Complexité de AMP : \\
L'algorithme analyse d'abord la base de données pour calculer le support de chaque séquence. La complexité de cette étape est  $\mathcal{O}(n \times m)$ où n et m sont le nombre de transactions et la longueur moyenne des transactions respectivement. Ensuite, un élagage se fait pour éliminer les séquences nulles et les 1-séquences rares,  une recherche récursive des séquences plus grande en terme de taille, à l'aide d'une recherche en profondeur d'abord. Pour chaque appel récursif, l'algorithme nécessite un temps de $\mathcal{O}(n\times m)$ pour créer une base de données projetée des séquences de préfixe, puis il nécessite un temps de $\mathcal{O}(n)$ pour calculer le support dans cette base projetée. 
\end{itemize}
En effet, la complexité théorique, au pire des cas, de deux composants MMP et CoC ou AMP et CoC est la même est égale à :\\
$\mathcal{O}((n\times m + n + a \times a^b+1)$.

\section{Conclusion}
L'objectif principal de  \textrm{\emph{SPaCov}} est d'assurer une couverture élevée en utilisant le plus petit nombre possible de \textit{RSU}. Il prend en entrée une base de données séquentielle $\mathcal{D} $ représentant un grand nombre de trajectoires de véhicules et fournit en sortie les emplacements \textit{RSU} appropriés. Cette méthode comporte deux composants, à savoir l'extraction de modèles de mobilité maximale (MMP) et le calcul de la couverture (CoC). Le premier module extrait un ensemble de modèles de mobilité maximale de la base de données de séquences $\mathcal{D}$. Ce dernier calcule les emplacements \textit{RSU} appropriés qui couvrent tous les modèles de mobilité extraits. Il est à noter que, dans la littérature, un pattern de mobilité représente une séquence de junctions. La méthode  \textrm{\emph{SPaCov+}} est une amélioration de la méthode  \textrm{\emph{SPaCov}} puisqu'elle couvre non seulement les patterns de mouvement fréquents mais aussi les patterns de mouvement rares ce qui rend cette méthode plus représentative à la totalité des patterns de mouvement des véhicules. La méthode  \textrm{\emph{SPaCov+}} comporte deux composants, à savoir AMP et Coc. Le premier composant 'All mobility patterns' (AMP) extrait tous les patterns de mobilité des véhicules (les patterns de mobilité des véhicules sont divisés en deux types ; patterns de mobilité fréquents et mobilité de patterns rares), le deuxième composant CoC c'est celui de  \textrm{\emph{SPaCov}} (déjà détaillé dans cette section). Dans le chapitre suivant, nous proposons une deuxième méthode de déploiement des \textit{RSU}s basée sur une contrainte budgétaire en termes de nombre de \textit{RSU} (fixé par l'utilisateur) et utilise un algorithme heuristique pour extraire les meilleurs emplacements afin de déployer les $k$ \textit{RSU}s de l'utilisateur.
\newpage
\thispagestyle{empty}
\mbox{}
\newpage
\chapter{HeSPiC: couvrir avec une contrainte budgétaire}
\chaptermark{HeSPiC}
\minitoc
\section{Introduction}
Dans ce chapitre, nous présentons notre deuxième approche spatio-temporelle nommée  \textrm{\emph{HeSPiC}}. Cette dernière est une méthode Heuristique pour la couverture des modèles de mouvement des véhicules basée sur une contrainte budgétaire, elle permet le déploiement des \textit{RSU} en respectant une contrainte budgétaire supervisée (nombre de \textit{RSU} à déployer). La méthode vise à maximiser le taux de couverture soumis à une contrainte budgétaire, elle prend en entrée la base de données séquentielle $\mathcal{D}$ et le budget sous forme d'un nombre de \textit{RSU} $k$ et il fournit en sortie les meilleurs emplacements dans lesquels les \textit{RSU} seront placées (meilleur en terme de couverture).  \textrm{\emph{HeSPiC}} peut être divisé en deux modules: 
\begin{enumerate}
    \item \textbf{Module 1:} Calcule ($i$) l'ensemble $\mathcal{FS}$ de séquences fréquentes, ($ii$) l'ensemble $\mathcal{MFS}$ séquences fréquentes maximales et l'ensemble $\mathcal{MRS}$ séquences rares minimales de la base de données séquentielle $\mathcal{D}$. Ceci est effectué en utilisant les mêmes modules de  \textrm{\emph{SPaCov}} et  \textrm{\emph{SPaCov+}} à savoir MMP et AMP;    
\item \textbf{Module 2:} le module exécute le composant HBCC pour calculer les meilleurs emplacements dans lesquels les \textit{RSU}s seront placées afin de maximiser la couverture. 
\end{enumerate}

\section{Fondements mathématiques}
\begin{definition} \textsc{\textbf{\textsc{(}Poids d'un croisement \textsc{)}}}\mbox{}\newline\label{jw}
Le poids $\mathcal{W}(j_i)$ d'un croisement $j_i$ est un couple de deux nombres positifs:
\[
    \mathcal{W}(j_i)=\Big(supp(j_i)_{\mathcal{MFS}},supp(j_i)_{\mathcal{MRS}}\Big) , avec
\]
\[
supp(j_i)_{\mathcal{MFS}}=  |\{s \in \mathcal{MFS} | j_i \in s \}|
\]
\[
supp(j_i)_{\mathcal{MRS}}=  |\{s \in \mathcal{MRS} | j_i \in s \}|
\]
\end{definition}
Nous définissons $\mathcal{W}=\left( \mathcal{W}(j_1) \mathcal{W}(j_2),\ldots , \mathcal{W}(j_n)\right) $ comme un vecteur de poids de tous les croisements. Ce vecteur peut être trié de la manière suivante:
\[
\mathcal{W}(j_p) > \mathcal{W}(j_q) \iff \    \left \{
   \begin{array}{r c l}
      supp(j_p)_{\mathcal{MFS}}  & > & supp(j_q)_{\mathcal{MFS}} \\
         & OU & \\
      supp(j_p)_{\mathcal{MFS}} & = & supp(j_q)_{\mathcal{MFS}} \ \ \& \ \ supp(j_p)_{\mathcal{MRS}}> supp(j_q)_{\mathcal{MRS}}
      
   \end{array}
   \right . 
\]
D'une manière informelle, $\mathcal{W}(j_p)$ est supérieur à $\mathcal{W}(j_q)$ ou nous pouvons dire que $j_p$ est plus important que $j_q$ point de vue poids si et seulement si son apparition dans l'ensemble $\mathcal{FS}$ est supérieur à l'apparition de $j_q$ dans le même ensemble. Au cas où il y a une égalité nous regardons la fréquence d'apparition de deux croisements dans l'ensemble des séquences rares $\mathcal{RS}$ celui qui possède la valeur supérieur c'est celui qui est le plus important.

\begin{example}
Dans cet exemple, nous considérons la base de données séquentielle $\mathcal{D}$ de l'exemple \ref{cxt2}, le poids des croisement sont donné comme suit : 
\[
\mathcal{W}(j_1)= (0,1); \mathcal{W}(j_2)= (1,5) ; \mathcal{W}(j_3)= (3,7) 
\]
\[
\mathcal{W}(j_4)= (0,1); \mathcal{W}(j_5)= (2,6) ; \mathcal{W}(j_6)= (4,6) ; \mathcal{W}(j_7)= (2,5)
\]
L'ensemble $\mathcal{W}$ est sera comme suit : \\
\[
\mathcal{W}= \{(0,1) ; (1,5) ; (3,7) ;(0,1) ;(2,6); (4,6) ; (2,5)\}
\]

\end{example}

\begin{definition} \textsc{\textbf{\textsc{(} Probabilité de traversé un croisement - Poisson distribution\textsc{)}}}\mbox{}\newline\label{jp}
La distribution de Poisson est une probabilité discrète qui exprime la probabilité qu'un nombre donné d'événements se produisent dans un intervalle de temps ou d'espace fixe si ces événements se produisent avec un taux constant connu et indépendamment du temps écoulé depuis le dernier événement \cite{poisson}.\\
La distribution de Poisson peut également être utilisée pour le nombre d'événements dans d'autres intervalles spécifiés tels que la distance, la surface ou le volume. Un événement peut se produire à $0, 1, 2 $, etc. fois dans un intervalle. Le nombre moyen d'événements dans un intervalle est noté $\lambda$. $\lambda$ est le taux d'événement, \textit{a.k.a} le paramètre \textit{rate}. La probabilité d'observer des événements $m$ dans un intervalle est donnée par l’équation:

\[
\mathcal{P}(X=m,\lambda,T)=e^{-\lambda} \frac{\lambda^{m}}{m!}
\]
Avec:
\begin{itemize}
    \item $\lambda$: est le nombre moyen d'événements dans T.
    \item $e$: Le nombre d'Euler égal à $2.71828$.
    \item $m$: prend les valeurs $0,1,2,$ etc. et représente le nombre d'événements dans T.
    \item $T$: intervalle du temps fini.
\end{itemize}
Dans un souci de simplification, nous pouvons écrire $\mathcal{P} (X = m, \lambda, T)$ simplement $\mathcal{P}(X = m)$. En effet, la distribution de Poisson pourrait être appliquée pour estimer la probabilité de franchir un croisement (junction) $j_i$ par
$m$ Véhicules. Cette probabilité est formalisée comme suit:

\[
\mathcal{P}(j_i,X=m)=e^{-\lambda_i} \frac{\lambda_i^{m}}{m!}
\]
Où $\lambda_i$ désigne le nombre de véhicules qui ont traversé $j_i$ pendant l'intervalle de temps fini $T$.\\
Par conséquent, la probabilité de franchir un croisement $j_i$ par tous les véhicules notée $\mathcal{P}(j_i)$ la probabilité de franchissement de $j_i$ est définie comme suit:

\[
\mathcal{P}(j_i)= \sum_{m=1}^{M} \mathcal{P}(j_i,X=m)
\]
Nous définissons $ \mathcal{P} = \left(\mathcal{P}(j_1), \mathcal{P}(j_2), \ldots, \mathcal{P}(j_n) \right)$ comme vecteur de probabilité de franchissement de tous les croisements. Ce vecteur est trié par ordre décroissant par rapport à la probabilité de franchissement.
\end{definition}

\begin{example}

Dans cet exemple, nous considérons la base de données séquentielle $\mathcal{D}$ de l'exemple \ref{cxt2} avec $\lambda_i$ désigne le nombre de véhicules qui ont traversé chaque croisements pendant l'intervalle de temps fini $T$. les probabilités de franchissement des croisements sont donnée comme suit : 
\[
\mathcal{P}(j_1,X=m)=e^{-\lambda_i} \frac{\lambda_i^{m}}{m!} = 0,630 ~ pour ~ \lambda = 1
\]
\[
\mathcal{P}(j_2,X=m)=e^{-\lambda_i} \frac{\lambda_i^{m}}{m!} = 0,863 ~ pour ~ \lambda = 2
\]
\[
\mathcal{P}(j_3,X=m)=e^{-\lambda_i} \frac{\lambda_i^{m}}{m!} = 0,859 ~ pour ~ \lambda = 5
\]
\[
\mathcal{P}(j_4,X=m)=e^{-\lambda_i} \frac{\lambda_i^{m}}{m!} = 0,630 ~ pour ~ \lambda = 1
\]
\[
\mathcal{P}(j_5,X=m)=e^{-\lambda_i} \frac{\lambda_i^{m}}{m!} = 0,938 ~ pour ~ \lambda = 3
\]
\[
\mathcal{P}(j_6,X=m)=e^{-\lambda_i} \frac{\lambda_i^{m}}{m!} = 0,740 ~ pour ~ \lambda = 7
\]
\[
\mathcal{P}(j_7,X=m)=e^{-\lambda_i} \frac{\lambda_i^{m}}{m!} = 0,930 ~ pour ~ \lambda = 4
\]

\end{example}

\begin{definition} \textsc{\textbf{\textsc{(}Distance\textsc{)}}}\mbox{}\newline\label{dis}
La distance entre deux croisements $j_p$ et $j_q$, notée $d_{pq}$, est définie comme le chemin le plus court. La matrice suivante $\mathcal{D}is  = (d_{pq})_{n \times n}$ représente la distance entre tous les croisements (junctions).

\[\mathcal{D}is =  \begin{pmatrix}
0 & d_{1~2} & \cdots && d_{1~n} \\
     d_{2~1} & \ddots & && \vdots \\
      \vdots &&0& &\vdots\\
       \vdots & && \ddots & d_{n-1~n}\\
      d_{n~1} & \cdots && d_{n~n-1}& 0 \\
\end{pmatrix}\]
\end{definition}

\begin{example}
Dans cet exemple, nous considérons la base de données séquentielle de l'exemple \ref{cxt2} et la matrice qui représente la distance entre tous les croisement donnée par $\mathcal{D}is$ est la suivante : \\
\[
\mathcal{D}is =\begin{pmatrix}
  0 & 110 & 170 & 40& 50 & 60& 20  \\
  110 & 0 & 70& 95& 95 & 45& 115 \\
  170 & 70 & 0 & 190 & 110 & 90 & 150 \\
  40& 95 & 190 & 0 & 90 & 90 & 60 \\
  60 & 95 & 110& 90 & 0 & 50 & 50 \\
  60 & 45 & 90 & 40 & 50 & 0 & 80  \\
  20 & 115& 150 & 60 & 50 & 80 & 0
\end{pmatrix}
\]
\end{example}

\section{Méthodologie et algorithme proposé}
À présent, nous introduisons notre deuxième méthode pour le déploiement des \textit{RSU}s, cette méthode utilise trois métriques pour choisir les emplacements des \textit{RSU}s sur la carte. Les métriques sont: ($i$) l'importance du croisement (Poids), ($ii$) La déduction des modèles de mobilité en se basant sur un historique de mouvement et ($iii$) le positionnement des \textit{RSU}s d'une manière incrémentale en tenant compte de la disposition spatiale des \textit{RSU}s déjà déployés.\\
Cette méthode est composée de deux parties (voir Figure \ref{hespic}) ; le premier composant nommé AMP celui de  \textrm{\emph{SPaCov+}} (détaillé dans le Chapitre 2) et le deuxième nommé HBCC (Heuristic Budget-Constrained Coverage).\\
Dans ce qui suit, nous détaillons la méthode  \textrm{\emph{HeSPic}}.
\begin{figure}[!ht]
\centering
\includegraphics[scale=0.6]{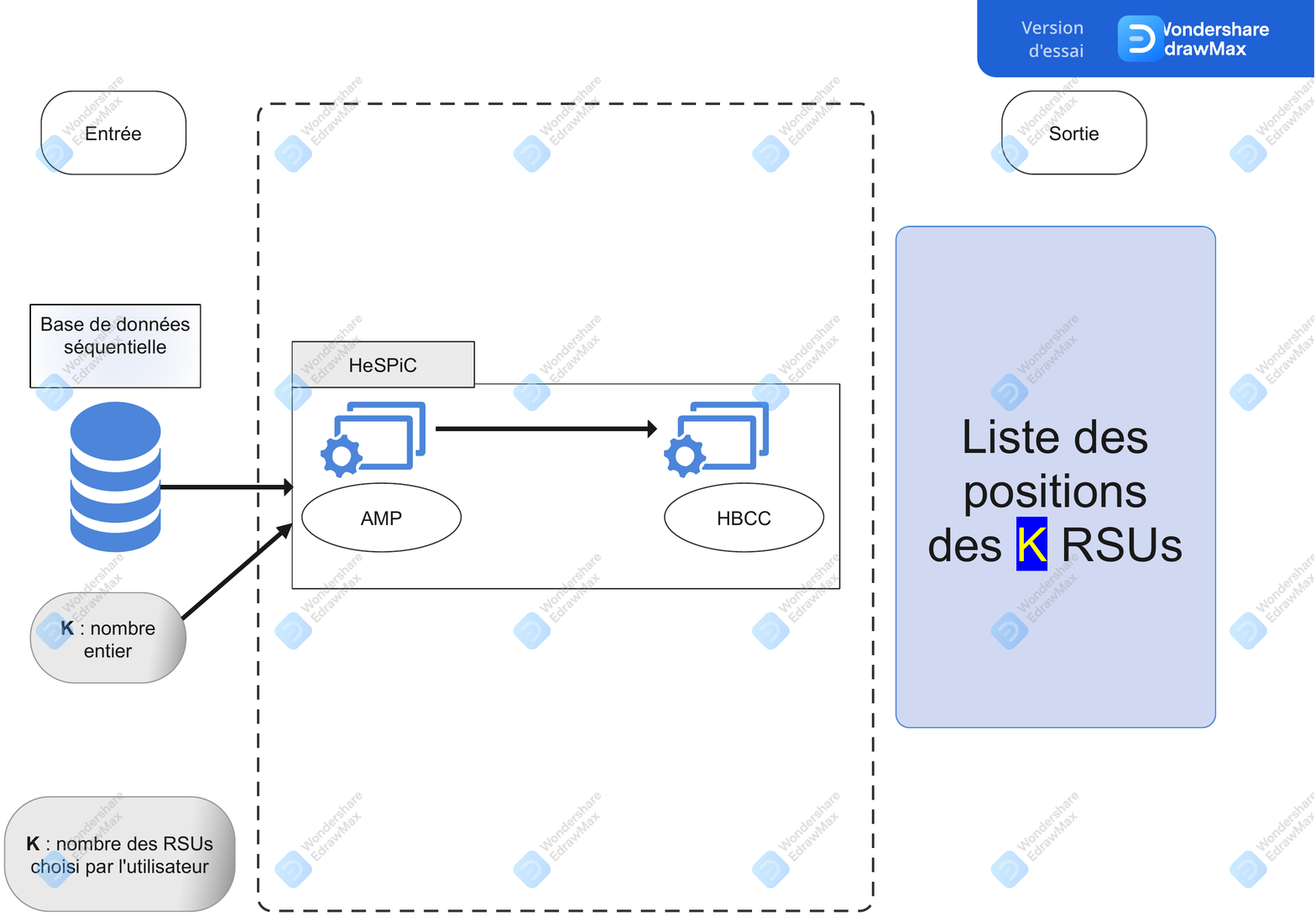}
\caption{L'architecture de la méthode HeSPiC.}
\end{figure}\label{hespic}
\subsection{Description de HeSPiC}
 \textrm{\emph{HeSPiC}} est une méthode de déploiement des \textit{RSU}s spatio-temporelle avec une contrainte budgétaire (supervisé), qui vise à maximiser le taux de couverture tout en tenant compte d'une contrainte budgétaire. La méthode prend en entrée la base de données séquentielle $\mathcal {D}$ et le budget sous forme d'un nombre de \textit{RSU} $ k $ et elle fournit en sortie les $k$ meilleurs emplacements dans lesquels les \textit{RSU}s seront placées.  \textrm{\emph{HeSPiC}} divise le processus de déploiement en deux étapes comme indique la Figure \ref{hespic}. Dans ce qui suit, nous détaillons les étapes de  \textrm{\emph{HeSPiC}}: 
\begin{enumerate}
    \item \textbf{Étape 1:} le module AMP calcule l'ensemble $ \mathcal{FS} $ des séquences fréquentes, l'ensemble $ \mathcal{MFS} $ séquences fréquentes maximales et l'ensemble $ \mathcal{MRS} $ séquences rares minimales de la base de données séquentielle $ \mathcal{D} $. Ceci est fait en utilisant les mêmes algorithmes utilisés par  \textrm{\emph{SPaCov}} et  \textrm{\emph{SPaCov+}} ;    
\item \textbf{Étape 2:} le module HBCC est invoqué en deuxième phase pour calculer les $k$ meilleurs emplacements dans lesquels les \textit{RSU}s seront placées.
\end{enumerate}

\subsection{Algorithme proposé}
L'algorithme de classement 'Ranking algorithm' (voir algorithme ) est la partie la plus importante de HBCC. Tout d'abord, nous commençons par présenter le composant HBCC. Ensuite, nous détaillons les algorithmes qui le composent.
\subsection*{Heuristic Budget-Constrained Coverage (HBCC)}
Pour trouver les $k$ meilleurs emplacements pour déployer les \textit{RSU}s, HBCC calcule un score pour chaque croisement. Dès lors, les \textit{RSU}s seront placées dans les croisements ayant les scores les plus élevés. Pour ce faire, il exploite les informations spatio-temporelles des mouvements des véhicules ainsi que les attributs spatiales de la zone à couvrir. En effet, HBCC s'appuie sur les trois métriques présentées dans la section précédente à savoir $\mathcal{W}$, $\mathcal{P}$ et $\mathcal{D}$ pour calculer le score de chaque croisement  $j_i$ $(1 \leq i \leq n)$. Dans ce qui suit, nous décrivons la formule de classement des croisements utilisée par l'algorithme de classement 'Ranking algorithm':
 \[
    score (j_i)=\frac{\alpha \mathcal{O}(j_i,\mathcal{W}) + \beta \mathcal{O}(j_i,\mathcal{P}) + \delta \mathcal{O}(j_i,\mathcal{L})}{\alpha + \beta + \delta}
   \]
\begin{algorithm}
\BlankLine
  \KwIn{$\mathcal{FS}, \mathcal{MFS}$,: Ensemble des séquences fréquentes et maximales fréquentes.\\
  $\mathcal{MRS}$: ensemble des séquences rares minimales.\\
  $\mathcal{D}is$: matrice de distances.\\
  $k$: nombre entier représente le nombre maximal des \textit{RSU}s (couverture budgétaire).
  }
  \KwOut{R: Liste des TOP $k$ croisements classés.}
  \BlankLine
  
\ForEach{$j$ of $\mathcal{S}$}
        {
        \BlankLine
       $\mathcal{W}(j)$ = getWeightOfJunction($j$) \\
       $\mathcal{P}(j)$ = getProbabilityOfJunction($j$) \\
         
        }
      sortWeight($\mathcal{W}$)\\  
      sortProbability($\mathcal{P}$)\\  
      $j_p$=popElement($\mathcal{W}$)\\
      $\mathcal{L}$=greedyLongestPath($\mathcal{S}$,$j_p$,$\mathcal{D}is$)\\
    \ForEach{$j$ of $\mathcal{S}$}
    {
        \BlankLine
       \[
    score(j)=\frac{\alpha \mathcal{O}(j,\mathcal{W}) + \beta \mathcal{O}(j,\mathcal{P}) + \delta \mathcal{O}(j,\mathcal{L})}{\alpha + \beta + \delta}
   \]
   insert(score(j),R(j))
        } 
  return ($k$ premiers éléments de R)

\caption{Ranking algorithme}
\label{rankalgo}
\end{algorithm}
\subsection*{Description de l'algorithme}
L'algorithme de classement est le module le plus important du composant HBCC, ce module est composé de huit étapes représentées comme suit: 

\begin{itemize}
    \item \textbf{Entrée:} 
    \begin{itemize}
    \item $\mathcal{FS}$: ensemble des séquences fréquentes.
    \item $\mathcal{MFS}$: ensemble des séquences fréquentes maximales.
    \item $\mathcal{MRS}$: ensemble des séquences rares minimales.
    \item $\mathcal{D}is$: matrice de distances.
    \item $K$: nombre entier représente le nombre maximal des \textit{RSU}s (couverture budgétaire).
\end{itemize}
    \item \textbf{Sortie:}
    \begin{itemize}
    \item Liste des tops k croisements classés.
\end{itemize}
\end{itemize}

\begin{itemize}

    \item \textbf{Étape 1:} calculer le vecteur Poids pour tous les croisements (jucntions) $\mathcal{W}=\left( \mathcal{W}(j_1), \mathcal{W}(j_2),\ldots , \mathcal{W}(j_n)\right) $;
   \item \textbf{Étape 2:} Trier $\mathcal{W}$ en utilisant la relation définie par la Définition \ref{jw};
   \item \textbf{Étape 3:} Calculer le vecteur de probabilité de franchissement pour tous les croisements $\mathcal{P}=\left( \mathcal{P}(j_1), \mathcal{P}(j_2),\ldots , \mathcal{P}(j_n)\right)$;
   \item \textbf{Étape 4:} Trier $\mathcal{P}$ dans un ordre décroissant par rapport à la probabilité de franchissement; 
   \item \textbf{Étape 5:}  Choisir le croisement $j_p$ ayant le poids le plus élevé dans $\mathcal{W}$;
   \item \textbf{Étape 6:} Calculer le plus long chemin $\mathcal{L}$ à partir de $j_p$ qui contient tous les croisements (junctions) $j_i$ $(i \neq p \  \& \  1 \leq i \leq n)$ en utilisant notre Algorithme heuristique, nommé  greedyLongestPath (voir Algorithme \ref{greedyAlgo}).
   \item \textbf{Étape 7: } Calculer le score de chaque croisement $j_i$ $1 \leq i \leq n$) de la manière suivante :
  \[
    score (j_i)=\frac{\alpha \mathcal{O}(j_i,\mathcal{W}) + \beta \mathcal{O}(j_i,\mathcal{P}) + \delta \mathcal{O}(j_i,\mathcal{L})}{\alpha + \beta + \delta}
   \]
 Avec :
\begin{itemize}
    \item $\mathcal{O}$($j_i$,$\mathcal{W}$) est le rang de $j_i$ dans le vecteur trié $\mathcal{W}$;
    \item $\mathcal{O}$($j_i$,$\mathcal{P}$) est le rang de $j_i$ dans le vecteur trié $\mathcal{P}$;
    \item $\mathcal{O}$($j_i$,$\mathcal{L})$ est le rang de $j_i$ dans l'ensemble $\mathcal{L}$ qui représente le plus long chemin;
    \item $\alpha$,$\beta$ and $\delta$ sont des coefficients de pondération.
\end{itemize}
    \item \textbf{Étape 8:} Sélectionner les top $k$ premiers croisements ayant les meilleurs score. 
\end{itemize}
\begin{algorithm}
\BlankLine
  \KwIn{\\$\mathcal{J}$: Ensemble des croisements.\\ $j_p$ : Croisement ayant le Poids le plus élève.\\ $\mathcal{D}is$: Matrice distance.}
  \KwOut{\\$\mathcal{L}$: Ensemble contenant le plus long chemin entre tous les croisements.}
  \BlankLine
     $\mathcal{L} =\emptyset$\\
     $e=j_p$\\
     insert($e$,$\mathcal{L}$)\\
  \While{($\mathcal{J} \neq \emptyset$)}
        {
          \BlankLine
           $f=$ getFarthestJunction(e,$\mathcal{D}is$)\\
           insert($f$,$\mathcal{L}$)\\
           $e=f$\\
           delete($e$,$\mathcal{J}$)
        } 
  return ($\mathcal{L}$)

\caption{greedyLongestPath}
\label{greedyAlgo}
\end{algorithm}

\subsection{Exemple illustratif}
\begin{example}
Dans cet exemple, nous considérons les ensembles $\mathcal{FS}$, $\mathcal{MFS}$ and $\mathcal{MRS}$, extraits à partir de la base de données séquentielle $\mathcal{D}$ de l'exemple \ref{cxt2}, donnés comme suit:  
\begin{itemize}
     \item $\mathcal{FS}=\{j_{2},j_{3}, j_{5}, j_{6}, j_{7}, <j_{2}j_{6}>,<j_{2}j_{7}>,<j_{3}j_{6}>,<j_{3}j_{7}>,<j_{5}j_{7}>,<j_{5}j_{3}>,<j_{6}j_{7}>,<j_{6}j_{5}>,<j_{6}j_{3}>,<j_{2}j_{6}j_{7}>,<j_{5}j_{3}j_{7}>\}$ 
    \item $\mathcal{MFS}=\{<j_{6}j_{5}>,<j_{3}j_{6}>,<j_{6}j_{3}>,<j_{2}j_{6}j_{7}>,<j_{5}j_{3}j_{7}>\}$
    \item $\mathcal{MRS}=\{j_{1},j_{4},<j_{2}j_{3}>,<j_{2}j_{5}>,<j_{3}j_{2}>,<j_{3}j_{5}>,<j_{5}j_{6}>,<j_{5}j_{2}>,<j_{6}j_{2}>,<j_{7}j_{6}>,<j_{7}j_{5}>,<j_{7}j_{3}>,<j_{3}j_{6}j_{7}>,<j_{6}j_{5}j_{7}>,<j_{6}j_{5}j_{3}>,<j_{6}j_{3}j_{7}>\}$ 
\end{itemize}
Soit $\mathcal{D}$ la matrice de distance entre tous les croisements dans $\mathcal{MFS} \bigcup \mathcal{MRS}$ donné par la Figure \ref{graphexp}(b).

\begin{figure}[ht]
\centering
\begin{subfigure}{.5\textwidth}
\centering
  \includegraphics[scale=0.3]{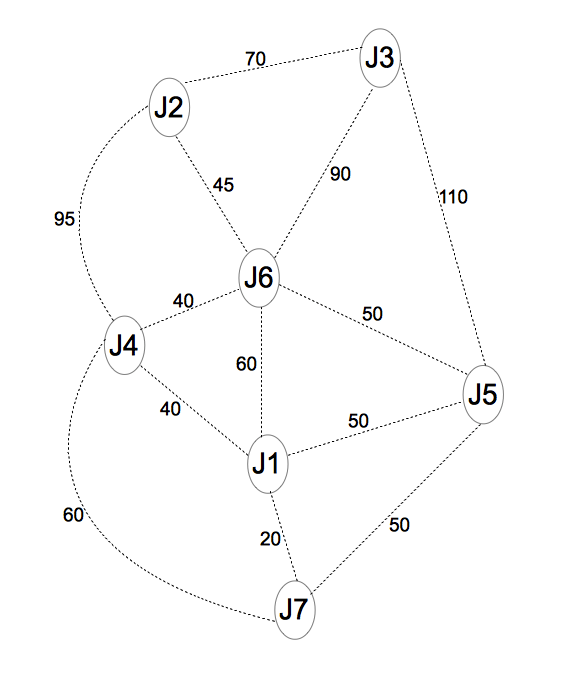}
\end{subfigure}%
\begin{subfigure}{.5\textwidth}
  $$
\mathcal{D}is =\begin{pmatrix}
  0 & 110 & 170 & 40& 50 & 60& 20  \\
  110 & 0 & 70& 95& 95 & 45& 115 \\
  170 & 70 & 0 & 190 & 110 & 90 & 150 \\
  40& 95 & 190 & 0 & 90 & 90 & 60 \\
  60 & 95 & 110& 90 & 0 & 50 & 50 \\
  60 & 45 & 90 & 40 & 50 & 0 & 80  \\
  20 & 115& 150 & 60 & 50 & 80 & 0
\end{pmatrix}
$$
\end{subfigure}%
  \caption{(a) Exemple d'une topologie routière et (b) Matrice de distance nommée $\mathcal{D}is$} \label{graphexp}
\end{figure}
 \textrm{\emph{HeSPiC}} invoque le module AMP pour avoir les résultats des ensembles $\mathcal{FS}, \mathcal{MFS}$ et $\mathcal{MRS}$, en deuxième lieu il invoque le module HBCC, qui appel a son tour l'algorithme de classement Ranking algorithm, afin de calculer les emplacements des $k$ \textit{RSU}s sur la carte. L'algorithme de classement, un module responsable de l'ordonnancement des croisements en terme d'importance (selon un score), suit huit étapes.\\
Considérant $k=4$ le budget choisi par l'utilisateur, qui représente le nombre des \textit{RSU}s maximal à déployer, le composant HBCC procède comme suit pour trouver les meilleurs emplacements des croisements dans lesquels les $k$ \textit{RSU}s seront placées.

\begin{enumerate}
   \item  Il commence par calculer le vecteur Poids $\mathcal{W}$ comme indiqué dans la colonne $\mathcal{W}$ du Tableau \ref{WTab}.
  
   \item Il attribut un classement pour chaque croisement $j_i$ dans $\mathcal{W}$ comme illustré dans la colonne $\mathcal{O(W)}$ de la Table \ref{WTab}.
   
   \item  Il calcule le vecteur Probabilité de franchissement $\mathcal{P}$ comme indiqué dans la colonne $\mathcal{P}$ du Tableau \ref{probTab}.
   
   \item  Il classe chaque croisement dans $\mathcal{P}$ comme illustré dans la colonne $\mathcal{O(P)}$ du Tableau \ref{probTab}.
   
   \item  Dans cet exemple, le croisement $j_6$ a le Poids le plus élevé dans $\mathcal{W}$.
  
   \item  Il calcule le plus long chemin $\mathcal{L}$ à partir de $j_6$ en utilisant la matrice des distances $\mathcal{D}is$. $\mathcal{L} = \left( j_6, j_3, j_4, j_2, j_7, j_5, j_1 \right)$.
    
   \item Il calcule le score de chaque croisement comme indiqué dans la colonne $score$ du Tableau \ref{HBCCTab}. les colonnes $\mathcal{O}(j_i,\mathcal{W)}$, $\mathcal{O}(j_i,\mathcal{P)}$ and $\mathcal{O}(j_i,\mathcal{L)}$ représente l'ordre du croisement $j_i$ dans $\mathcal{W}$, $\mathcal{P}$ et $\mathcal{L}$ respectivement.
   
    \item Les croisements mentionnés en gras dans le Tableau \ref{HBCCTab} sont les meilleurs classés par l'algorithme. Par conséquent, $j_6$, $j_3$, $j_2$ et $j_5$ sont les $k=4$ meilleurs emplacements pour mettre les \textit{RSU}s défini à l'avance par l'utilisateur. 
\end{enumerate}

\end{example}

\begin{table}[ht]
 \centering
 \begin{center}

\begin{tabular}{||p{2cm} p{2cm} p{2cm} p{2cm}||c||}
\hline
Junction&$supp(j_i)_{\mathcal{MFS}}$&$supp(j_i)_{\mathcal{MRS}}$&$\mathcal{W}$&$\mathcal{O}(j_i,\mathcal{W})$\\
\hline
\hline
$j_{1}$& 0&1&(0,1)&1\\

$j_{2}$&1& 5&(1,5)&3\\

$j_{3}$&3& 7&(3,7)&6\\

$j_{4}$& 0&1&(0,1)&2\\

$j_{5}$& 2&6&(2,6)&5\\

$j_{6}$&4&6&(4,6)&7\\

$j_{7}$&2&5&(2,5)&4\\
\hline
\end{tabular}
\end{center}
\caption{Poids des croisements}\label{WTab}
\end{table}

\begin{table}[ht]
 \centering
 \begin{center}

\begin{tabular}{||p{2cm} p{2cm} p{2cm} ||c||}
\hline
Junction ID&$\lambda$&$\mathcal{P}$&$\mathcal{O}(j_i,\mathcal{P})$\\
\hline
\hline
$j_{1}$& 1& $0.630$&1\\

$j_{2}$&2& $0.863$&5\\

$j_{3}$&5& $0.859$&4\\

$j_{4}$& 1& $0.630$&2\\

$j_{5}$& 3&$0.938$&7\\

$j_{6}$&7&$0.740$&3\\

$j_{7}$&4&$0.930$&6\\
\hline
\end{tabular}
\end{center}
\caption{Probabilité de franchissement des croisements.}\label{probTab}
\end{table}

\begin{table}[ht]
 \centering
 \begin{center}

\begin{tabular}{|| c | p{2cm} p{2cm} p{2cm}||c||}
\hline
$\mathcal{O}(\mathcal{W})$
Junction&$\mathcal{O}(j_i,\mathcal{W})$&$\mathcal{O}(j_i,\mathcal{P})$&$\mathcal{O}(j_i,\mathcal{L})$&$score$\\
\hline
\hline
$j_{1}$&$1$& $1$& $1$&$1$\\

$\mathbf{j_{2}}$&$3$& $5$& $4$&$\mathbf{5}$\\

$\mathbf{j_{3}}$&$6$&$4$& $6$&$\mathbf{6}$\\

$j_{4}$&$2$& $2$& $5$&$3$\\

$\mathbf{j_{5}}$&$5$& $7$&$2$&$\mathbf{4}$\\

$\mathbf{j_{6}}$&$7$&$3$&$7$&$\mathbf{7}$\\

$j_{7}$&$4$&$6$&$3$&$2$\\
\hline
\end{tabular}
\caption{Score des croisements}\label{HBCCTab}
\end{center}
\end{table}

\section{Étude de la complexité et de la validité}
Dans ce qui suit, nous prouvons la terminaison de notre algorithme liée au composant HBCC puis nous calculons sa complexité au pire des cas. Dans le cas de la méthode  \textrm{\emph{HeSPiC}} nous ne pouvons pas étudier la correction et la complétude de cette dernière à cause de type heuristique de la méthode. En effet, la méthode n'assure pas une couverture optimale dans toutes les situations.
\subsection{Terminaison de l'algorithme}
\begin{proposition}
L'algorithme lié à HBCC se termine.
\end{proposition}
\begin{proof}
Le nombre des \textit{RSU}s générés par l'algorithme est fini.\\
En effet, le nombre $k$ des \textit{RSU}s donné par l'utilisateur est au pire des cas $k$ doit être égal à la cardinalité de $|\mathcal{J}|$. Puisque l'ensemble fourni par l'algorithme est fini, les quatre boucles de l'algorithme parcourant cet ensemble sont alors, eux aussi, finis.\\
De plus, l'algorithme du premier composant AMM qui génère $\mathcal{MRS}$ et $\mathcal{MFS}$ est prouvé fini dans le chapitre précédent. Donc, l'algorithme lié au module HBCC se termine.
\end{proof}
\subsection{Complexité théorique}
Dans le cas de la méthode \textrm{\emph{HeSPiC}} nous pouvons déduire sa complexité par la somme des complexités des composant MMP et HBCC. {\`A} partir d'une base de données séquentielle de $n$ véhicules et $m$ croisements , nous avons : 
\begin{itemize}
\item Complexité de MMP : \\
L'algorithme analyse d'abord la base de données pour calculer le support de chaque séquence. La complexité de cette étape est  $\mathcal{O}(n \times m)$. Ensuite, un élagage se fait pour éliminer les séquences non fréquentes et une recherche récursive, des séquences plus grande en terme de taille, à l'aide d'une recherche en profondeur d'abord. Pour chaque appel récursif, l'algorithme nécessite un temps de $\mathcal{O}(n\times m)$ pour créer une base de données projetée des séquences de préfixe, puis il nécessite un temps de $\mathcal{O}(n)$ pour calculer le support dans cette base projetée. 
\item Complexité de HBCC : \\
L'algorithme calcul les poids des croisements pour un coût au pire des cas de $\mathcal{O}(m)$. Ensuite, il procède à un calcul de probabilité pour les croisements ce qui nécessite un temps de l'ordre de  $\mathcal{O}(m)$. {\`A} la fin, l'algorithme calcul l'ordre des croisements en se basant sur les distances qui les sépare pour un coût, au pire des cas, qui vaut  $\mathcal{O}(m*m)$. Le coût de la formule score sera ignoré puisqu'il vaut $\mathcal{O}(1)$.
\end{itemize}
En effet, la complexité théorique, au pire des cas, de deux composants MMP et HBCC est égale à :\\
$\mathcal{O}((n\times m) + \mathcal{O}(m*m) +\mathcal{O}(m) + \mathcal{O}(m)$ $\cong$  $\mathcal{O}(m^2+n*m)$ .

\section{Conclusion}
Dans ce chapitre, nous avons présenté notre deuxième méthode de déploiement des \textit{RSU}s basé sur une contrainte budgétaire supervisé par l'utilisateur. Ensuite, nous avons étudié la complexité et la validité de l'algorithme assurant l'extraction des croisements dans lesquels les \textit{RSU}s seront placés. En effet, notre  \textrm{\emph{HeSPic}} ne peut pas être comparé aux  autres méthodes proposé dans ce manuscrit puisque son résultat final dépend, contrairement aux autres méthodes de déploiement de la littérature, non seulement de la base des mouvements des véhicules mais aussi d'une constante représentant le nombre des \textit{RSU}s $k$ (en termes de budget) que l'utilisateur souhaite les mettre en place (la méthode essaie d'atteindre un taux de couverture maximal avec ce nombre $k$. Dans le chapitre suivant, nous proposons une troisième méthode de déploiement des \textit{RSU}s basée sur de nouvelles métriques à savoir (i) l'utilité d'une séquence, (ii) le bénéfice d'une séquence et (iii) la densité d'une séquence.

\newpage
\thispagestyle{empty}
\mbox{}
\newpage
\chapter{MIP: réduire pour déployer}
\chaptermark{MIP}
\minitoc
\section{Introduction}
Dans ce chapitre, nous présentons notre troisième approche spatio-temporelle qui permet d'exploiter les modèles de mobilité à partir des trajectoires des mouvements de véhicules. À cet égard, nous décrivons une nouvelle méthode de représentation des modèles de mobilité de véhicules basée sur les modèles utiles en data mining. Dans ce qui suit, nous rappelons et adaptons certaines notions mathématiques des modèles séquentiels.\\

\section{Fondements mathématiques}
\begin{definition}\textsc{\textbf{\textsc{(}Trajectoire Vs séquence \textsc{)}}}\mbox{}\newline
Une trajectoire est une séquence de croisements qui représente la trajectoire (ou le chemin) d'un véhicule $v$ noté $< j_{1} j_{2} ... j_{k} >$ entre deux points $A$ et $B$. En effet, il existe une connexion spatiale entre les croisements d'une trajectoire. Cependant, une séquence est un ensemble de croisements est une combinaison aléatoire de croisement dans n'importe quel ordre, évidemment un ensemble de croisements peut ne pas avoir lieu dans l'espace (c.à.d. ensemble de point non adjacent). 
\end{definition}
\begin{example}
Dans cet exemple, nous considérons la topologie spatiale des croisements présenté dans la Figure \ref{graphexp} (a) et nous présentons, dans ce qui suit, un exemple d'une trajectoire et d'une séquence: \\  
L'ensemble $< j_{5},j_{3},j_{2},j_{6}>$ forme une trajectoire puisque (i) elle est composé d'un ensemble de croisements $j_i$ et (ii)  spatialement les $j_{i}$ successives sont deux à deux adjacents. \\
$< j_{5},j_{4},j_{3},j_{7}>$ est une séquence $s$, mais elle ne peut pas être une trajectoire, puisque (i) elle est composé d'un ensemble de croisements et (ii) il y a au moins un pair de croisement successive non adjacent ($<j_{5},j_{4}>$ non adjacent ou $<j_{3},j_{7}>$ non adjacent).
\end{example}
\begin{definition}\textsc{\textbf{\textsc{}Base de données séquentielle (ensemble de trajectoires) \textsc{}}}\mbox{}\label{defDB}
Une base de données séquentielle $\mathcal{D}_t$ peut être vu comme un ensemble de trajectoires, noté $\mathcal{T}$. Dans notre cas, une trajectoire est une liste ordonnée de croisement traversées par un véhicule $v$, noté par $t_{i}$=$< j_{1} j_{2} ... j_{k} >$.
\end{definition}

\begin{example}
Le Tableau \ref{cxt22} illustre la base de données séquentielle transformée $\mathcal{D}_t$ à partir de la base de données initiale représenté par la Table \ref{cxt2}, où $\mathcal{T} = \{t_{0}; t_{1}; t_{2}; t_{3}; t_{4}; t_{5}; t_{6}\}$ est un ensemble de trajectoires de véhicules et $\mathcal{S}$ représente leurs trajectoires associées. La transformation se fait de la manière suivante: \\
\begin{itemize}
    \item L'ensemble $\mathcal{V}$ sera remplacé par l'ensemble $\mathcal{T}$ puisque nous nous intéressons plus aux trajectoires pour avoir une vue macroscopique des mouvements des véhicules.
    \item Le calcul de l'utilité de chaque trajectoire $t_i$ afin qu'on puisse évaluer l'importance des trajectoires au sein de la base de données. En effet, l'utilité des trajectoires (3ème colonne du Tableau \ref{cxt22}) se calcule en se basant sur la définition \ref{def_utilite}
\end{itemize}
\end{example}

\begin{table}[ht]
\begin{tabular*}{\hsize}{@{\extracolsep{\fill}}|l|c|c|@{}}
\hline
$\mathcal{T}$ & $\mathcal{S}$ & Utility(T)\\
\hline
$t_{0}$& $j_{6} ; j_{5} ; j_{3}$ &14\\
$t_{1}$& $j_{6};j_{5};j_{3};j_{7}$ &18\\
$t_{2}$&$j_{5} ; j_{3}; j_{7}$ &12\\
$t_{3}$& $j_{3} ; j_{6}$ &11\\
$t_{4}$& $j_{3} ; j_{6} ; j_{4}$ &12\\
$t_{5}$&$j_{2} ; j_{6} ; j_{7}$ &12\\
$t_{6}$&$j_{2} ; j_{6} ; j_{7} ; j_{1}$ &13\\
\hline
\end{tabular*}

\caption{Base de données séquentielle avec l'utilité de chaque trajectoire $\mathcal{D}_t$.}\label{cxt22}
\end{table}

\begin{definition}\textsc{\textbf{\textsc{(}Support d'une trajectoire\textsc{)}}}\mbox{}\newline\label{defsupp}
Soit $s \subseteq \mathcal{S}$ une trajectoire d'une base de données séquentielle nommée $\mathcal{D}_t$. Le support de $s$, noté $supp(s)$, est une indication de la fréquence à laquelle la trajectoire $s$ apparaît dans $\mathcal{D}_t$. Il est calculé comme suit: 
\[ supp(s)=\frac{|{t \in \mathcal{T} | s \lesssim t}|}{|\mathcal{T}|}\]
Lorsque la séquence est composée de $l$ croisements (ou junction) elle est appelée \textit{l-junction}.
\end{definition}
Note: Le symbole $\lesssim$ signifie une inclusion ordonnées, par exemple $(j_{3}j_{4})$ $\lesssim$ $(j_{2}\mathbf{j_{3}}j_{5}j_{6};\mathbf{j_{4}})$ mais $(j_{3}j_{4})$ $\not \lesssim$. $(j_{2};\mathbf{j_{4}}j_{5}j_{6};\mathbf{j_{3}})$

\begin{example}
Soit $s=j_{6}j_{7}$ une trajectoire de la base de données séquentielle $\mathcal{D}_t$ le support de $s$ est donné comme suit: 
\[ 
supp((j_{6}j_{7}))=\frac{|\{v_1,v_5,v_7\}|}{|\mathcal{D}_t|} = \frac{3}{8} = 42\%
\]
Avec $|\mathcal{D}_t|$ égal au nombre de transactions de la base de données $\mathcal{D}_t$. 
\end{example}
\begin{definition}\textsc{\textbf{\textsc{(}Trajectoire Fréquente\textsc{)}}}\mbox{}\newline
La trajectoire $s$ est dite fréquente si et seulement si son support dans $\mathcal{D}_t$ est supérieur ou égal à $minsup$ (\emph{seuil spécifié par l'utilisateur}). Formellement, nous définissons l'ensemble $\mathcal{FS}_t$ de trajectoires fréquentes dans $\mathcal{D}_t$ comme:
\[ \mathcal{FS}_t = \{s \in \mathcal{D}_t | supp(s) \geqslant minsup\}\]
\end{definition}
\vspace{- 0.4cm}

\begin{example}
Considérant un seuil $minsup=2/8$ et la base de données séquentielle $\mathcal{D}_t$ de l'exemple \ref{cxt22}, alors la trajectoire $s=<j_{2}j_{6}j_{7}>$ est fréquente puisque son support $supp(s)=2/8  \geqslant minsup$.
\end{example}
\begin{definition}\textsc{\textbf{\textsc{(}Utilité d'une séquence ou trajectoire \textsc{)}}}\mbox{}\newline\label{def_utilite}
Dans cette définition, nous définissons une nouvelle métrique nommée Utilité d'une séquence (respectivement d'une trajectoire). L'utilité d'une séquence recherche les séquences qui apparaissent dans un trafic intense. En effet, l'utilité d'une séquence $s$ dans $\mathcal{D}_t$, notée $U(s)$, est la somme du support de chaque sous-ensemble à 1-junction de $s$ dans $\mathcal{D}_t$. Formellement, $U(s)$ (respectivement $U(t)$) est défini comme suit:

\[
U(s) = \sum\limits_{\substack{s' \subseteq s}} supp(s^{'}).   ~~ tel ~~que ~~  s^{'}~~ est ~~~ 1-junction~~séquence.
\]

\[
U(t) = \sum\limits_{\substack{t' \subseteq t}} supp(t^{'}).   ~~ tel ~~que ~~  t^{'}~~ est ~~~ 1-junction~~séquence.
\]
Remarque: Une séquence $s$ est considérée comme \underline{très importante} lorsque sa valeur \underline{d'utilité est élevée}. 
\end{definition}
\begin{example}
Considérant la base de données séquentielle $\mathcal{D}_t$ de l'exemple \ref{cxt22} et la séquence $s=<J_{2}J_{6}J_{7}>$, l'utilité de $s$ est calculé comme suit: \\
$U(s)= supp(J_{2})+supp(J_{6})+supp(J_{7})=14$.
\end{example}
\begin{definition}\textsc{\textbf{\textsc{(}Densité d'une séquence \textsc{)}}}\mbox{}\newline
La densité d'une séquence $s$ dans la trajectoire $t$, noté $d_{t}(s)$ est l'utilité de $t$ si $ s \subseteq t $. Formellement, $d_{t}(s)$ est définit comme suit: 
\[
d_{t}(s)= U(t), ~~ Si ~~ s \subseteq t.
\]
La densité d'une \underline{séquence} $s$ dans une \underline{trajectoire} $t$ est la somme des supports des croisements qui compose $t$. En effet, l'idée est de repérer les croisements (ou séquences: ensemble de croisements) qui apparaissent dans des trajectoires de haute utilités.
\end{definition}

\begin{example}
Considérant la base de données séquentielle $\mathcal{D}_t$ de l'exemple \ref{cxt22}, la densité de la séquence $J_ {6} J_ {5}$ est calculé de la manière suivante : 
\[
d_{t}(J_{6}J_{5})= U(t), ~~ pour~ tous~  J_{6}J_{5} \subseteq t.
\]
\[
d_{t_0}(J_{6}J_{5})= U(t_0) = 14.
\]
\[
d_{t_1}(J_{6}J_{5})= U(t_1) = 18.
\]
\end{example}

\begin{definition}\textsc{\textbf{\textsc{(}Bénéfice d'une séquence \textsc{)}}}\mbox{}\newline
Le bénéfice d'une séquence $s$, noté $Bf(s)$, est une métrique qui évalue l'utilité de $s$ dans une base de données séquentielle. Nous définissons $Bf(s)$ comme suit: 

\[
Bf(s) = \frac{\sum\limits_{\substack{t \in \mathcal{T}}} d_{t}(s) }{supp(s)} + \sum\limits_{\substack{t \in \mathcal{T}}} \frac{U(s) }{d_{t}(s)}, ~~ o\Grave{u} ~~ s \in \mathcal{FS}
\]
Toutes les séquences $s$ avec $Bf(s)$ supérieur ou égal à un seuil \textit{minbenefit} (seuil fixé par l'utilisateur) sont considérées comme bénéfiques.
\end{definition}

\begin{example}
Considérant la base de données séquentielle $\mathcal{D}_t$ de l'exemple \ref{cxt22} et un seuil $minbenefit= 17$, la séquence $J_ {6} J_ {5}$ est considérée comme étant bénéfique puisque son $Bf(J_ {6} J_ {5})$ est supérieur ou égal à $minbenefit$. Cependant, la séquence $J_ {3} J_ {7}$ est considérée comme étant non bénéfique puisque $Bf(J_ {3} J_ {7})$ est inférieur à $minbenefit$. Dans ce qui suit, nous présentons les traces de calcul de  $Bf(J_{6}J_{5})$ et $Bf(J_{3}J_{7})$ respectivement: 

\[
Bf(J_{6}J_{5})= \frac{d_{t0}(J_{6}J_{5}) + d_{t1}(J_{6}J_{5}) }{supp(J_{6}J_{5})} + \frac{U(J_{6}J_{5})}{d_{t_{0}}(J_{6}J_{5})}+\frac{U(J_{6}J_{5})}{d_{t_{1}}(J_{6}J_{5})}. \\
\]

\[
Bf(J_{3}J_{7})= \frac{d_{t1}(J_{3}J_{7}) + d_{t2}(J_{3}J_{7}) }{supp(J_{3}J_{7
})} + \frac{U(J_{3}J_{7})}{d_{t_{1}}(J_{3}J_{7})}+\frac{U(J_{3}J_{7})}{d_{t_{2}}(J_{3}J_{7})}. \\
\]
\[
Bf(J_{6}J_{5})= \frac{14 + 18 }{2} + \frac{9}{14}+\frac{9}{18} = 17,14 > minbenefit\] 
\[Bf(J_{3}J_{7})= \frac{30 }{2} + \frac{9}{18}+\frac{9}{12} = 16,25 < minbenefit
\]
\end{example}
Remarque: La formule mathématique pour calculer le bénéfice d'une séquence se compose de deux parties:\\
($i$) La contribution de la séquence dans tous les mouvements des véhicules dans la base de données séquentielle.\\
($ii$) La contribution des éléments de la séquence, qui sont des croisements, dans chaque transaction de la base de données séquentielle.\\

La séquence $s$ est dite plus importante lorsque la valeur de $Bf(s)$ est plus élevée.\\
Pour favoriser les séquences de grandes importances avec le minimum de croisement possible nous définissons une nouvelle métrique que nous appelons \underline{Ratio} d'une séquence.

\begin{definition}\textsc{\textbf{\textsc{(}Ratio d'une séquence \textsc{)}}}\mbox{}\newline
Le Ratio d'une séquence $s$, noté $R(s)$, est défini comme suit: 
\[
R(s)= \frac{|s|}{Bf(s)} (s \in \mathcal{FS} ~~et~~ Bf(s) \geqslant minBenefit)
\]
\end{definition}
\underline{Remarque}: Contrairement à la métrique précédente Bénéfice, les valeurs Ratio des séquences qui sont proche de zéro sont considérées comme étant des séquences très importantes. \\
Considérant deux séquences $s1$ et $s2$ ayant le même Bénéfice $Bf(s_1)=Bf(s_2)$  mais avec deux cardinalités différentes $|s1| =l_1$ et $|s2|=l_2$ respectivement, alors celle qui a la valeur la moins élevée, point de vue nombre de croisement de la séquence, c'est elle qui va avoir la valeur Ratio la moins élevée. D'où elle est considérée comme étant la séquence la plus importante.\\
$R(s_1) < R(s_2)$ Si et Seulement si $l_1 < l_2$.\\ Dans ce cas,  $s_1$ est considéré plus importante que $s_2$.

\begin{example}
Considérant la base de données séquentielle $\mathcal{D}_t$ de l'exemple \ref{cxt22} et un seuil $minbenefit= 17$, la séquence <$J_ {6}J_ {5}$> est considérée comme étant bénéfique puisque son $Bf(J_ {6} J_ {5})$ est supérieur ou égal à $minbenefit$. Cependant, la séquence <$J_ {3}J_ {7}$> est considérée comme étant non bénéfique. Dans ce qui suit, nous calculons le Ratio des séquences <$Bf(J_{6}J_{5})$> et <$Bf(J_{5}J_{7})$> : 

\[
R(J_{6}J_{5})= \frac{|J_{6}J_{5}|}{Bf(J_{6}J_{5})} = \frac{2}{17,14} =  0,116.
\]

\[
R(J_{5}J_{7})= \frac{|J_{5}J_{7}|}{Bf(J_{5}J_{7})} = \frac{2}{16} = 0,125.
\]
Puisque $R(J_{6}J_{5}) < R(J_{5}J_{7})$ alors la séquence $R(J_{6}J_{5})$ est considéré plus importante que $R(J_{5}J_{7})$.
\end{example}

\section{Méthodologie et algorithme proposé}
\subsection{Description de la méthode  \textrm{\emph{MIP}}}
Dans cette sous section, nous décrivons notre troisième approche de déploiement des \textit{RSU}s.  Notre nouvelle approche définit de nouvelles mesures de qualité afin qu'elle puisse classer les croisement et/ou l'ensemble des croisements couramment nommées séquences, comme utiles ou peu utiles. En effet, L'approche  \textrm{\emph{MIP}} se compose de deux parties comme suit: 
\begin{enumerate}
\item \textbf{Chercher un ensemble représentatif à partir de la base de données FiRT }(\textbf{Fi}nd \textbf{R}epresentative \textbf{T}ransactions): \\
il repose sur de nouvelles métriques, que nous avons introduits, pour extraire les séquences les plus importantes, telles que:\\
($i$) l'utilité de la séquence, ($ii$) le bénéfice de la séquence et ($iii$) le Ratio de la séquence.
\item \textbf{Calcul de couverture (CovC)} \textbf{Co}verage \textbf{C}ompute: \\
il calcule un ensemble minimal de croisement, qui représente les emplacements dans lesquels les \textit{RSU}s seront placés. Cet ensemble, couvre l'ensemble de modèles de mobilités séquentielles représentatifs extrait par le composant FiRT dans la première phase.
\end{enumerate}
À partir d'une base de données séquentielle, nous pouvons extraire un grand nombre de modèles de mobilité. Cependant, couvrir tous ces modèles nécessite un nombre écrasant de \textit{RSU}. Pour éviter ce problème, nous cherchons des séquences spécifiques que nous appellerons dans le reste de ce manuscrit \textit{les modèles hautement bénéfiques}.\\
Dans ce qui suit, nous appelons l'ensemble des patterns hautement important  \textrm{\emph{MIP}} (Most Important Pattenrs). Il est à noter que, dans la littérature, un modèle de mobilité est représenté par une séquence de croisement (junction). Dans la suite, nous présentons notre algorithme MIPA (Most Important Patterns Algorithm) qui permet de trouver les modèles les plus importants  \textrm{\emph{MIP}} à partir de la trajectoire des véhicules\\

\subsection{Algorithme proposé}
L'idée derrière  \textrm{\emph{MIP}} est de trouver les patterns de mobilité les plus importants et les plus fréquents qui respecte un seuil $minbenefit$. Pour ce faire, nous nous appuyons tout d'abord sur l'algorithme CEPN \cite{fou19} pour extraire les modèles de mobilité à haute utilité (c'est-à-dire les séquences à haute utilité) à partir des trajectoires des véhicules. En outre, nous définissons l'algorithme MIPA (c.f. algorithme \ref{mipa}), qui calcule et sélectionne les séquences bénéfiques parmi celles qui sont utiles (voir l'architecture de la méthode dans la Figure \ref{archMIP}).
\begin{figure}[!ht]
  \begin{center}
  \label{pr}
  \includegraphics[scale=0.6]{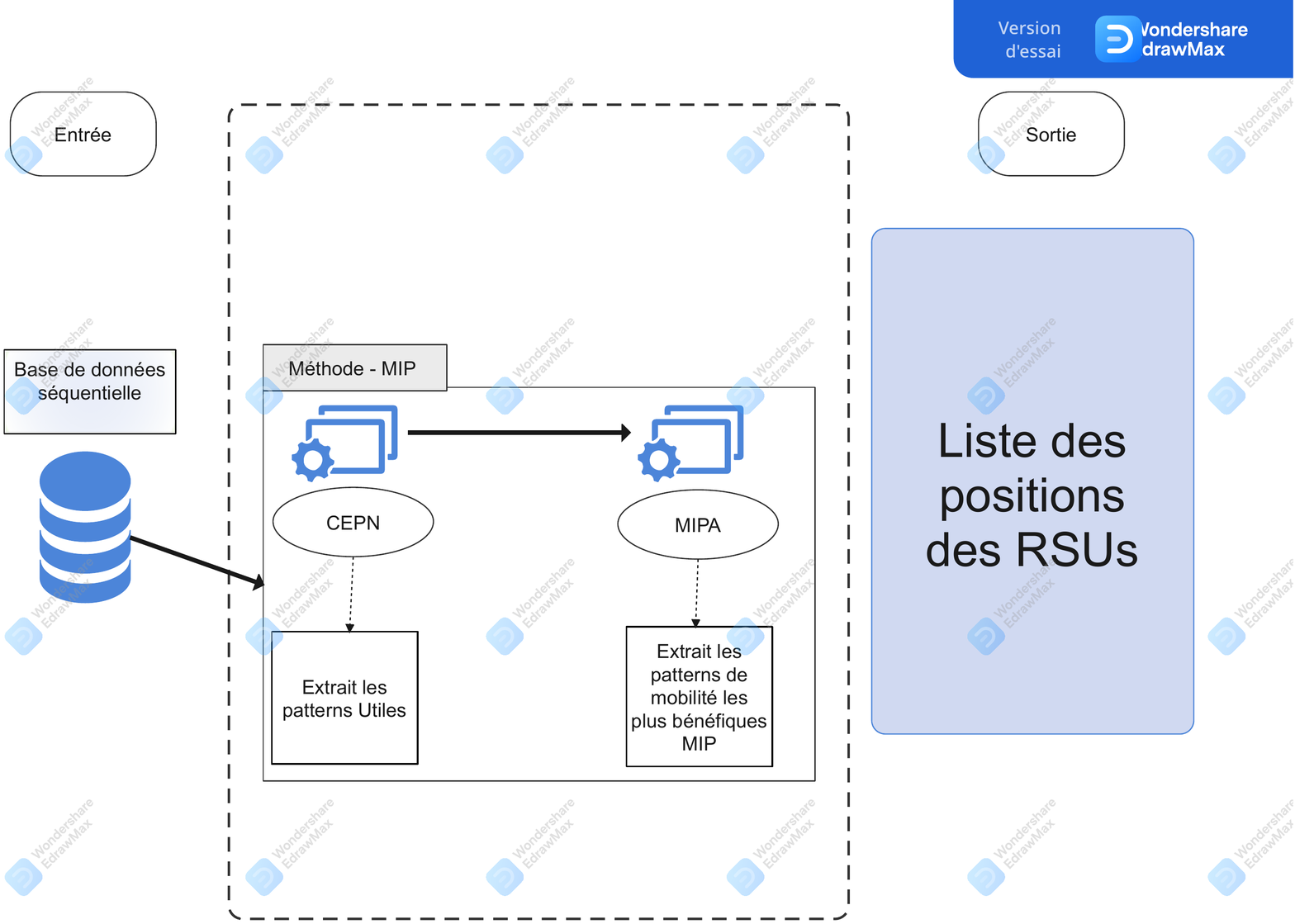}
  \caption{L'architecture de la méthode  \textrm{\emph{MIP}}.} \label{archMIP}
\end{center} 
\end{figure}

\begin{algorithm}
\BlankLine
  \KwIn{$\mathcal{FS}$: set of frequent sequences, minbenefit: threshold of benefit.}
  \KwOut{ \textrm{\emph{MIP}}: set of most important patterns.  \textrm{\emph{MIP}} is sorted according to the ratio of sequences.}
  \BlankLine
  $\mathcal{MIP} = \emptyset$\\
\ForEach{$E$ of $\mathcal{FS}$}
        {
        \BlankLine
         \If{(getBenefit($E$)$>=$ $minbenefit$)}
         {insertIntoMIP($E$, $getRatio(E)$)}
         {}
        }
  return( \textrm{\emph{MIP}})
\caption{MIPA}
\label{mipa}
\end{algorithm}

\begin{proposition}\label{p1}
L'ensemble des modèles importants \textrm{\emph{MIP}} ne coïncide pas avec l'ensemble des séquences fermées.
\end{proposition}
\begin{proof} Dans ce qui suit, nous prouvons la validité de la proposition \ref{p1}. Pour se faire, nous admettons à une démonstration à l'aide du raisonnement par l'absurde. \\

\textbf{Hypothèse}: Nous considérons que $Pm$ est un élément de  \textrm{\emph{MIP}} (pattern important) et $P$ est un pattern fermé.\\

\textbf{But}: Prouvons que $Pm$ est différent de $P$\\

Soit $P$ est un pattern fermé ce qui implique que :\\
\begin{equation}
supp(P)=supp(Pm) (P pattern fermé)
\end{equation}
d'où :
\begin{equation}
d_{t}(P)>d_{t}(Pm) \forall t \in \mathcal{T}
\end{equation}
\begin{equation}
d_{t}(P) = d_{t}(Pm)+c_{t}
\end{equation}
\begin{equation}
Bf(P) = \frac{d_{t}(P)}{supp(P)} + \frac{U(p)}{d_{t}(P)}
\end{equation}
\begin{equation}
Bf(Pm)=  \frac{d_{t}(Pm)}{supp(Pm)} + \frac{U(pm)}{d_{t}(Pm)}
\end{equation}
puisque 
\begin{equation}
d_{t}(P)  > d_{t}(Pm) \longrightarrow \frac{d_{t}(P)}{supp(P)} > \frac{d_{t}(Pm)}{supp(Pm)}
\end{equation}
et 
\begin{equation}
\frac{U(P)}{d_{t}(P)}=\frac{U(Pm)}{d_{t}(Pm)} + \frac{U(P)-U(Pm)}{d_{t}(Pm)} > \frac{U(Pm)}{d_{t}(Pm)}
\end{equation}
alors, 
\[ U(P) > U(Pm) \]
d'après (5.4) et (5.5) nous évaluons $\frac{R(P)}{R(Pm)} = \frac{|P|}{|Pm|} * \frac{U(Pm)}{U(P)}$ \\
Nous avons : \\
$\frac{|P|}{|Pm|} > 1$. (La valeur de $\frac{|P|}{|Pm|}$ est représenté par $n$ dans la Figure \ref{pr}).\\
Reste à rechercher un minimum pour l'expression suivante: $\frac{U(Pm)}{U(P)} > = < ?$\\
Nous avons : 
\begin{equation}
\frac{\sum d_{t}(Pm)}{\sum d_{t}(Pm) + supp(P) * \alpha  } avec \alpha  = d_{t}(P)-d_{t}(Pm)
\end{equation}
\begin{equation}
\longrightarrow \frac{1}{1 + \frac{supp(P) * \alpha}{\sum d_{t}(Pm)}  }
\end{equation}
\begin{equation}
\longrightarrow supp(P) * \alpha < \sum d_{t}(Pm)
\end{equation}
d'après (5.7) l'équation 
\begin{equation}
1 + \frac{supp(P) * \alpha}{\sum d_{t}(Pm)} < 2  \longrightarrow \frac{1}{1 + \frac{supp(P) * \alpha}{\sum d_{t}(Pm)}  } > 1/2
\end{equation}


Dans la partie précédente, nous prouvons que $\frac{1}{1 + \frac{supp(P) * \alpha}{\sum d_{t}(Pm)}}$ est toujours supérieure à 1/2 quelle que soit la valeur de $\frac {|P|}{|Pm|}$. Dans le cas où la valeur de $\frac{|P|}{|Pm|} $ est supérieure à 1 (le cas où le numérateur est supérieur au dénominateur: prouvé par l'équation (5.1)) la fonction est toujours supérieure à 1 comme représenté sur la figure \ref{pr}.
D'où $P$ est différent de $Pm$.
\end{proof}

\begin{figure}[!ht]
  \begin{center}
  \label{pr}
  \includegraphics[scale=0.42]{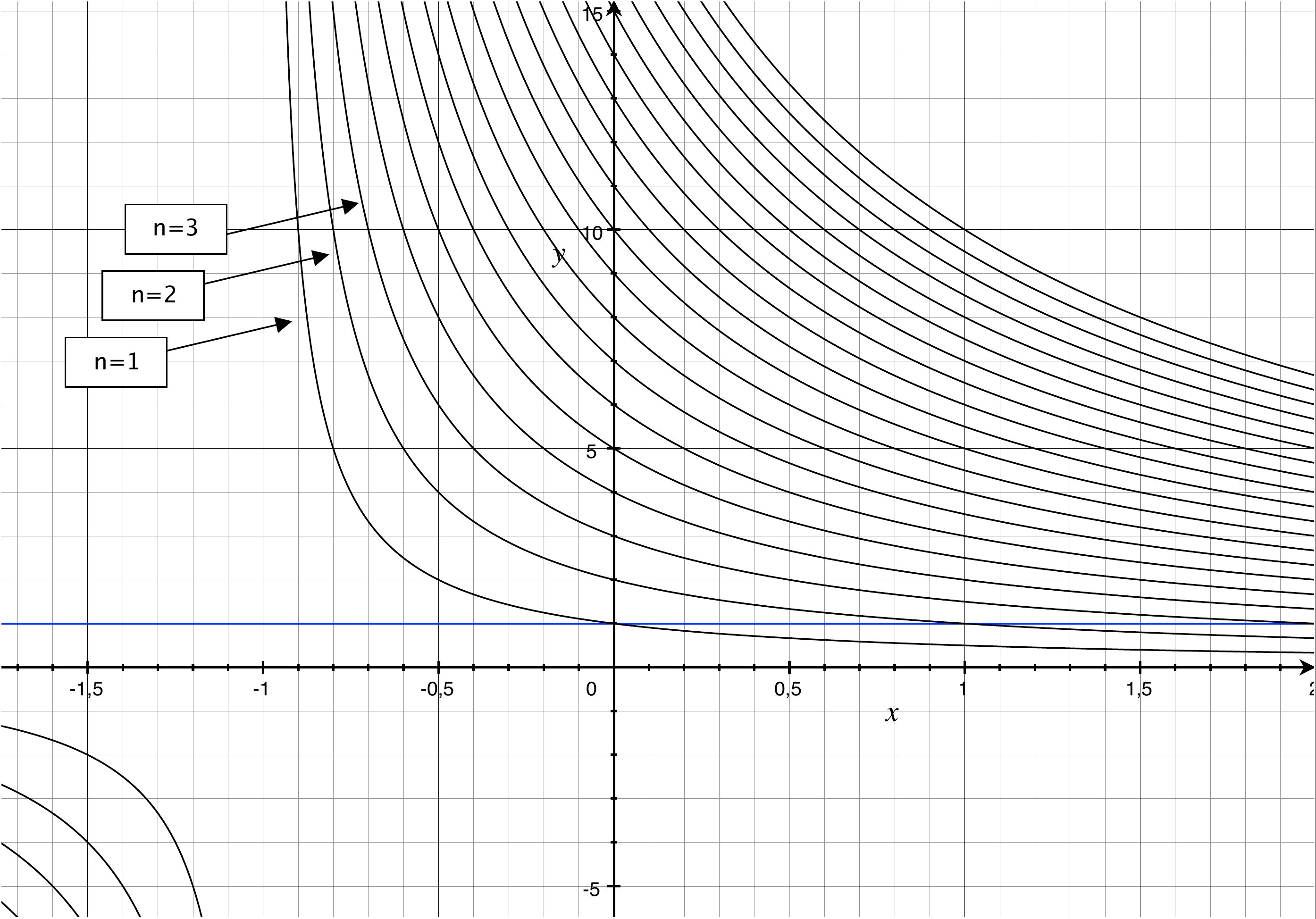}
  \caption{Evolution du métrique Ratio par la variation des paramètres cardinalité et utilité.} \label{pr}
\end{center}
\end{figure}

\begin{figure}[!ht]
  \begin{center}
   \includegraphics[scale=0.42]{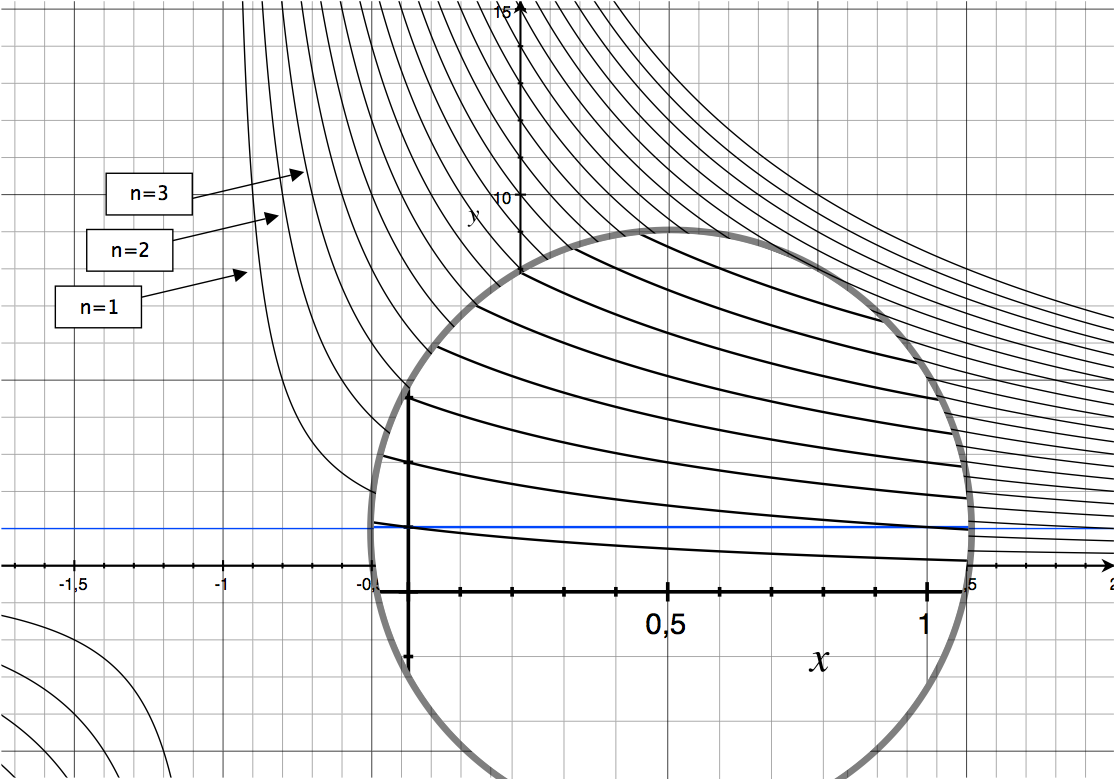}
  \caption{Agrandissement d'une partie de la courbe de l'évolution du métrique Ratio par la variation des paramètres cardinalité et utilité.} \label{preuveBis}
\end{center}
\end{figure}

\subsection{Exemple illustratif}
\begin{example}
Considérant la base de données séquentielle $\mathcal{D}$ de l'exemple \ref{cxt2}, seuil $minsupp = 1/7$ et $minbenefit=17$.\\
Afin d'extraire les modèles de mobilité les plus bénéfiques notre algorithme MIPA procède comme suit:

\begin{itemize}
\item Premièrement, il invoque l'algorithme CEPN pour extraire les séquences fréquentes:

\[
\mathcal{FS} =\{J_{2};J_{3};J_{5};J_{6};J_{7}; J_{2}J_{6};J_{2}J_{7};J_{5}J_{3};J_{3}J_{6}; J_{3}J_{7};
 J_{6}J_{7}; J_{6}J_{3};J_{6}J_{5} ; J_{2}J_{6}J_{7} ; J_{5}J_{3}J_{7}\}
\]
L'ensemble $\mathcal{FS}$ est réduit en effectuant les calculs suivants :

\[
Bf(J_{3}J_{6}) = \frac{d_{t0}(J_{3}J_{6}) + d_{t1}(J_{3}J_{6}) }{supp(J_{3}J_{6})} + \frac{U(J_{3}J_{6})}{d_{t_{0}}(J_{3}J_{6})}+\frac{U(J_{3}J_{6})}{d_{t_{1}}(J_{3}J_{6})}.
\]

\[
Bf(J_{6}J_{5})= \frac{d_{t0}(J_{6}J_{5}) + d_{t1}(J_{6}J_{5}) }{supp(J_{6}J_{5})} + \frac{U(J_{6}J_{5})}{d_{t_{0}}(J_{6}J_{5})}+\frac{U(J_{6}J_{5})}{d_{t_{1}}(J_{6}J_{5})}. \\
\]
\[
Bf(J_{6}J_{3}) = \frac{d_{t0}(J_{6}J_{3}) + d_{t1}(J_{6}J_{3}) }{supp(J_{6}J_{3})} + \frac{U(J_{6}J_{3})}{d_{t_{0}}(J_{6}J_{3})}+\frac{U(J_{6}J_{3})}{d_{t_{1}}(J_{6}J_{3})}.
\]
\[
Bf(J_{5}J_{7})= \frac{d_{t1}(J_{5}J_{7}) + d_{t2}(J_{5}J_{7}) }{supp(J_{5}J_{7})} + \frac{U(J_{5}J_{7})}{d_{t_{1}}(J_{5}J_{7})}+\frac{U(J_{5}J_{7})}{d_{t_{2}}(J_{5}J_{7})}. \\
\]
\[
Bf(J_{5}J_{3}) = \frac{d_{t0}(J_{5}J_{3}) + d_{t1}(J_{5}J_{3}) + d_{t2}(J_{5}J_{3}) }{supp(J_{5}J_{3})} + \frac{U(J_{5}J_{3})}{d_{t_{0}}(J_{5}J_{3})}+\frac{U(J_{5}J_{3})}{d_{t_{1}}(J_{5}J_{3})} + \frac{U(J_{5}J_{3})}{d_{t_{2}}(J_{5}J_{3})}.
\]
\[
Bf(J_{2}J_{6}J_{7})= \frac{d_{t5}(J_{2}J_{6}J_{7}) + d_{t6}(J_{2}J_{6}J_{7}) }{supp(J_{2}J_{6}J_{7})} + \frac{U(J_{2}J_{6}J_{7})}{d_{t_{5}}(J_{2}J_{6}J_{7}}+\frac{U(J_{2}J_{6}J_{7})}{d_{t_{6}}(J_{2}J_{6}J_{7})}. 
\]
\[
Bf(J_{5}J_{3}J_{7})= \frac{d_{t5}(J_{5}J_{3}J_{7}) + d_{t6}(J_{5}J_{3}J_{7}) }{supp(J_{5}J_{3}J_{7})} + \frac{U(J_{5}J_{3}J_{7})}{d_{t_{5}}(J_{5}J_{3}J_{7}}+\frac{U(J_{5}J_{3}J_{7})}{d_{t_{6}}(J_{5}J_{3}J_{7})}.
\]
\item Deuxièmement, il appelle notre algorithme MIPA, qui utilise l'ensemble $\mathcal{FS}$ pour extraire l'ensemble de modèles de mobilité les plus bénéfique suivant en calculant les Ratios des motifs ayant des valeurs de Bénéfice supérieur à minbenefit :\\

\[
R(J_{6}J_{5})= \frac{|J_{6}J_{5}|}{Bf(J_{6}J_{5})} = \frac{2}{17,14} =  0,116.
\]
\[
R(J_{6}J_{3})= \frac{|J_{6}J_{3}|}{Bf(J_{6}J_{3})} = \frac{2}{17,5} = 0,114.
\]
D'où l'ensemble \emph{MIP} :
\[\textrm{\emph{MIP}} =\{<J_{6}J_{3}>;<J_{6}J_{5}>\}. \]
\end{itemize}
\end{example}

\section{Étude de la complexité et de la validité}
Dans cette section, nous prouvons la complétude et la terminaison de notre algorithme qui implémente la méthode  \textrm{\emph{MIP}} puis nous calculons la complexité de cette méthode au pire des cas.\\

\subsection{Terminaison et complétude de l'algorithme MIPA}.
\begin{proposition}
L'algorithme MIP termine.
\end{proposition}
\begin{proof}
Le nombre des \textit{RSU}s générés par l'algorithme est un nombre fini. En effet, le nombre $k$ des \textit{RSU}s qui représente les emplacements des croisements les plus utiles, est au pire des cas,  égal à la cardinalité de $|\mathcal{J}|$. Puisque l'ensemble fourni par l'algorithme est un nombre fini alors nous pouvons considéré que l'algorithme se termine.
\end{proof}
\begin{proposition}
L'algorithme MIP est complet.
\end{proposition}
\begin{proof}
Le nombre des \textit{RSU}s générés par l'algorithme représente les croisements très utiles selon un seuil utilisateur nommé \textit{minbenefit}, au pire des cas, lorsque ce dernier égale à $1$ l'algorithme fourni la totalité des croisements comme résultat. Donc, l'algorithme MIP est considéré comme complet. 
\end{proof}
\subsection{Complexité théorique}
Dans le cas de la méthode \textrm{\emph{MIP}} nous pouvons déduire sa complexité par la somme des complexités des composant CEPN, MIPA et CoC. {\`A} partir d'une base de données séquentielle de $n$ véhicules et $m$ croisements , nous avons : 
\begin{itemize}
\item Complexité de CEPN : \\
L'algorithme analyse d'abord la base de données pour calculer le support de chaque séquence. La complexité de cette étape est  $\mathcal{O}(m \times n)$. Ensuite, un élagage se fait pour éliminer les séquences non fréquentes et une recherche récursive, des séquences plus grande en terme de taille, à l'aide d'une recherche en profondeur d'abord. Pour chaque appel récursif, l'algorithme nécessite un temps de $\mathcal{O}(m\times n)$ pour créer une base de données projetée des séquences de préfixe, puis il nécessite un temps de $\mathcal{O}(m)$ pour calculer le support dans cette base projetée. 
\item Complexité de MIPA : \\
L'algorithme calcul les valeurs de l'utilité, le bénéfice et le ratio de chaque croisement. Cette phase nécessite un coût  de $\mathcal{O}(1)$ puisque le calcul de ces métriques nécessite uniquement l'information support qui a été évalué dans l'algorithme CEPN. Ensuite, l'ordonnancement des croisements, selon la valeur du Ratio, nécessite un coût  de  $\mathcal{O}(m)$ au pire des cas.
\item Complexité de CoC \\
{\`A} partir d'un hypergraphe de $a$ sommets et $b$ hyperarêtes nous avons une complexité exponentielle, au pire des cas, de l'ordre de $\mathcal{O}(b * a^{b+1})$ avec $a=|\mathcal{J}|$ et $b=|\mathcal{E}|$, les tests pour vérifier une traverse minimale s'effectuent en $\mathcal{O(1)}$.
\end{itemize}
En effet, la complexité théorique, au pire des cas, de l'algorithme MIPA est égale à :\\
$\mathcal{O}(n* m) + \mathcal{O}(n* m) + \mathcal{O}(m)+ \mathcal{O}(m) +\mathcal{O}(b * a^{b+1})$ avec $a$ et $b$ toujours inférieur à $m$ et $n$ respectivement.
D'où la complexité peut s'écrire : 
$\mathcal{O}(n*m + (b * a^{b+1}))$
\section{Conclusion}
L'objectif principal de la méthode  \textrm{\emph{MIP}}  est d'assurer une couverture élevée en utilisant le plus petit nombre de \textit{RSU}. En outre, l'entrée de la méthode est une base de données séquentielle enrichi par une information \underline{utilité} pour chaque transaction et fournit en sortie un ensemble réduit de trajectoire dite "trajectoire avec hautes utilités", à partir de ces dernières la méthode choisis quelques emplacements dans lesquels déploie les unités de bord de la route couramment nommé \textit{RSU}. En effet, ces emplacements sont nommés des séquences à hautes utilités, le calcul de ces dernières se fait en utilisant trois nouvelles métriques qui sont respectivement: Utility, Benefit et Ratio.\\
La méthode  \textrm{\emph{MIP}} est composé de deux parties à savoir CEPN et MIPA.

\newpage
\thispagestyle{empty}
\mbox{}
\newpage
\chapter{Expérimentations et Résultats}
\minitoc
\section{Introduction}
Dans ce chapitre, nous présentons les différentes métriques d'évaluation utilisés, les mesures d'efficacités et des performances afin d'établir une étude expérimentale entre les différentes approches et méthodes de la littérature. De plus, nous démontrons   les performances des méthodes  \textrm{\emph{SPaCov+}} et \textrm{\emph{HeSPiC}} vs. Les autres méthodes de la littérature en termes d'efficacité et de performance. Les méthodes sont testées sur un simulateur du trafic routier afin de bien les évalués.\\
Une étude expérimentale pour chaque système est détaillée. Dans chaque étude expérimentale différents scénarios de test seront discutés. Nous présentons pour chaque étude les résultats expérimentaux associés à chaque scénario de test.

\section{Métriques d'évaluation}
Dans cette section, nous présentons les métriques d'évaluation utilisées pour comparer les méthodes. D'une manière générale, l'évaluation d'une méthode porte sur deux points importants à savoir l'efficacité et la performance. Le premier critère "efficacité" signifie le taux de ressemblance du résultat retourné par rapport aux besoins de l'utilisateur, un taux de ressemblance élevé signifie une meilleure efficacité de résultat. Le deuxième critère "performance" signifie le temps de la requête qui est la somme entre le temps de calcul et le temps de réponse, lorsque ce temps tend vers zéro, la performance devient plus importante. Pour cela, plusieurs mesures d'efficacité et de performance sont proposées dans la littérature.

\subsection{Mesure d'efficacité}
Les mesures d'efficacité les plus utilisées sont la \underline{précision} et le \underline{rappel}. Rappelons que la précision notée Pr, est la capacité qu'un système retourne tous les résultats pertinents. D'une manière formelle, la précision peut s'écrire de la manière suivante:\\ 
\[ 
Pr = \frac{|Prt| |Ret|}{|Ret|} avec |Ret| \neq 0
\]
Le rappel, noté R, reflète la capacité du système à retourner tous les résultats pertinents. D'une manière formelle le rappel s'écrit de la manière suivante: 
\[ 
R = \frac{|Prt| |Ret|}{|Prt|} avec |Prt| \neq 0
\]
Où:\\
Prt: L'ensemble des résultats pertinent dans le système pour une requête données.\\
Ret: les résultats retournés par le système.
\subsection{Mesure de performance}
Les mesures de performance les plus utilisées sont : 
\begin{itemize}
    \item Taux de couverture ou Coverage ratio.
    \item Coût ou Cost.
    \item Temps de latence ou Latency.
    \item Surcharge ou Overhead.
\end{itemize}
Dans cette partie nous définissons chaque métrique.
\begin{enumerate}
    \item \textbf{Taux de couverture}: Il évalue le pourcentage de couverture. Le taux de couverture pour un message d'événement donné \textit{e} est défini comme suit: \\

\[ Taux ~ de ~couverture~ pour~ un ~message~ = \frac{|V^{'}|}{|V|} \]
Où:
\begin{itemize}
\item  $V^{'}$: L'ensemble de véhicules informés (\textit{i.e.}, ceux qui ont reçu un message).
\item  $V$: Tous les véhicules en mouvement.
\end{itemize}

\item \textbf{Coût}: Le nombre total des \textit{RSU}s utilisées pour couvrir une zone considérée.

\item \textbf{Temps de latence}: Le temps nécessaire pour délivrer un message d'un \textit{RSU} à un véhicule. Le temps de latence moyen noté AL est défini comme suit: 
\[AL=\frac{\sum (t_{i}-T)} {NumberOfVehicles} \]

\item \textbf{Surcharge}: Le nombre de messages envoyés.
\end{enumerate}

\section{Environnement de simulation}
Dans cette section, nous présentons notre environnement de simulation utilisé pour évaluer et tester nos méthodes vs. Les autres méthodes de la littérature. Dans ce contexte plusieurs simulateurs ont été développés en tenant compte des besoins des chercheurs. Dans notre cas, nous avons utilisé deux simulateurs complémentaires: le premier nommé SUMO pour la simulation du trafic routier et le deuxième nommé OMNET++ pour la simulation du trafic réseaux. Dans ce qui suit, nous détaillons les deux simulateurs choisis.

\subsection{Simulateur SUMO}
Simulation of Urban MObility (SUMO) \cite{B1} est un progiciel de simulation de trafic multi-modal open source, portable, microscopique et continu conçu pour gérer de grands réseaux. La simulation du trafic au sein de SUMO utilise des outils logiciels pour la simulation et l'analyse du trafic routier et des systèmes de gestion du trafic. De nouvelles stratégies de trafic peuvent être mises en œuvre via une simulation pour les analyser avant d'être utilisées dans des situations réelles. Plusieurs outils et logiciel ont été développés au dessus de la couche haute du logiciel de simulation SUMO à l'instar de: MOVE, CIVITAS, OMNET++, etc.
\subsection{Simulateur OMNET++}
Objective Modular Network Testbed in C++ (OMNET++) \cite {OMNET10} est une bibliothèque et un cadre de simulation C ++ modulaires basé sur des composants principalement pour la construction de simulateurs de réseau. OMNET++ seul est un framework de simulation sans modèles pour les protocoles réseau les plus connus comme IP et HTTP. Les principaux modèles de simulation de réseaux informatiques sont disponibles dans d'autres framworks externes, qui peuvent s'ajouter à OMNET, comme le cas du framework INET qui offre une variété de modèles pour tous les types de protocoles réseaux et des technologies à l'exemple de IPv6, BGP (Border Gatway Protocol), etc.


\section{Étude Expérimentale des solutions proposées}
\subsection{Configuration et simulation}
Dans cette section, nous évaluons les performances respectives de  \textrm{\emph{SPaCov}},  \textrm{\emph{SPaCov+}} et  \textrm{\emph{HeSPiC}} via le simulateur OMNET ++ \cite{OMNET10}. Nous évaluons les performances de nos méthodes de couverture et nous comparons les résultats de deux premières méthodes ( \textrm{\emph{SPaCov}} et  \textrm{\emph{SPaCov+}} à ceux obtenus par MPC \cite{B1}. Pour la troisième méthode,  \textrm{\emph{HeSPiC}}, nous évaluons ses performances pour différents paramètres comme le nombre de véhicules, le numéro \textit{RSU}, la zone couverte, la portée de communication et le nombre de croisement. Les paramètres de simulation sont résumés dans la sous-section suivante.
\begin{table}[ht]
\centering

\label{caract}
\begin{tabular}{|c|c|c|}
 \hline \cline{1-3}
           \multirow{3}{*}{Physical layer} & Frequency band &  5.9 GHz \\  \cline{2-3}
            &Transmission power & 40 mW \\ \cline{2-3}
            &Bandwidth &10 MHz\\ \cline{2-3}
     \hline \hline
      \multirow{4}{*}{Link layer} & Bit rate &  6 Mbit/s \\  \cline{2-3}
            &CW &15.1023\\ \cline{2-3}
            &Slot time&13 us\\ \cline{2-3}
            &SIFS&32 us\\ \cline{2-3}
       \hline \hline
 \multirow{3}{*}{Scenarios} & Message size &  2312 Bytes \\  \cline{2-3}
            &Message frequency &0.5 Hz\\ \cline{2-3}
            &Slot time&13 us\\ \cline{2-3}
            &\#Runs&30 times\\ \cline{2-3}
            &Communication Range (Km) & 0.3 \\ \cline{2-3}
  \hline \cline{1-3}
\end{tabular}
\caption{Paramètres de simulation}\label{caract}
\end{table}

\begin{table}[ht]
 \centering
 \begin{center}

{
\begin{tabular}{||p{3.5cm}  | p{3cm}  | p{3cm} | p{2cm} ||}
\hline
Map&Number of Raods&Number of Junctions &Complexity\\
\hline
\hline
SETIF& 462& 381& Low\\

TOZEUR & $812$&  $624$& Moderate\\

Menzel Bourguiba& $1218$ & $1109$ & High\\

\hline
\end{tabular}
}
\caption{Caractéristiques des topologies routières}\label{MapTab}
\end{center}
\end{table}

\begin{figure}[ht]
  \begin{center}
  \includegraphics[scale=0.6]{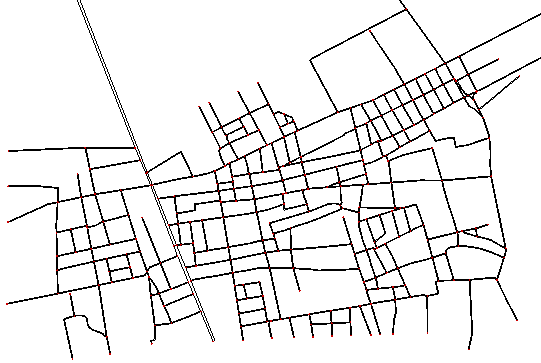}
  \caption{Exemple de la topologie routière de la ville de Tozeur.} \label{map1}
\end{center}
\end{figure}

\begin{figure}[ht]
  \begin{center}
  \includegraphics[scale=0.6]{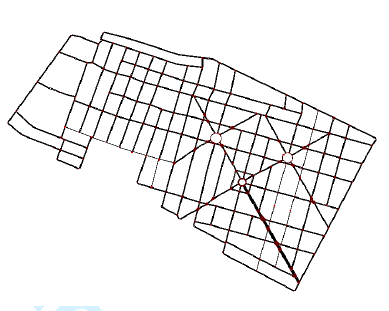}
  
  \caption{Exemple de la topologie routière de ville de Menzel Bourguiba.} \label{map2}
\end{center}
\end{figure}

\begin{figure}[ht]
  \begin{center}
  
   \includegraphics[scale=0.6]{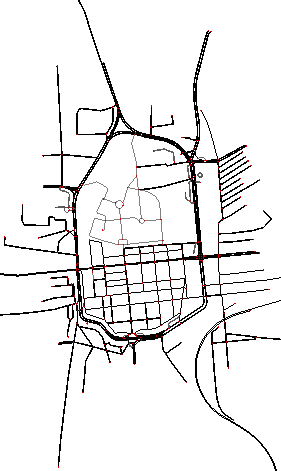}
  \caption{Exemple de la topologie routière de la ville de Setif.} \label{map3}
\end{center}
\end{figure}

La mobilité urbaine est simulée par simulation microscopique et continue du trafic routier à l'aide du SUMO \cite{SUMO12}. De plus, nous utilisons l'environnement OMNeT ++ pour réaliser la simulation de réseau. Les paramètres de simulation sont résumés dans le Tableau \ref{caract}. Dans ce qui suit, nous décrivons l'ensemble de jeux de données et les scénarios de tests.\\

\textbf{Jeux de données:} Les expériences présentées, dans ce manuscrit, sont basées sur un ensemble de données de trajectoires réelles représentant des portions de zones urbaines des villes de Sétif à l'Algérie Figure \ref{map3}, Tozeur au sud Tunisien Figure \ref{map1} et Menzel Bourguiba au nord de la Tunisie Figure \ref{map2}. Les cartes de ces villes sont téléchargées depuis le site \textit{OpenStreetMap} \footnote{https://www.openstreetmap.org/} sous forme de fichiers OSM. Ensuite, ils sont convertis en fichiers réseau, contenant les routes et les intersections sur lesquelles les véhicules simulés circulent. Les caractéristiques des topologies routières de chaque plan de ville sont résumées dans le Tableau \ref{MapTab}.\\

\subsection{Scénario de simulation}
Dans cette partie expérimentale, les objectifs sont d'évaluer l'efficacité et la performance de nos approches par rapport aux autres stratégies de la littérature. Nous testons les méthodes dans différentes conditions et avec différents scénarios à savoir le nombre de véhicules, les surfaces des cartes utilisées, les topologies des routes et la portée de communication pour les \textit{RSU}s. En effet, nous avons utilisé entre $200$ et $1000 $ véhicules en mouvement pour des surfaces couvertes entre $2$ $km^{2}$ et $4,8$ $km^{2}$. De même, nous avons choisi trois topologies de carte différentes: éparse, moyenne et dense afin de bien évaluer les performances de nos méthodes dans chaque type de topologie. En premier lieu, nous comparons les méthodes  \textrm{\emph{SPaCov}},  \textrm{\emph{SPaCov+}} et  \textrm{\emph{MPC}} en changeant quelques paramètres à savoir la surface couverte, le nombre de véhicules, la densité afin d'évaluer leurs impacts sur le nombre de \textit{RSU}s à déployer, le temps de latence et l'overhead.\\ Cependant,  \textrm{\emph{HeSPiC}} ne peut pas être comparé avec les autres approches citées ci-dessus puisqu'elle fait partie de la famille des méthodes avec une contrainte budgétaire, ce qui signifie que la couverture offerte par  \textrm{\emph{HeSPiC}} est limité par un seuil budgétaire en terme de nombre de \textit{RSU}. Nous avons testé les composants de la méthode  \textrm{\emph{HeSPiC}} à savoir la fonction Poids, la fonction Probabilité et la fonction Distance sur plusieurs paramètres comme la portée de communication, le nombre de véhicules, le taux de couverture et le nombre de \textit{RSU}. En second lieu, nous comparons  \textrm{\emph{SPaCoV+}},  \textrm{\emph{MPC}} et  \textrm{\emph{MIP}} sur une base de test pré-traitée afin que nous pouvons évaluer les intérêts de chaque croisement à savoir la surface couverte, le nombre de véhicules et la densité, sont utilisés pour évaluer le nombre de \textit{RSU} déployer par chaque méthode.\\



\section{Résultats d'expérimentations}
Dans ce qui suit, nous rapportons d'abord les résultats de  \textrm{\emph{SPaCov}} et  \textrm{\emph{SPaCov+}} par rapport à MPC. Par la suite, nous présentons les résultats de  \textrm{\emph{MIP}} par rapport à  \textrm{\emph{SPaCoV+}} et  \textrm{\emph{MPC}}. À la fin, nous présentons les résultats obtenus par l'exécution de  \textrm{\emph{HeSPiC}} tout en changeant les conditions des tests et les paramètres des fonctions interne de la méthode (poids, probabilité et distance).\\

\subsection{Résultats des méthodes  \textrm{\emph{SPaCov}} et  \textrm{\emph{SPaCov+}}}

La Figure \ref{CR} montre l'évolution du taux de couverture de  \textrm{\emph{SPaCov}} et  \textrm{\emph{SPaCov+}} par rapport à la méthode MPC tout en changeant le nombre de véhicules pour la ville de Setif. La Figure \ref{CR:MB} et la Figure \ref{CR:T} montre que le taux de couverture est inversement proportionnel à la densité du trafic indépendamment du type de la carte. En fait, le taux de couverture diminue lorsque le nombre de véhicules augmente. Ceci peut s'expliquer par l'effet de la surcharge du réseau et des messages d'erreurs dus à l'augmentation du nombre de communications établies sur la même surface et au même instant. Par conséquent, ce phénomène augmente le nombre de véhicules qui ne recevront pas les messages, ce qui diminue le taux de couverture. Nous notons que notre méthode  \textrm{\emph{SPaCov+}} est toujours meilleure pour les différentes topologies de cartes par rapport aux deux autres méthodes en terme de taux de couverture. En effet, cela est dû au nombre supplémentaire des \textit{RSU}s utilisées et aux zones dans lesquelles elles sont placées (elle couvre deux types de zones: les zones très peuplées et les zones peu fréquentées).\\ 
Les deux méthodes  \textrm{\emph{SPaCov}} et  \textrm{\emph{MPC}} fournissent les mêmes résultats, pour les deux topologies routières représentées par les villes Menzel Bourguiba et Tozeur, lorsque le nombre de véhicules ne dépasse pas les 400, au delà de ce nombre  \textrm{\emph{SPaCov}} dépasse  \textrm{\emph{MPC}} en termes de couverture, cela peut être expliqué par le fait que  \textrm{\emph{SPaCov}} utilise un jeu de données ordonné, contrairement à  \textrm{\emph{MPC}}, ce qui rend l'emplacement des \textit{RSU}s plus sophistiqué pour le cas de  \textrm{\emph{SPaCoV}}.

\begin{figure}[!ht]
\centering
\includegraphics[scale=1]{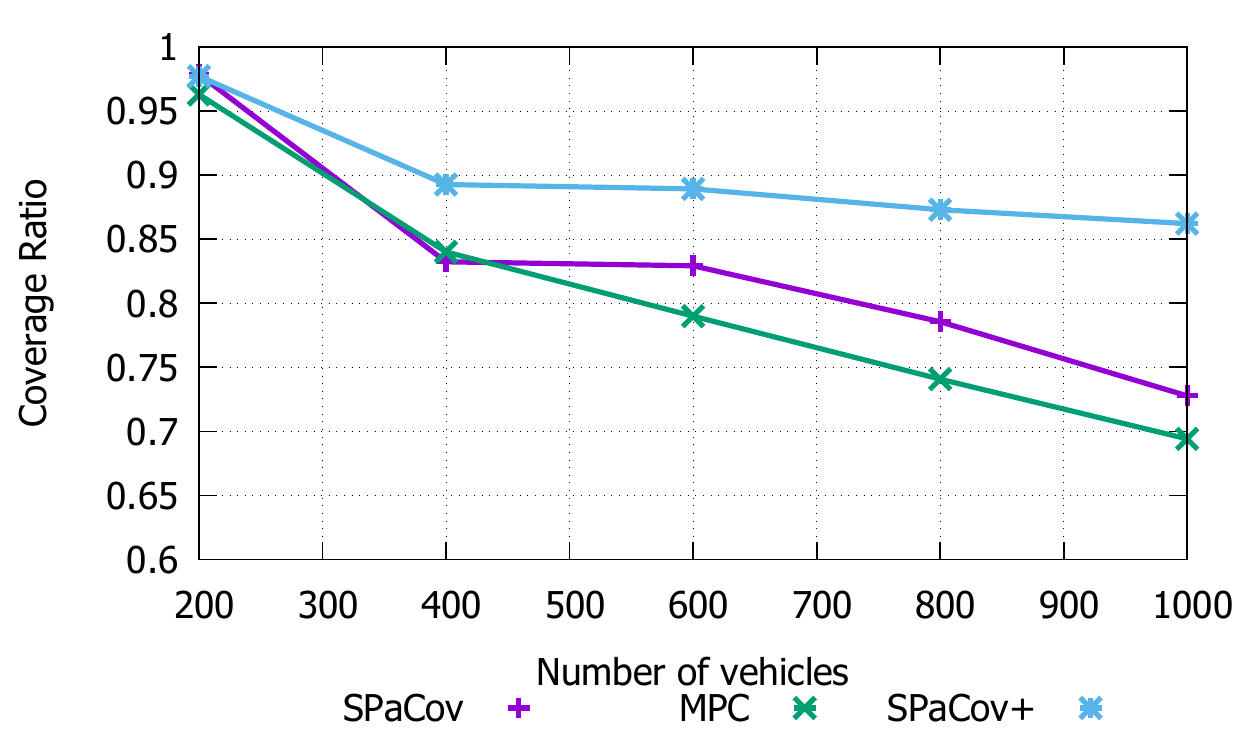}
\caption{Variation du taux de la couverture par rapport au nombre de véhicules pour la ville de Setif.}
\label{CR}
\end{figure}
\begin{figure}[!ht]
\centering
\centering
\includegraphics[scale=1]{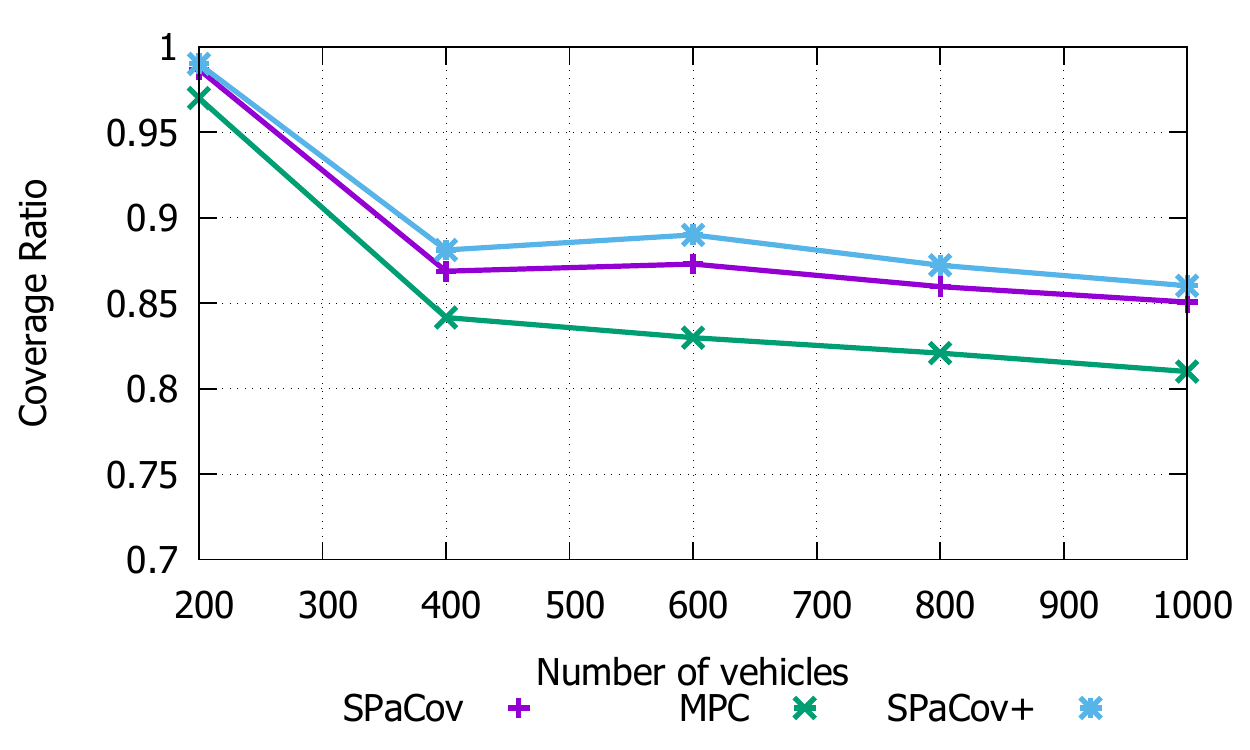}
\caption{Variation du taux de couverture par rapport au nombre de véhicule en mouvement pour la ville de Menzel bourguiba.}
\label{CR:MB}
\end{figure}

\begin{figure}[!ht]
\centering
\includegraphics[scale=1]{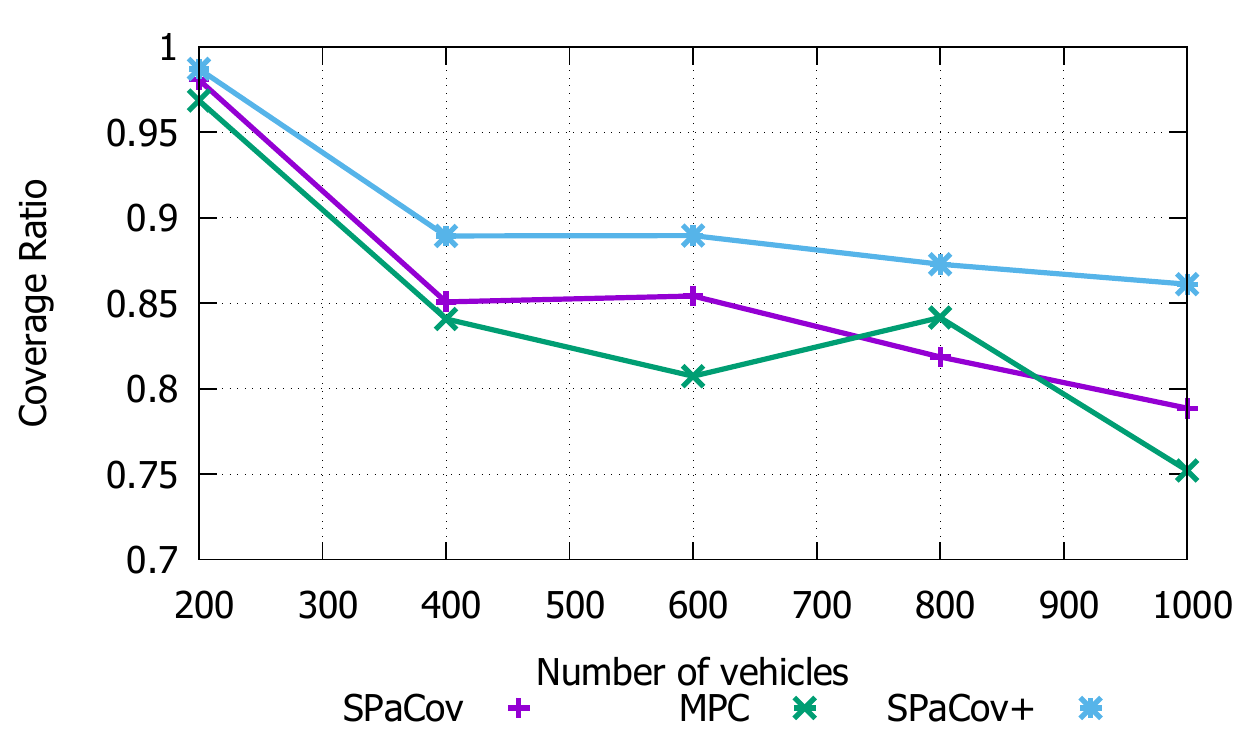}
\caption{Variation du taux de couverture par rapport au nombre de véhicule en mouvement pour la ville de Tozeur.}
\label{CR:T}
\end{figure}

Figure \ref{Lat:MB}, Figure \ref{Lat:T},  et Figure \ref{Lat} montrent l'évolution du temps de la latence de nos méthodes par rapport à la méthode MPC tout en variant le nombre de véhicules sous différentes topologies routières représentées par les villes de Menzel borguiba, Tozeur et Setif. Nous observons que le temps de latence augmente au fur et à mesure que le nombre de véhicules dans le réseau augmente. La Figure \ref{Lat} montre qu'à partir d'un nombre élevé de véhicules (dans l'exemple de Setif: $800$ véhicules) la différence de l'évolution du temps de latence, entre les méthodes, devient remarquable par rapport à un nombre réduit de véhicules (inférieur à $400$). Les trois méthodes considérées ont une latence relative par rapport au nombre de \textit{RSU} déployées par chacune. La méthode qui déploie plus de \textit{RSU} a moins du temps de la latence. Pour toutes les villes considérées dans cette simulation,  \textrm{\emph{SPaCov+}} surpasse ses concurrents  \textrm{\emph{SPaCov}} et  \textrm{\emph{MPC}}.\\
\begin{figure}[ht]
\centering
\centering
\includegraphics[scale=1]{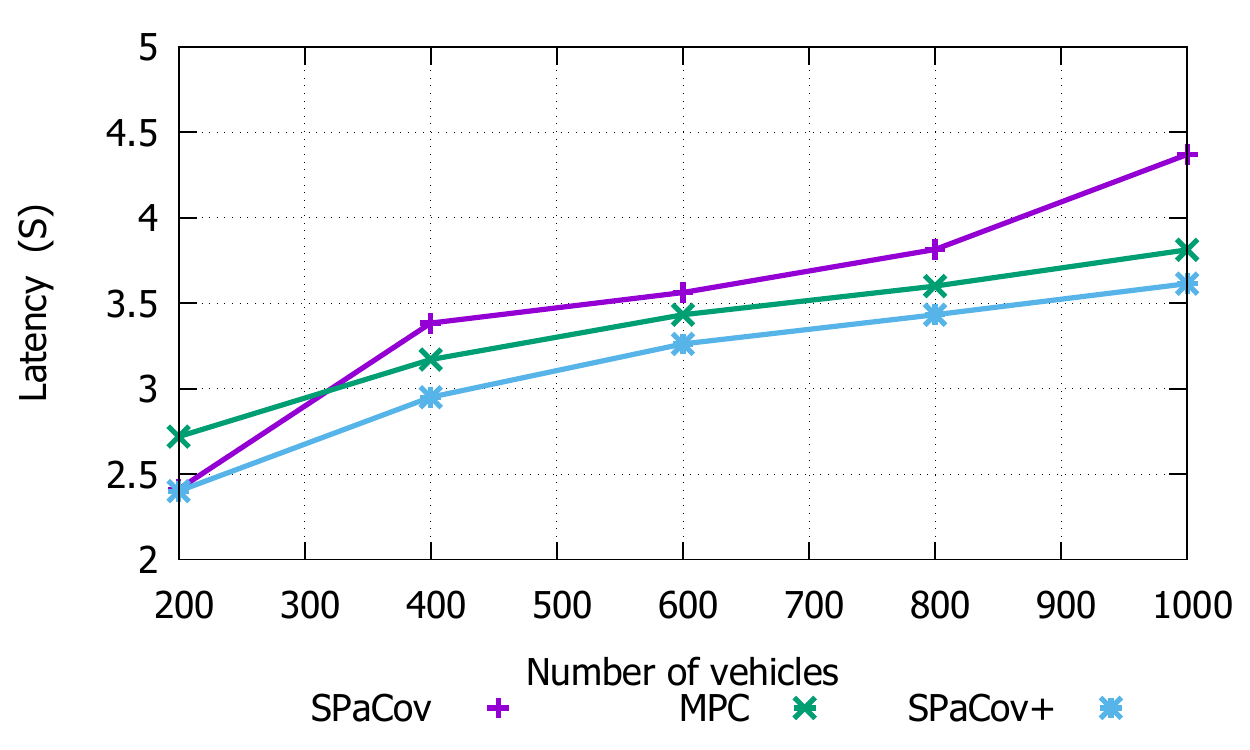}
\caption{Variation du temps de latence par rapport au nombre de véhicules pour la ville de Menzel Borguiba.)}
\label{Lat:MB}
\end{figure}
\begin{figure}[ht]
\centering
\includegraphics[scale=1]{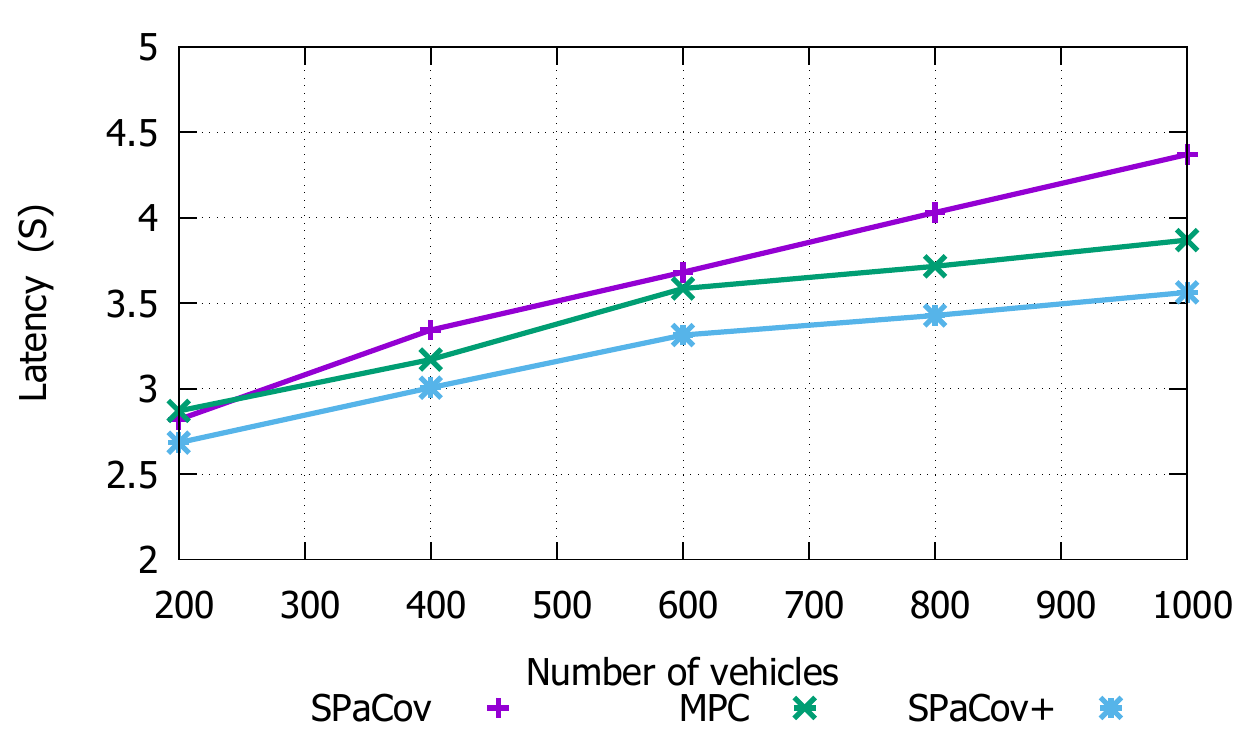}
\caption{Variation du temps de latence par rapport au nombre de véhicules pour la ville de Tozeur.)}
\label{Lat:T}
\end{figure}
\begin{figure}[!ht]
\centering
\includegraphics[scale=1]{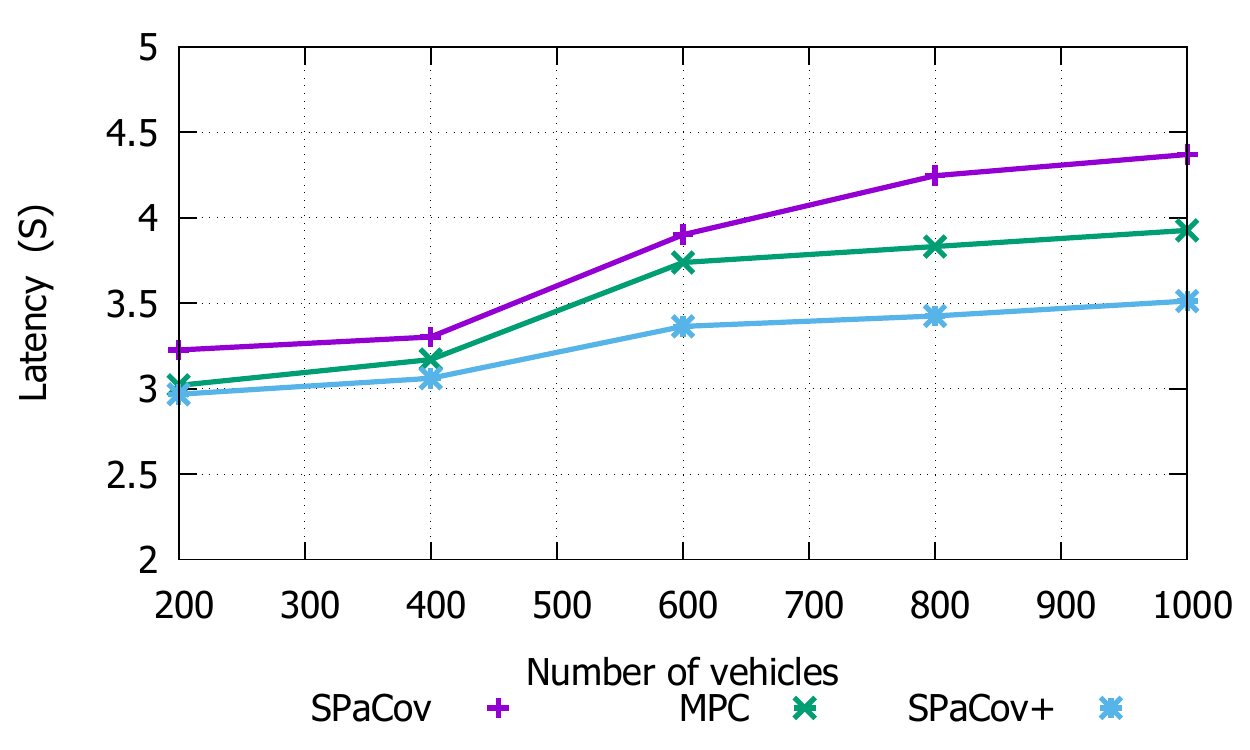}
\caption{Latence moyenne par rapport à la variation du nombre de véhicules pour la ville de Setif.}
\label{Lat:ST}
\label{Lat}
\end{figure}

D'une manière générale, les Figures \ref{OH:MB} et \ref{OH:T} et \ref{OH} montrent l'évolution de la surcharge des messages de nos méthodes par rapport à la méthode MPC en changeant le nombre de véhicules sous différentes topologies de route représentés par les villes Menzel borguiba, Tozeur et Setif. Nous observons que la surcharge augmente au fûr et à mesure que le nombre de véhicules augmente.\\
Plus précisément, la Figure \ref{OH:MB} montre que l'overhead augmente proportionnellement avec le nombre de véhicules. De plus, nous remarquons que  \textrm{\emph{SPaCov}} a le meilleur surcoût puisqu'il déploie moins de \textit{RSU} par rapport  \textrm{\emph{SPaCov+}} et  \textrm{\emph{MPC}}.  \textrm{\emph{SPaCov+}} signale un taux d'overhead plus faible que  \textrm{\emph{MPC}} malgré le nombre supplémentaire de \textit{RSU} déployées dans les zones peu fréquentées.\\
La Figure\ref{OH:T} (ville de Tozeur) montre que les trois méthodes considérées ont des valeurs d'overhead similaire (pour un nombre de véhicule réduit inférieur à $400$ véhicules). 

\begin{figure}[!ht]
\centering
\centering
\includegraphics[scale=1]{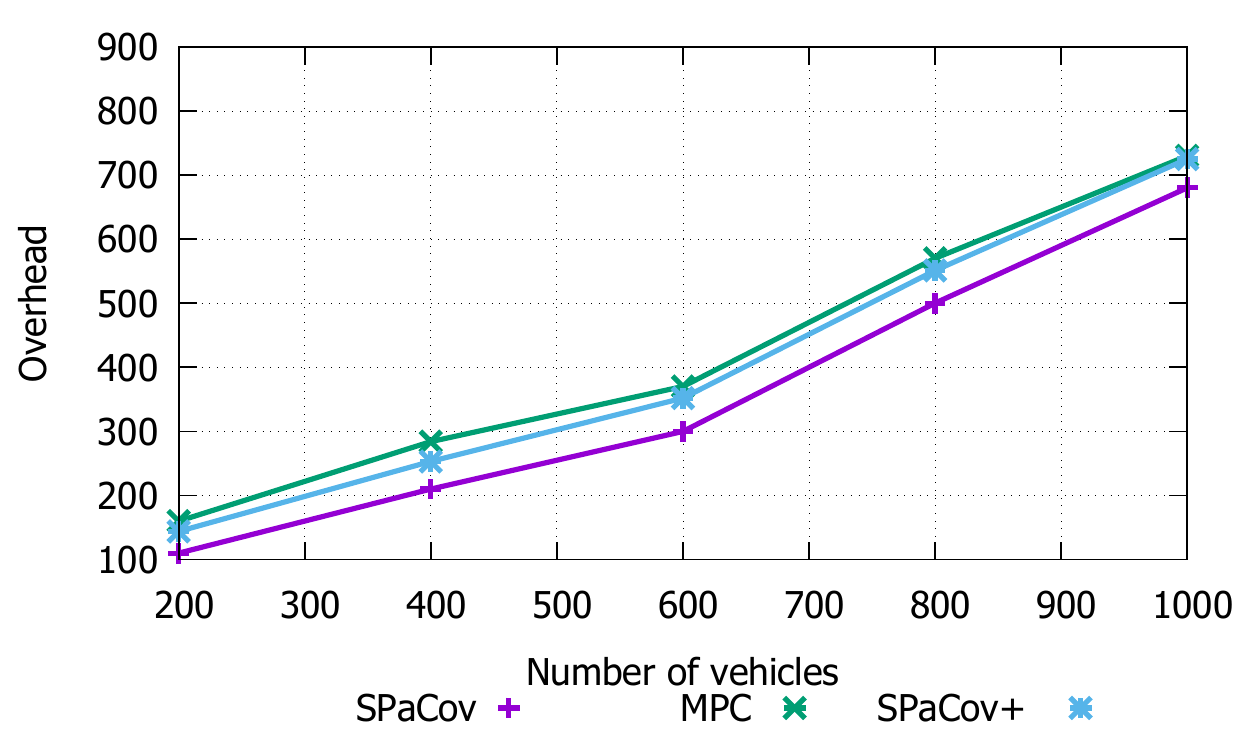}
\caption{Variation de l'verhead par rapport au nombre de véhicules pour la ville de Menzel Borguiba.}
\label{OH:MB}
\end{figure}
\begin{figure}[!ht]
\centering
\includegraphics[scale=1]{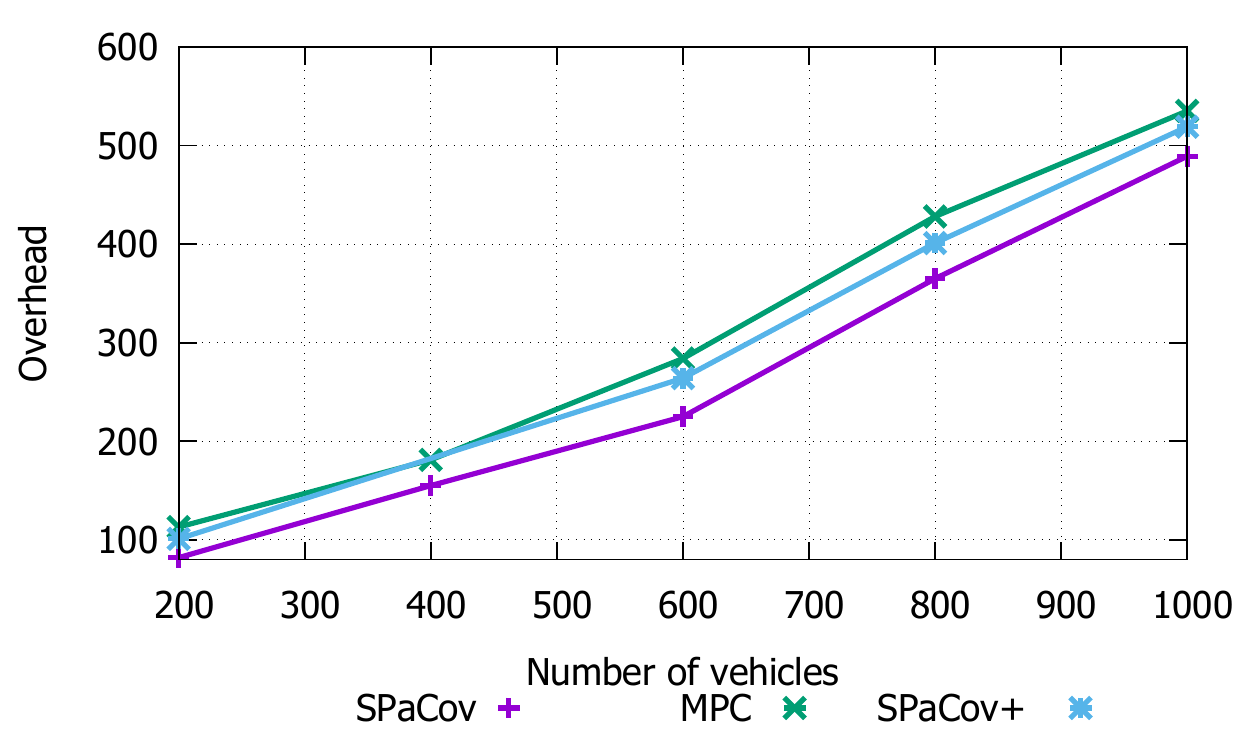}
\caption{Variation de l'overhead par rapport au nombre de véhicules pour la ville de Tozeur.}
\label{OH:T}
\end{figure}
Cependant la Figure \ref{OH} (ville de Setif) montre une différence d'overhead considérable entre les trois méthodes  par rapport au nombre de \textit{RSU} déployées par chacune.\\
\begin{figure}[!ht]
\centering
\includegraphics[scale=1]{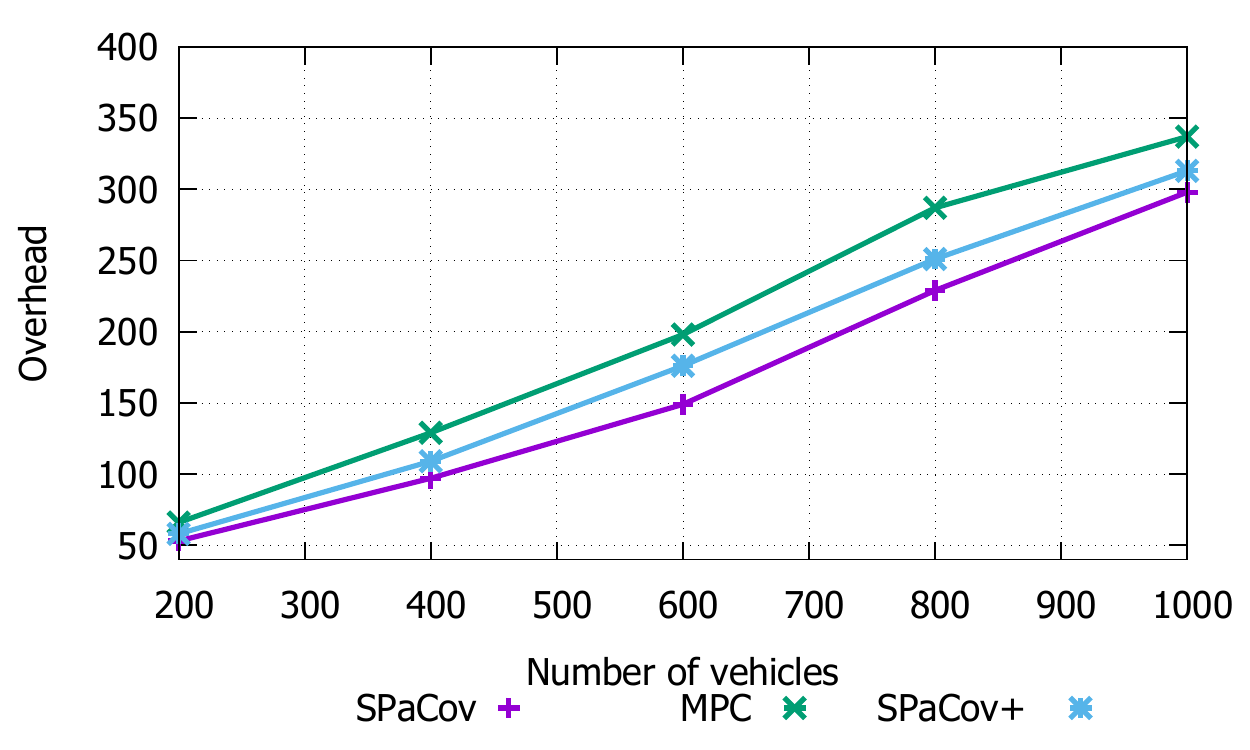}
\caption{Variation de l'overhead par rapport au nombre des véhicules pour la ville de Setif.}
\label{OH}
\end{figure}

Les Figures \ref{cost:MB}, \ref{cost:T} et \ref{cost:ST} représentent l'évolution du coût de déploiement, en terme du nombre de \textit{RSU}, pour les méthodes  \textrm{\emph{SPaCov}},  \textrm{\emph{SPaCov+}} et  \textrm{\emph{MPC}} dans les différentes villes étudiées qui représentent  différentes topologies routières.
\begin{figure}[!ht]
\centering
\centering
\includegraphics[scale=1]{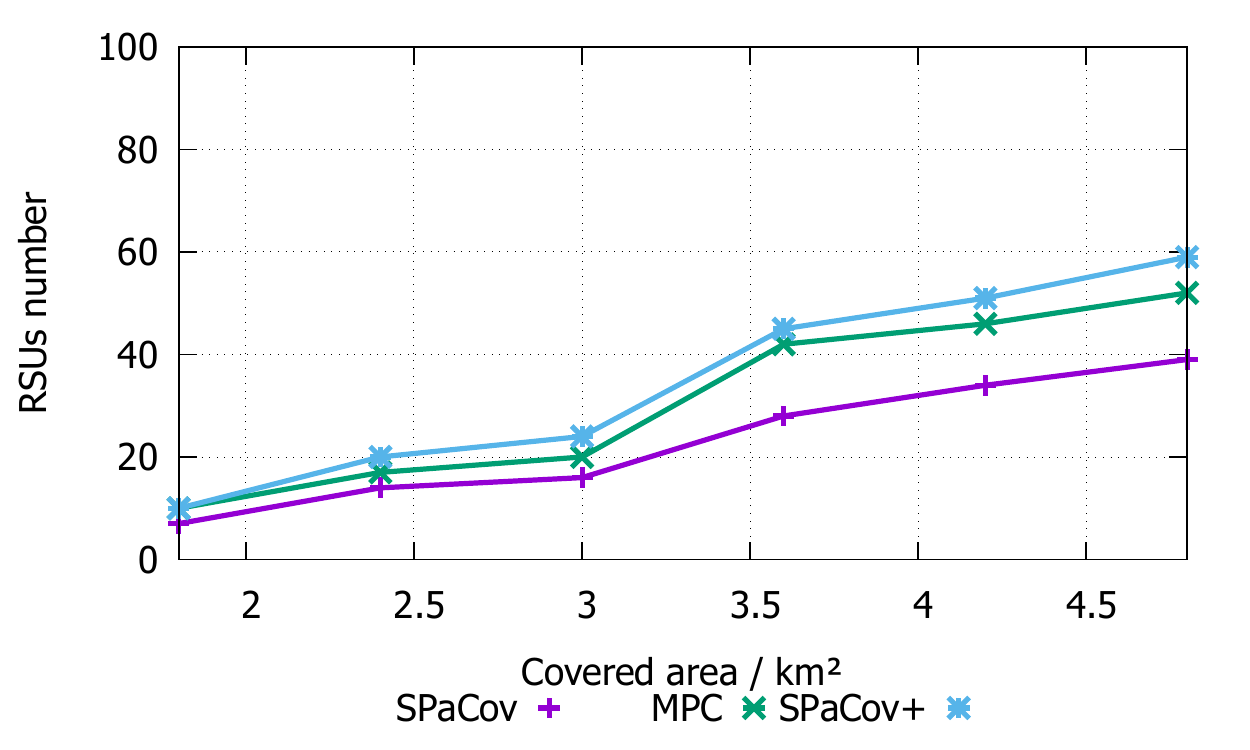}
\caption{Variation du nombre des \textit{RSU}s par rapport à la couverture spatiale pour la ville de Menzel Borguiba.}
\label{cost:MB}
\end{figure}
\begin{figure}[!ht]
\centering
\includegraphics[scale=1]{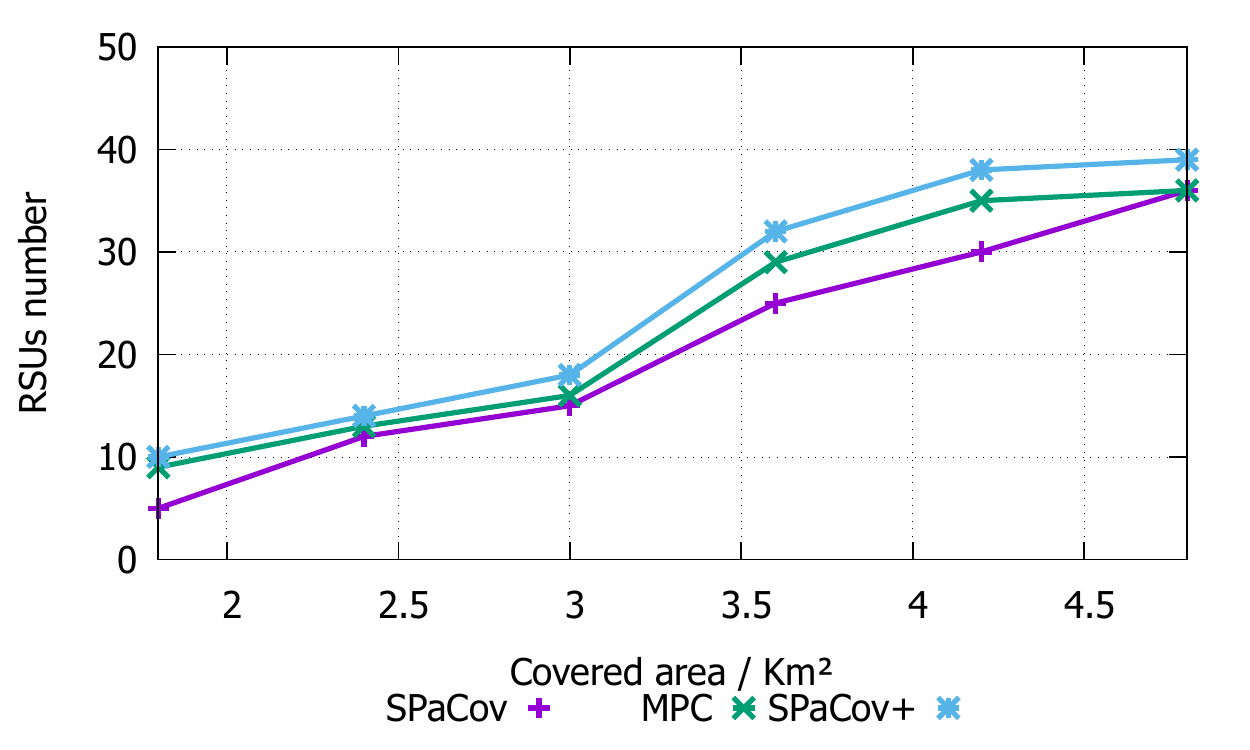}
\caption{Variation du nombre des \textit{RSU}s par rapport à la couverture spatiale de la ville de Tozeur.}
\label{cost:T}
\end{figure}
\begin{figure}[!ht]
\centering
\includegraphics[scale=1]{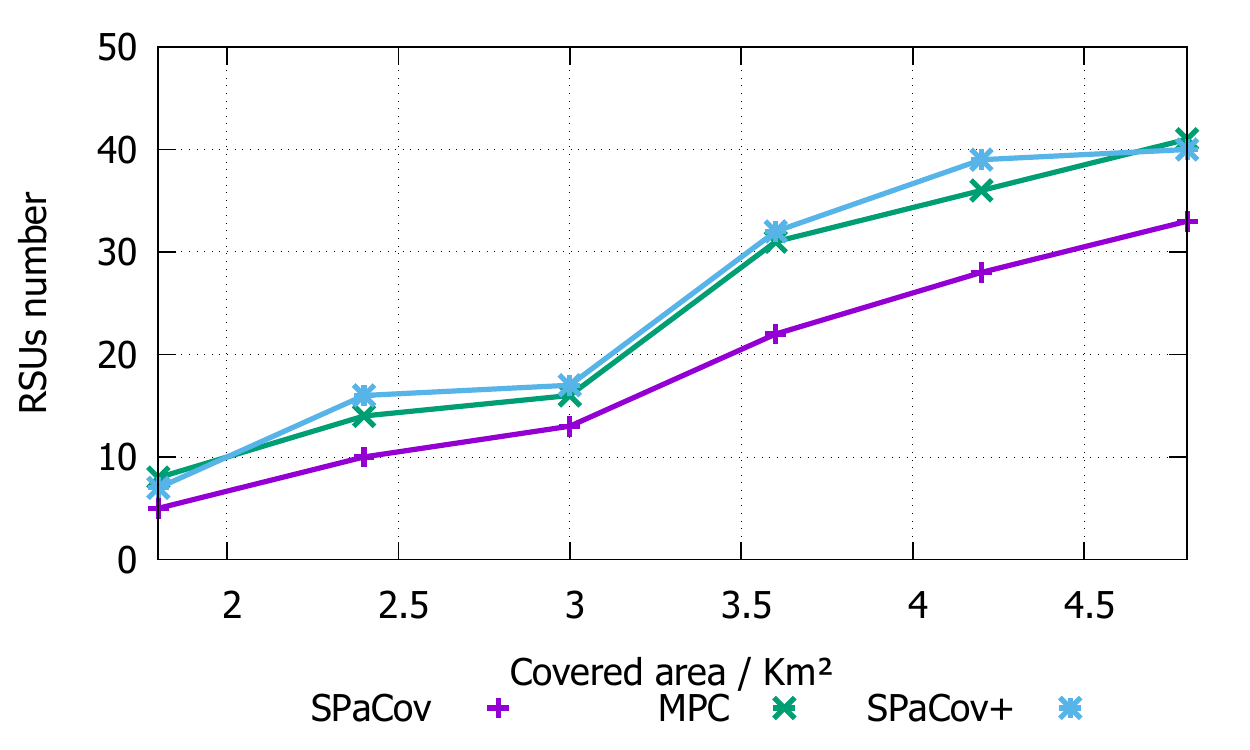}
\caption{Variation du nombre des \textit{RSU}s par rapport à la couverture spatiale pour la ville Setif.}
\label{cost:ST}
\end{figure}
Comme prévu, le nombre des \textit{RSU}s augmente au fur et à mesure que la surface de la zone couverte augmente pour toutes les cartes considérées. Les deux méthodes  \textrm{\emph{MPC}} et  \textrm{\emph{SPaCoV+}} ont presque le même nombre des \textit{RSU} sur toutes les topologies routières. Cependant,  \textrm{\emph{SPaCov}} génère moins de \textit{RSU} sur les trois types de carte utilisées. En effet, ces résultats peuvent s'expliquer par le fait que  \textrm{\emph{SPaCov}} ne recherche que les séquences maximales dans les modèles de mouvement.\\
\subsection{Résultats de la méthode  \textrm{\emph{HeSPiC}}}
Pour évaluer davantage les performances de la méthode  \textrm{\emph{HeSPiC}}, nous avons défini la zone de couverture à $4,8 km ^ {2}$ et la portée de communication à $0,3 m$  pour différentes topologies routières. Les paramètres de simulation de ce scénario sont détaillés dans le Tableau \ref{caractHespic}.\\

 \begin{table}[!ht]
  \centering

\begin{tabular}{|c|c|c|}
  \hline
 \multirow{3}{*}{Scénario} & Taille du message &  2312 Bytes \\  \cline{2-3}
            &Fréquence de Messages&0.5 Hz\\ \cline{2-3}
            &Slot time&13 us\\ \cline{2-3}
            &\#Runs&30 fois\\ \cline{2-3}
            &Portée de communication (Km) & 0.3 \\ \cline{2-3}
            &Nombre de véhicules & 400 \\ \cline{2-3}
  \hline \cline{1-3}
\end{tabular}
\caption{Les paramètres de simulation pour  \textrm{\emph{HeSPiC}}}\label{caractHespic}
\end{table}

\begin{figure}[!ht]
\centering
\includegraphics[scale=0.6]{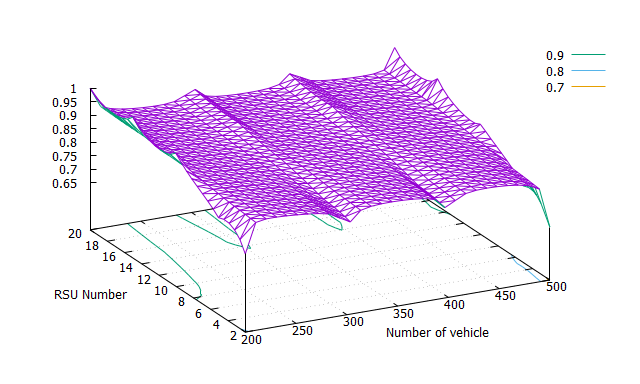}
\caption{Tozeur (Tunisia)}
\label{CR3D:T}
\end{figure}
Les Figures \ref{CR3D:T} et \ref{CR3D:ST} représentent la relation entre les trois facteurs: nombre de véhicules, le nombre des \textit{RSU}s et le taux de couverture. Comme prévu, le maximum, pour toutes les courbes et pour toutes les topologies routières, est atteint lorsque le nombre des \textit{RSU}s est assez élevé par rapport au nombre de véhicules. De plus, la comparaison entre les différentes topologies routières montre une corrélation entre la complexité de la carte et le taux de couverture. En effet, le taux de couverture maximal devient rapidement atteint lorsque la carte devient moins complexe (voir la partie basse de la courbe de Menzel Bourguiba Figure \ref{CR3D:MB}). \\
\begin{figure}[!ht]
\centering
\centering
\includegraphics[scale=0.6]{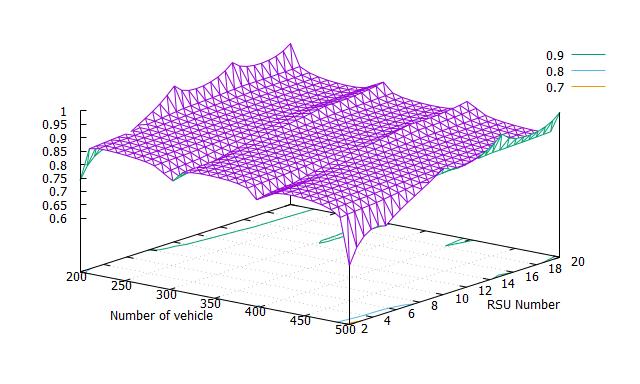}
\caption{Menzel Borguiba (Tunisia)}
\label{CR3D:MB}
\end{figure}
\begin{figure}[!ht]
\centering
\includegraphics[scale=0.6]{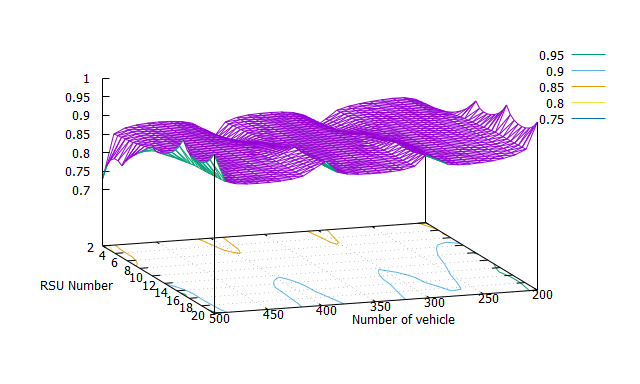}
\caption{Setif (Algeria)}
\label{CR3D:ST}
\end{figure}

\begin{figure}[!ht]
\centering
\centering
\includegraphics[scale=1]{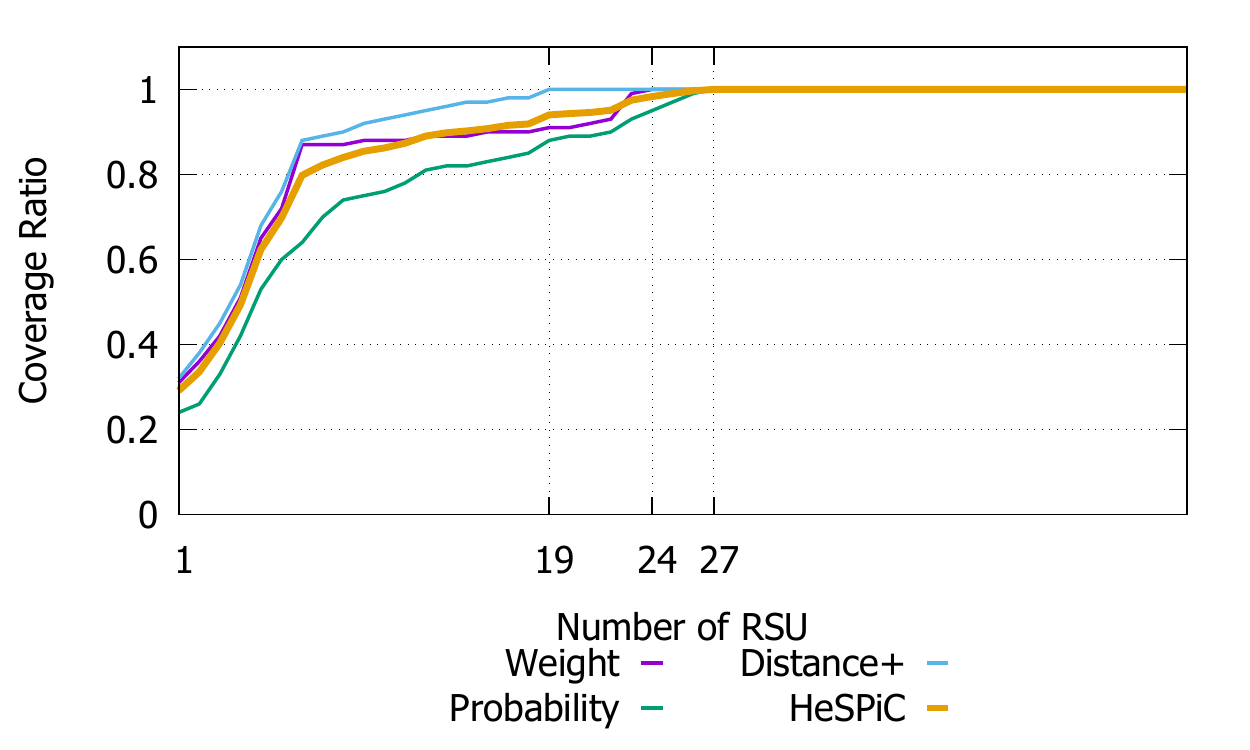}
\caption{L'évolution du taux de couverture de la méthode HeSPiC par rapport à la variation du nombre des \textit{RSU}s pour la ville de Menzel Borguiba.}
\label{RSUCOV:MB}
\end{figure}
\begin{figure}[!ht]
\centering
\includegraphics[scale=1]{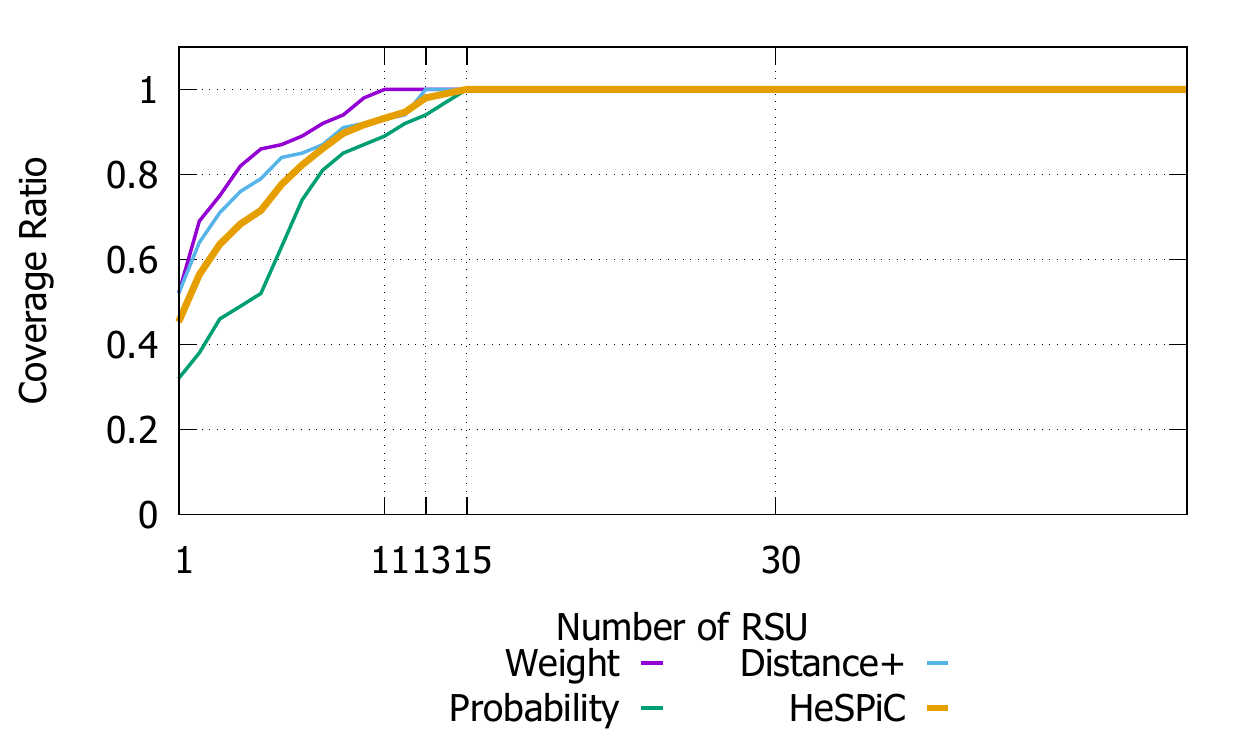}
\caption{L'évolution du taux de couverture de la méthode HeSPiC par rapport à la variation du nombre des \textit{RSU}s de la ville de Setif.}
\label{RSUCOV:ST}
\end{figure}

Les Figures \ref{RSUCOV:MB}, \ref{RSUCOV:ST} et \ref{RSUCOV:T} représentent l'évolution du taux de la couverture par rapport à la variation du nombre de RSU sous différentes topologies de carte. Comme mentionné ci-dessus, le taux de couverture augmente dans la mesure où le nombre des RSUs augmente. En effet, la méthode HeSPiC atteint le maximum (Ratio de Couverture = $1$) entre $19$ et $27$ RSUs sur la carte de Menzel Bourguiba (voir Figure \ref{RSUCOV:MB}). En revanche, il ne lui faut que $13$ RSUs pour avoir la même couverture maximale sur la carte de Sétif (Figure \ref{RSUCOV:ST}). En outre, nous observons que chacune des métriques utilisées dans HeSPiC (la probabilité, le poids et la distance) permet d'obtenir un meilleur rapport de couverture sur une carte particulière en fonction de sa complexité, par exemple, la distance est la meilleure métrique qui permet d'atteindre une couverture maximale avec un minimum de RSU pour la carte de Menzel Bourguiba. Cependant, le poids est le meilleur pour la carte de Sétif et la métrique de probabilité est la meilleur dans le cas de Tozeur (Figure \ref{RSUCOV:T}).\\
\begin{figure}[!ht]
\centering
\includegraphics[scale=1]{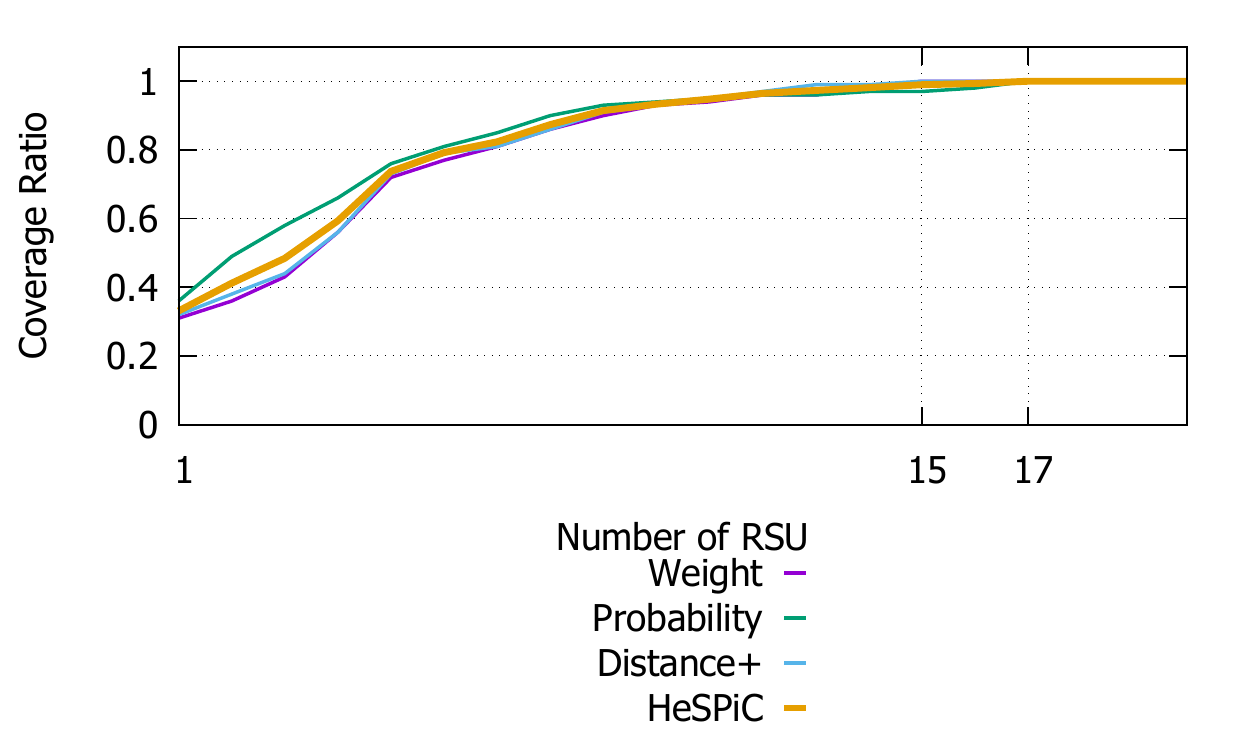}
\caption{Ville de Tozeur}
\label{RSUCOV:T}
\caption{L'évolution du taux de couverture de la méthode  \textrm{\emph{HeSPiC}} par rapport à la variation du nombre de RSUs pour la ville de Tozeur.}
\end{figure}

La Figure \ref{Crange} représente l'évolution du taux de couverture par rapport à la variation de la portée de communication sous différentes topologies routières. En effet, pour un nombre de véhicules égal à $400$ nous obtenons un nombre des \textit{RSU}s égal à $10$. Comme prévu, le taux de couverture augmente à mesure que la portée de communication augmente. En effet, une portée de communication relativement élevée permet une large couverture de la zone, ce qui implique un taux de couverture élevé.\\
\begin{figure}[!ht]
  \begin{center}
  \includegraphics[scale=1]{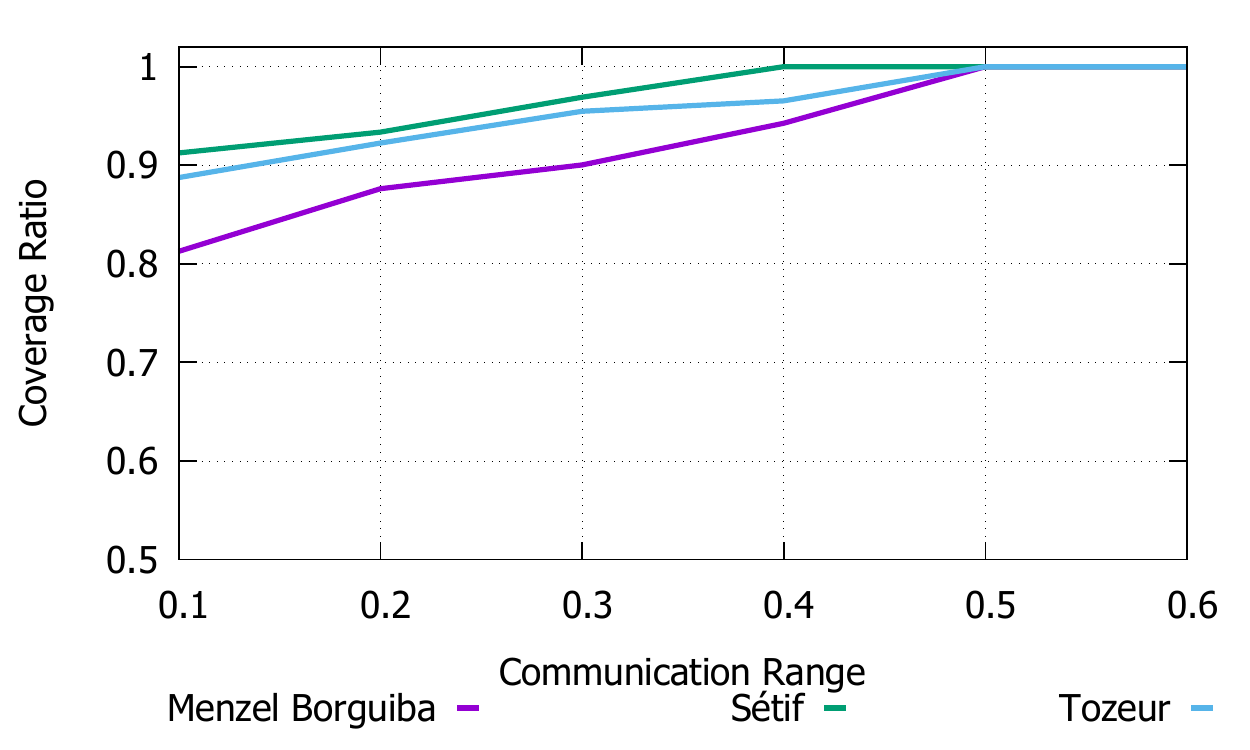}
  \caption{L'évolution du taux de couverture par rapport à la variation de la portée de communication sur différentes topologies routières en utilisant la méthode  \textrm{\emph{HeSPiC}} (vehicles = $400$, \textit{RSU}s et $K = 10$).} 
  \label{Crange}
\end{center}
\end{figure}
\subsection{Résultats de la méthode \textrm{\emph{MIP}}}
La Figure \ref{cost} montre l'évolution du coût de déploiement de  \textrm{\emph{MIP}} par rapport à  \textrm{\emph{SPaCov+}} \cite{Bch19} et MPC \cite{B1} sous différentes superficies des zones couvertes. En effet, le coût de déploiement augmente proportionnellement à la surface de couverte. En effet, plus de \textit{RSU} sont nécessaires pour couvrir une vaste zone. De plus, nous observons que  \textrm{\emph{MIP}} a un faible coût de déploiement par rapport à  \textrm{\emph{MPC}} et  \textrm{\emph{SPaCov+}} pour toute variation de la zone couverte. En effet, cela diminue le coût global de déploiement par rapport à la méthode  \textrm{\emph{MPC}} d'environ $13$\%, ceci est dû au fait que notre méthode  \textrm{\emph{MIP}} considère ($i$) l'utilité de chaque séquence de croisement en privilégiant les petites séquences, ($ii$) l'ordre des croisements lors de l'extraction des modèles de mobilité, ce qui conduit sans aucun doute à réduire le nombre total de patterns extraits. Par conséquent, moins de \textit{RSU} sont nécessaires pour couvrir les modèles extraits, ce qui entraîne une diminution du coût de déploiement.\\

\begin{figure}[ht]
\begin{center}
\includegraphics[scale=1]{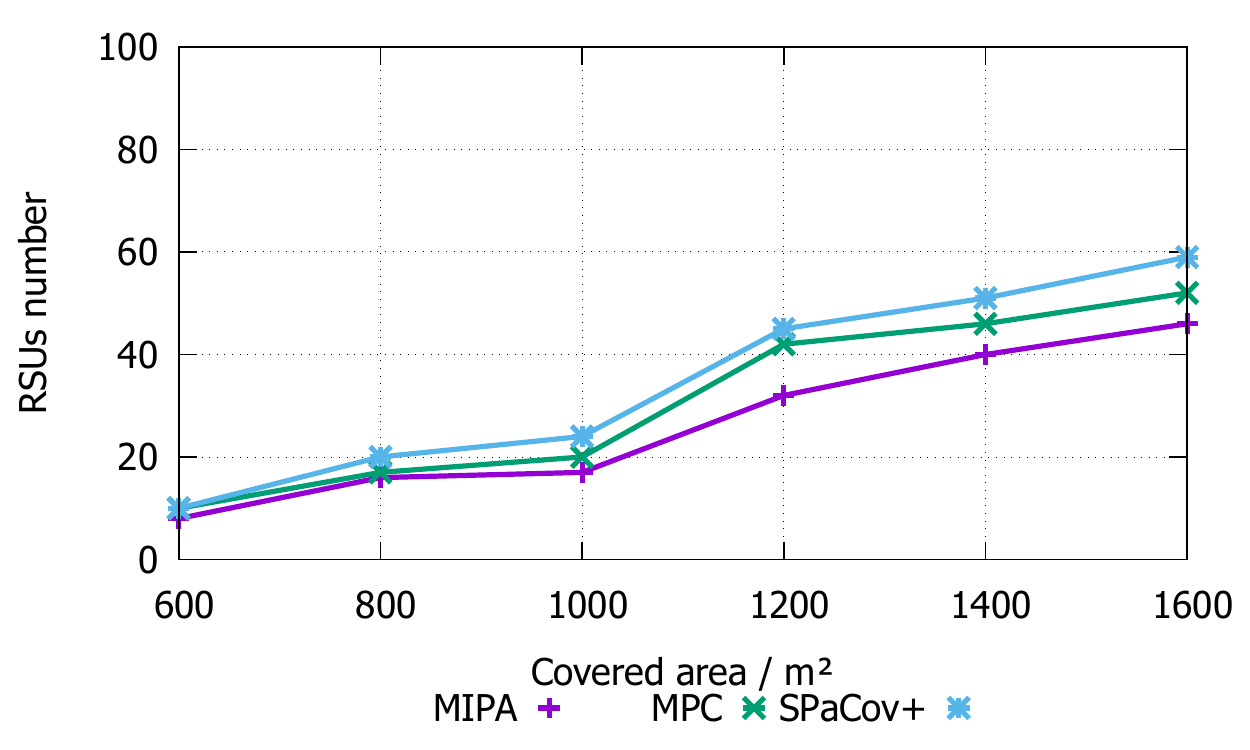}
\caption{Variation du nombre de \textit{RSU}s par rapport à la zone à couvrir pour la ville de Menzel Borguiba.}
\label{cost}
\end{center}
\end{figure}
\begin{figure}[ht]
\begin{center}
\includegraphics[scale=1]{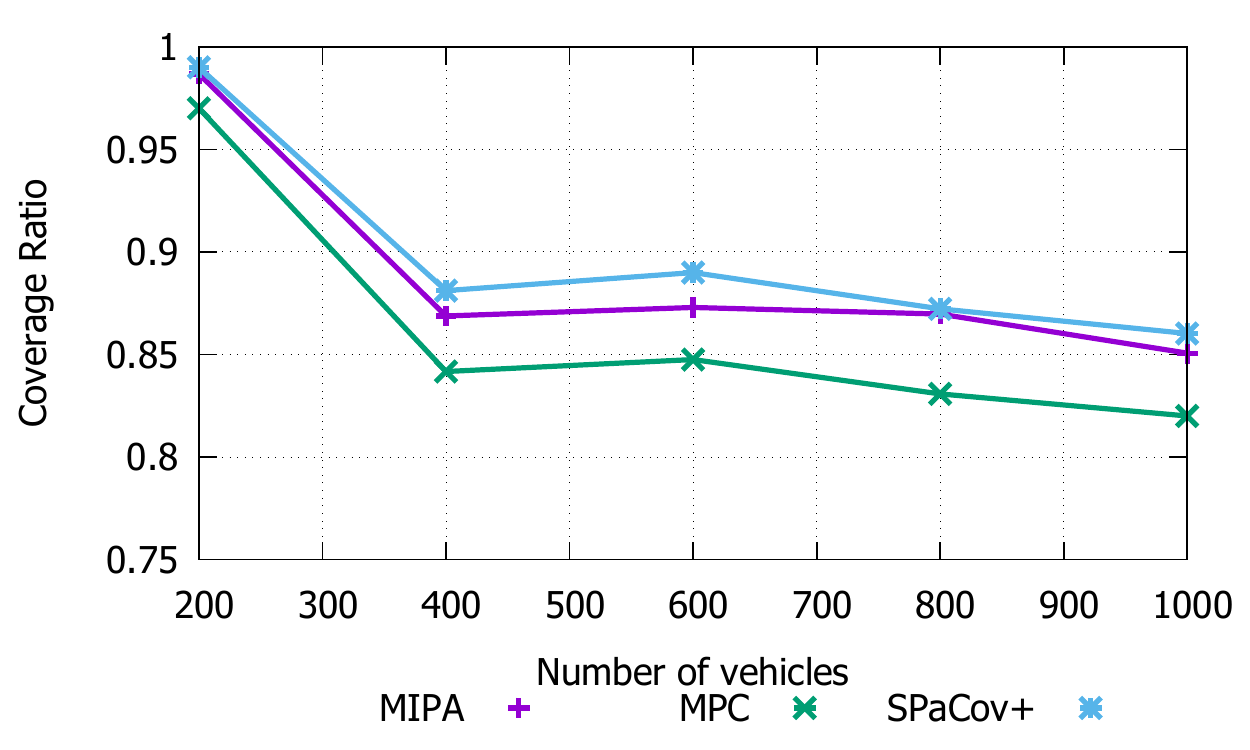}
\caption{Variation du taux de couverture par rapport au nombre de véhicules pour la ville de Menzel Borguiba.}
\label{cost:MIP}
\end{center}
\end{figure}
\begin{figure}[!ht]
\begin{center}
\includegraphics[scale=1]{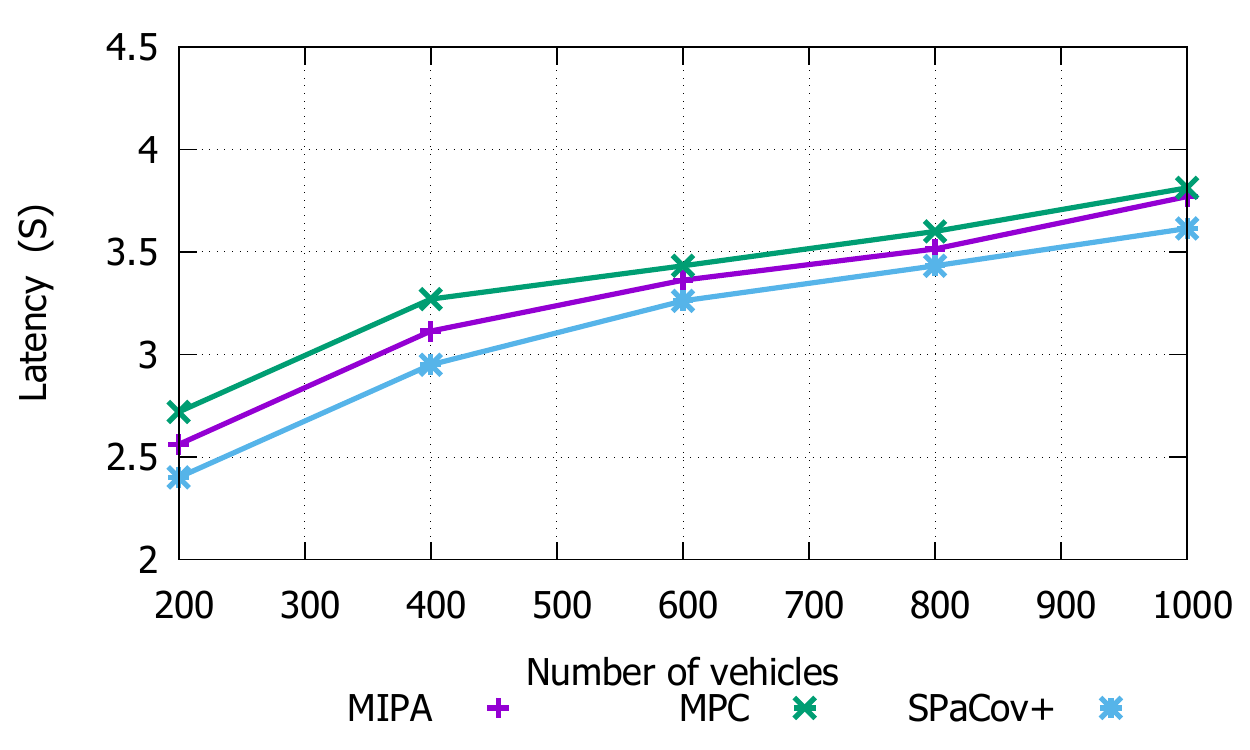}
\caption{Variation du temps de latence par rapport au nombre des véhicules pour la ville de Menzel Borguiba.}
\label{latency:MIP}
\end{center}
\end{figure}

\begin{figure}[!ht]
\begin{center}
\includegraphics[scale=1]{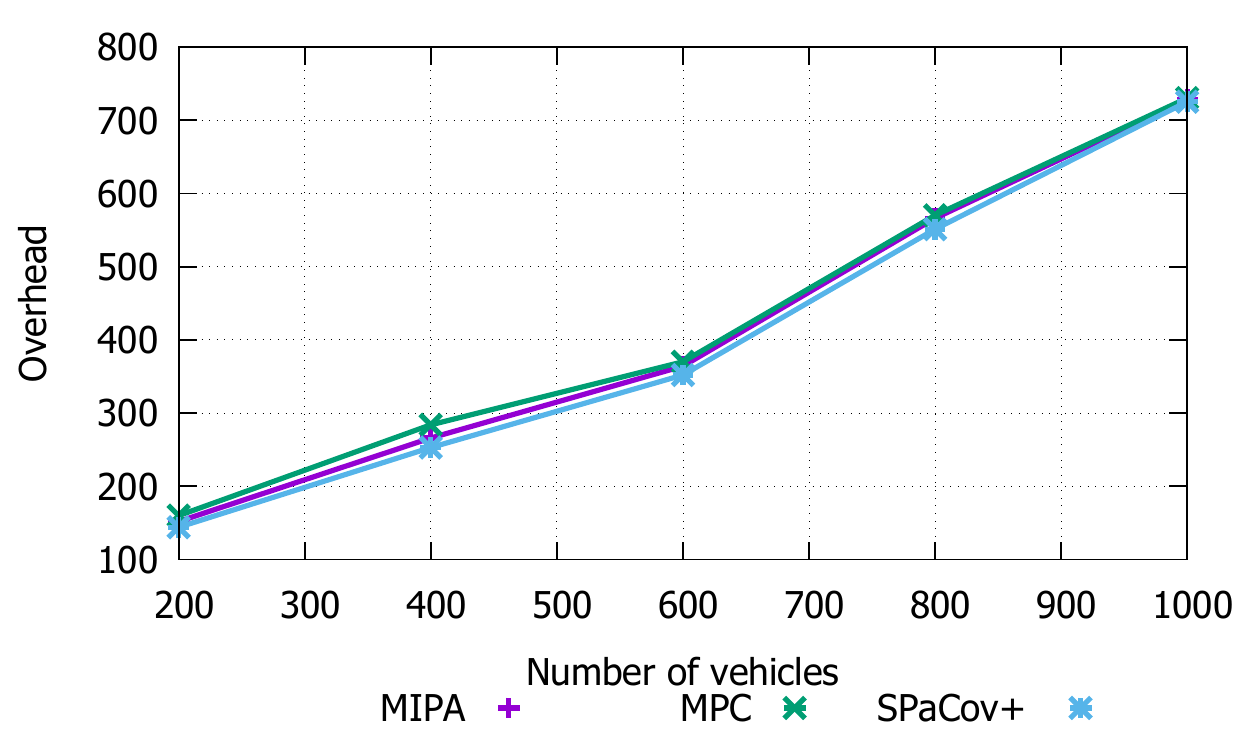}
\caption{Variation de l'overhead par rapport au nombre des véhicules en circulation pour la ville de Menzel Borguiba.}
\label{overhead:MIP}
\end{center}
\end{figure}

La Figure \ref{cost:MIP}, montre l'évolution du taux de couverture par rapport au nombre de véhicules en mouvement. En effet, il est normal que le taux de couverture soit inversement proportionnel à la densité du trafic, car l'augmentation du nombre de véhicules dans le réseau entraînera une augmentation du taux de paquets perdus, de sorte que de nombreux véhicules ne recevront pas les messages envoyés par les \textit{RSU}s. Nous observons que les taux de couverture de deux approches de couverture  \textrm{\emph{MIP}} et  \textrm{\emph{SPaCov+}} sont plus ou moins les mêmes. Cependant, le ratio de couverture du MPC est un peu inférieur à celui des autres méthodes. Par conséquent, cela peut s'expliquer par le fait que  \textrm{\emph{MPC}} couvre un ensemble de croisement spatialement non connectées (ne prend pas en compte l'intersection entre les croisements).\\

La Figure \ref{latency:MIP}, représente l'évolution de temps de latence de la méthode  \textrm{\emph{MIP}} par rapport à  \textrm{\emph{SPaCov+}} et  \textrm{\emph{MPC}} sous différentes densités de véhicules. Nous constatons que la latence des trois approches diminue à chaque fois que le nombre de véhicules augmente. En effet, l'augmentation du nombre de véhicules surcharge dramatiquement les \textit{RSU}s, ce qui entraîne une augmentation de la latence. De plus, la Figure \ref{latency:MIP} montre que les trois approches ont presque une même latence pour toutes les variations de densité de véhicules. Les valeurs de latence de  \textrm{\emph{MIP}} se situent entre les valeurs de deux autres méthodes sous différentes densités de véhicules.\\

La Figure \ref{overhead:MIP}, montre l'évolution de la surcharge de la méthode  \textrm{\emph{MIP}} par rapport à  \textrm{\emph{SPaCov+}} et  \textrm{\emph{MPC}} sous différentes densités de villes. Nous remarquons que l'overhead augmente au fur et à mesure que le nombre de véhicules augmente. En effet, plus le nombre de véhicules est élevé, plus les communications réseaux et l'overhead sont élevés. En outre, la Figure \ref{overhead:MIP} montre que toutes les approches ont presque un même overhead pour toutes les variations de densité de route. \\
Pour résumer, les résultats de la simulation mettent en évidence que notre méthode  \textrm{\emph{MIP}} offre un meilleur compromis entre le coût, le taux de couverture, l'overhead et le temps de latence par rapport aux autres méthodes à savoir  \textrm{\emph{SPaCov+}} et  \textrm{\emph{MPC}}. En effet, la méthode  \textrm{\emph{MIP}} permet de diminuer le coût global de déploiement des \textit{RSU}s et la surcharge par rapport aux autres méthodes, d'environ $13$\% et $2$\% respectivement, tout en conservant les mêmes résultats en termes de taux de la couverture et du temps de latence.

\section{Conclusion}
Dans ce chapitre, nous avons mené une étude expérimentale ayant comme objectif la quantification des unités de bord, ainsi que le taux de la couverture dans différents scénarios de test et dans plusieurs topologies routières de différente villes.\\
À la lumière de cette étude, nous avons constaté que le nombre des \textit{RSU}s déployés par nos méthodes, sont dans la plupart des cas, de taille inférieur à celle des méthodes proposées dans la littérature. De plus, le fait que le taux de la couverture, assuré par nos méthodes, soit respectable par rapport au nombre des \textit{RSU}s utilisés. Nous constatons que le nombre de \textit{RSU}s par rapport aux taux de la couverture est meilleur aux autres méthodes proposées dans la littérature.

\addcontentsline{toc}{chapter}{Conclusion générale}

\chapter*{Conclusion générale et Perspectives}
Dans ce travail, nous avons présenté les notions de bases utilisées par nos approches à savoir celles qui concernent les patterns séquentiels et la théorie des graphes ainsi que quelques fondements mathématiques concernant les lois probabilistes. Par la suite, nous avons fait une étude exhaustive des méthodes travaillant sur les systèmes de transport intelligent STI, nous avons également mené une étude critique des méthodes qui s'intéressent aux déploiements des \textit{RSU}s dans les réseaux VANET afin de déceler les limites et les avantages de chacune d'elles. Cependant, nous avons constaté qu'aucune de ces méthodes n'exploitent les modèles de mobilité des véhicules sous forme des patterns séquentiels. Tout de même, nous avons présenté les principales méthodes d'extraction des patterns séquentiels ainsi que les principales méthodes de couverture des hypergraphes.\newline

Afin de pallier ces lacunes, nous avons introduit trois approches de déploiement des \textit{RSU}s dans un réseau véhiculaire VANET la première, nommé  \textrm{\emph{SPaCov/SPaCov+}}, basée sur deux points important ($i$) les patterns séquentiels fréquents/rares et ($ii$) la couverture de ces patterns en utilisant les traverses minimales des hypergraphes. La deuxième approche, nommé  \textrm{\emph{HeSPiC}}, consiste à placer un nombre de \textit{RSU} fixé d'avance par l'utilisateur (méthode obéit à une contrainte budgétaire), l'approche utilise une fonction mathématique de classement afin d'ordonner les croisements en termes d'importance. La dernière approche, nommé  \textrm{\emph{MIP}}, consiste à réduire le nombre des transactions de départ et se restreindre à une base de données des mouvements des véhicules condensée qui ne contient que des transactions, évaluer par la méthode, comme très utiles. En effet, la fonction d'utilité des transactions est la partie originale de cette dernière approche puisqu’une preuve mathématique montre que les séquences dite utile par l'approche  \textrm{\emph{MIP}}, ne sont ni maximaux, ni minimaux et ni fermés en terme d'inclusion.\\

Afin de tester nos méthodes, nous avons introduit trois nouveaux algorithmes liés aux trois approches proposées, à savoir  \textrm{\emph{SPaCov}},  \textrm{\emph{HeSPic}} et  \textrm{\emph{MIP}}. Tout de même, nous avons prouvé la correction et la complétude de trois algorithmes. De plus, nous avons évalué la complexité au pire des cas, pour les trois algorithmes. Afin de s'assurer des résultats et la cohérence entre la partie théorique et la partie empirique, nous avons mené une étude expérimentale dans laquelle nous avons évalué nos méthodes avec plusieurs scénarios de test et sur différentes cartes routières. Enfin, Nous avons évalué l'aspect budgétaire en corrélation avec la bonne couverture de réseau.\\

À la lumière de cette étude, nous avons constaté que dans la plupart des scénarios et sur différentes cartes routières, la taille de \textit{RSU}s sont plus réduits que celles des méthodes de la littérature tout en gardant une bonne couverture réseau entre les véhicules. Nous avons prouvé expérimentalement que nos objectifs définis d'avance ont été atteints, à savoir pallier les problèmes liés à la relation entre la réduction du nombre des \textit{RSU}s sans trop perdre coté couverture. Ces problèmes peuvent se résumer dans les points suivants:
\begin{itemize}
  \item Un nombre très important des véhicules dans des cartes très complexes implique une taille importante de données de véhicules en mouvement.
  \item Le coût budgétaire de déploiement des \textit{RSU}s peut atteindre un montant assez élevé, ce qui rend le déploiement aveugle de \textit{RSU}s une solution obsolète.
\end{itemize}
Ces problèmes ont été résolus, puisque nous avons offert un taux de couverture plus élevé, que la plus part des méthodes de la littérature.
\\
La couverture dans les réseaux véhiculaires constitue une piste de recherche intéressante. Au meilleur de notre connaissance, nos méthodes sont les premières qui ont utilisé les patterns séquentiels combiné avec les hypergraphes et les traverses minimales dans les réseaux VANET. Cet axe de recherche est plein de promesse puisque plusieurs pistes ne sont pas explorées. Nous citons comme perspectives de recherche les point suivants:
\begin{enumerate}
    \item Élargir les types des \textit{RSU}s utilisés et tester d'autres scénarios de simulation à l'exemple des \textit{RSU}s mobiles (utilisation des UAV). Puisqu'ils peuvent donner plus de couverture avec un même nombre de \textit{RSU} par rapport à des approches classiques.
    \item Dans le cas des \textit{RSU}s mobiles, le problème de mouvements des \textit{RSU}s s'ajoute au problème initial, qui est le mouvement des véhicules. Afin de remédier à ce problème, il faut penser à mettre en place un système hybride, qui selon le modèle de mobilité des véhicules, adapte le système de déploiement des \textit{RSU}s.  
    \item Introduire et exploiter les informations multi dimensionnels, afin de mieux prévoir les mouvements futurs des véhicules. Par exemple, le final de la coupe de la Tunisie, qui se déroule au stade olympique de Rades entre 17h et 19h, peut provoquer un mouvement des véhicules anormale, aux alentours du stade de Rades avant et après quelques heures de déroulement du match, par rapport à un dimanche sans aucun évènement dans le même endroit et pendant la même plage horaire. Cependant, avec un système classique, qui ne considère pas les évènements, la prédictions des mouvements des véhicules n'arrive jamais à prévoir un embouteillage aux alentours du lieu d'un évènement, par contre avec un système multi-dimensionnels avec interprétation des événements la prédiction des mouvements futurs des véhicules sera plus adéquate avec plus de précision.   
\end{enumerate}

\nocite{ref1,ref2,ref3,ref4,ref5,ref6,ref7,ref8,ref9,ref10,ref11,ref12,ref13,ref14,ref15,ref16,ref17,ref18,ref19,ref20,ref21,ref22,ref23,ref24}

\bibliography{biblio}
\bibliographystyle{plain}
\end{document}